\documentclass[12pt]{article}
\usepackage{amsmath,amssymb,amsthm}
\usepackage{natbib}
\usepackage{graphicx}
\usepackage[colorlinks,citecolor=blue,urlcolor=blue]{hyperref}
\usepackage{multirow}
\usepackage{dsfont}
\usepackage{epstopdf}
\usepackage{booktabs}
\usepackage{textcomp}
\usepackage{ulem}
\usepackage{color}
\usepackage{threeparttable}
\usepackage{comment}
\usepackage{float}
\usepackage{indentfirst}
\usepackage{arydshln}
\usepackage{cleveref}
\usepackage{soul}
\usepackage{scalerel}
\usepackage{lipsum}

\newtheorem{assumption}{Assumption}
\newtheorem{theorem}{Theorem}
\newtheorem{proposition}{Proposition}
\newtheorem{lemma}{Lemma}
\newtheorem{remark}{Remark}

\newtheorem{corollary}{Corollary}

\newcommand{\bs}{\boldsymbol}

\def\Var{\textnormal{Var}}
\def\E{\operatorname{E}}
\def\bsw{\backslash \bs{W}}
\def\Cov{\textnormal{Cov}}

\def\op{o_{\mathbb{P}}}
\def\Op{O_{\mathbb{P}}}
\def\opnorm{\textnormal{op}}

\newcommand{\ind}{\scalebox{0.6}{\textrm{IND}}}
\newcommand{\dir}{\scalebox{0.6}{\textrm{DIR}}}
\newcommand{\tot}{\scalebox{0.6}{\textrm{TOT}}}
\newcommand{\ev}{\scalebox{0.6}{\textrm{EV}}}
\def\rhon{\rho_{\scaleobj{0.8}{N}}}
\def\Lambdan{\Lambda_{\scaleobj{0.8}{N}}}
\def\Deltan{\Delta_{\scaleobj{0.8}{N}}}
\def\tldDeltan{\tilde{\Delta}_{\scaleobj{0.8}{N}}}
\newcommand{\ora}{\textrm{ora}}

\DeclareMathOperator*{\argmin}{arg\,min}
\def\diag{\textrm{diag}}
\def\sumi{\sum_{i=1}^{N}}
\def\sumj {\sum_{j=1}^{N}}
\def\sumk{\sum_{k=1}^{N}}
\def\maxi{\max_{i=1,\ldots,N}}

\def\mini{\min_{i=1,\ldots,N}}
\def\tr{\operatorname{tr}}
\def\oprtnorm{\operatorname{op}}

\def\sd{\sigma_{\dir}}
\def\si{\sigma_{\ind}}

\def\spacingset#1{\renewcommand{\baselinestretch}
	{#1}\small\normalsize} \spacingset{1}

\begin{document}
	
	\def\spacingset#1{\renewcommand{\baselinestretch}
		{#1}\small\normalsize} \spacingset{1}
		
	\title{\bf Estimation and inference of average treatment effects under heterogeneous additive treatment effect model}
	\author{Xin Lu\thanks{These authors contributed equally to this work.}, Hongzi Li\footnotemark[1]~ and Hanzhong Liu\thanks{
			Corresponding author: lhz2016@tsinghua.edu.cn. Dr. Liu was supported by \textit{the National Natural Science Foundation of China (12071242).}}\hspace{.2cm}\\
		Center for Statistical Science, Department of Industrial Engineering,\\ Tsinghua University, Beijing, China}
	\date{}
	\maketitle

\bigskip
\begin{abstract}
Randomized experiments are the gold standard for estimating treatment effects, yet network interference challenges the validity of traditional estimators by violating the stable unit treatment value assumption and introducing bias. While cluster randomized experiments mitigate this bias, they encounter limitations in handling network complexity and fail to distinguish between direct and indirect effects. To address these challenges, we develop a design-based asymptotic theory for the existing Horvitz--Thompson estimators of the direct, indirect, and global average treatment effects under Bernoulli trials. We assume the heterogeneous additive treatment effect model with a hidden network that drives interference. Observing that these estimators are inconsistent in dense networks, we introduce novel eigenvector-based regression adjustment estimators to ensure consistency. We establish the asymptotic normality of the proposed estimators and provide conservative variance estimators under the design-based inference framework, offering robust conclusions independent of the underlying stochastic processes of the network and model parameters. Our method's adaptability is demonstrated across various interference structures, including partial interference and local interference in a two-sided marketplace. Numerical studies further illustrate the efficacy of the proposed estimators, offering practical insights into handling network interference.\\

\end{abstract}


\newpage
\spacingset{1.9}

\section{Introduction}

There has been a growing interest in estimating the global average treatment effect (GATE) within randomized experiments \citep{yu2022estimating, harshaw2023design}. The GATE represents the difference in the average outcomes between treating all units versus none in the target population. Under the Stable Unit Treatment Value Assumption (SUTVA) \citep{Rubin1980}, the difference in sample means of outcomes under treatment and control is an unbiased estimator of the GATE. However, in many applications, experimental units are interconnected within networks, such as social networks of family, community, or friendship \citep{bond201261, aral2017exercise, jiang2023auto}. In such scenarios, empirical evidence has demonstrated the existence of interference, i.e., the treatment assignment of one unit affects the outcomes of others \citep{bond201261,cai2015social,holtz2020interdependence,leung2022causal}. This interference violates SUTVA and may introduce bias to classic Horvitz--Thompson estimator of the GATE. Cluster randomized experiments offer a means to mitigate this bias under appropriate interference assumptions \citep{viviano2023causal, chen2024optimized, holtz2024reducing}.

However, cluster randomized experiments encounter limitations in addressing network interference. Firstly, they necessitate obtaining high-quality graph cuts of the network before experimentation, which is often hindered by the complex structure of real-world networks. For example, the optimal clustering strategy from \cite{saveski2017detecting} accounts for only $35.59\%$ of total edges, suggesting that bias remains even after the clustering. Secondly, while cluster randomized experiments are prevalent in social sciences, platforms like eBay, Facebook, and Twitter primarily conduct standard simple randomized experiments, or Bernoulli trials \citep{chin2019regression}. Lastly, cluster randomized experiments fail to identify direct and indirect average treatment effects, crucial for various applications such as understanding how policies in one region affect behavior in others \citep{holtz2020interdependence}. Thus, there is an urgent need to develop data-adaptive methods for estimating and inferring direct and indirect effects under Bernoulli trials with network interference.

To fill the gap, we develop a novel approach for estimating and inferring the direct average treatment effect (DATE), indirect average treatment effect (IATE), and GATE under the heterogeneous additive treatment effect (HATE) model \citep{sussman2017elements}. The HATE model posits that each unit's outcome is a linear combination of its own treatment assignment and those of its network neighbors, allowing the heterogeneity of indirect effects from different neighbors. This model has gained popularity for incorporating interference \citep{toulis2013estimation,pouget2018optimizing,hu2022average,yu2022estimating, harshaw2023design, chen2024optimized, holtz2024reducing} within the Neyman--Rubin potential outcomes framework \citep{Neyman1923, rubin1974}. We allow the HATE model to accommodate a latent or hidden network that drives interference \citep{yu2022estimating}. For example, in the context of a tech company, users' social media accounts may form a visible network, while interference is driven by the unobserved network of offline friendships, which only encompasses a fraction of the observed network edges. \cite{yu2022estimating} developed efficient randomized designs to produce an unbiased estimator of GATE when average baseline measurements of outcomes prior to the experiment are available. However, the asymptotic distribution of this GATE estimator remains unknown. Meanwhile, \cite{hu2022average} proposed unbiased Horvitz--Thompson estimators for the DATE, IATE, and GATE under a broader interference setting. Nevertheless, the asymptotic properties of these estimators are yet to be explored.

Our contributions are four-fold. Firstly, we formally establish the design-based central limit theory (CLT) for the Horvitz--Thompson estimators of DATE, IATE, and GATE. The design-based inference framework is prevalent in causal inference literature with interference \citep{leung2022causal,aronow2017estimating,hudgens2008toward}. This framework relies solely on the randomness of treatment assignment, conditioned on the network and parameters underlying the HATE model. Consequently, conclusions drawn from this framework do not depend on the stochastic process underlying the network or model parameters. For achieving this goal, we use the newly developed CLT for homogeneous sums \citep{koike2022high} instead of the traditional CLT that relies on the approach of dependency graph \citep{aronow2017estimating,chin2018central,harshaw2023design,ogburn2024causal}. Equipped with this tool, our CLT requires that the average node degree is of magnitude $o(N)$ ($N$ is the number of experimental units), a notable improvement over the dependency graph approach, which requires a maximum node degree of magnitude $o(N^{1/2})$ \citep{ogburn2024causal}.

Secondly, we obtain conservative variance estimators for the Horvitz--Thompson estimators of DATE, IATE, and GATE, facilitating the construction of asymptotically conservative confidence intervals. Compared to the method proposed by \cite{aronow2017estimating}, our variance estimators offer enhanced tractability and stability, particularly as the average node degree tends to infinity.

Thirdly, we find that the Horvitz--Thompson estimators of IATE and GATE become inconsistent in dense networks (in the sense that the average network degree divided by $\sqrt{N}$ does not tend to zero). To address this issue, we propose novel eigenvector-adjusted (EV-adjusted) estimators, ensuring consistency even in significantly dense networks. We establish the asymptotic normality of the EV-adjusted IATE and GATE estimators under the design-based inference framework. The idea of adjusting the eigenvectors was originally proposed by \cite{li2022random} to improve the convergence rates of IATE and GATE estimators in a graphon model. Our eigenvector-based regression adjustment is a model-assisted approach, allowing the graphon model to be misspecified.

Finally, we discuss the application of our method to various interference structures, such as partial interference \citep{sobel2006randomized,imai2021causal} and local interference in a two-sided marketplace \citep{masoero2024multiple}. In the context of partial interference, we demonstrate that the EV-adjusted GATE estimator achieves the same convergence rate as the classic unbiased estimator used in cluster-randomized experiments. Regarding local interference in a two-sided marketplace, we observe that the Horvitz--Thompson estimators lack consistency, while the EV-adjusted estimators maintain consistency and asymptotic normality. These concrete examples highlight the flexibility of our method in handling complex interference patterns.

While several papers have explored GATE estimation in experiments with interference, our work distinguishes itself in several aspects. For example, \cite{chin2019regression} focused on the super-population framework, while we consider the finite-population or design-based inference. \cite{harshaw2023design} proposed a bipartite experimental design, where the units receiving treatment are distinct from those whose outcomes are measured. They studied the estimation of the GATE under a linear exposure model, which requires that the strength of interference is captured by a known weight matrix. In contrast, we consider the Bernoulli trials and the more general HATE model, which allows the strength of interference between units unknown. 
Additionally, these studies did not explore the estimation of DATE and IATE. Another line of research employed exposure mapping to define average treatment effects under arbitrary interference, but often relying on stronger assumptions of network structure \citep{aronow2017estimating, leung2022causal, hoshino2023causal, savje2024causal}. For example, both \cite{aronow2017estimating} and \cite{hoshino2023causal} assumed that the network possesses bounded node degrees.

The remainder of this paper is organized as follows. In \Cref{sec:notation}, we introduce the framework and notation including the HATE model. In \Cref{sec:estimating ATE}, we derive unbiased estimators of the DATE, IATE, and GATE and discuss their consistency. In \Cref{sec:CLT-and-variance-estimator}, we establish the design-based central limit theory for these estimators and derive conservative variance estimators. In \Cref{sec:eigen-value-adjustment}, we introduce the eigenvector-based regression adjustment method to improve the convergence rates of the IATE and GATE estimators. In \Cref{sec:application-in-some-networks}, we provide concrete examples illustrating the application and benefits of the regression adjustment methods in three commonly studied interference networks. In Section~\ref{sec:sim}, we conduct a simulation and real data analysis to evaluate the finite-sample performance of the proposed methods. Lastly, we conclude the paper with a discussion of future research in Section~\ref{sec:con}. Proofs are relegated to the Supplementary Material.

\section{Framework and notation}
\label{sec:notation}
Consider a randomized experiment with $N$ units where each unit $i$ ($i=1,\ldots,N$) is independently assigned a binary treatment $Z_i\in \{0,1\}$. Denote $Z_i=1$ if unit $i$ is assigned to the treatment group and $Z_i=0$ if assigned to the control group. We assume that  $Z_i\sim \operatorname{Bernoulli}(r_1)$ for a constant $r_1 \in (0,1)$ determined by the experimenter with $r_0=1-r_1$. Let $\bs{Z} = (Z_1,\ldots,Z_N)\in \{0,1\}^N$ be the vector of treatment assignments for all units.  For $\bs{z}\in \{0,1\}^N$, let $Y_i(\bs{z})$ be the potential outcome 
of unit $i$ under treatment $\bs{z}$. The observe outcome of unit $i$ satisfies $Y_i = Y_i(\bs{Z}) = \sum_{\bs{z} \in \{0,1\}^{N}} I ( \bs{Z}= \bs{z} )Y_i(\bs{z})$, where $I(\cdot)$ is the indicator function. Let $\bs{\mathcal{E}} = \{\bs{E}=(E_{ij})_{1\leq i,j\leq N}: E_{ii}=0,E_{ij}\in\{0,1\},1\leq i,j\leq N\}$ denote the set of adjacency matrices of directed network with no self-loop. Suppose that the experimental units interact with each other through an unobserved directed network with a no self-loop adjacency matrix $\tilde{\bs{E}}\in \bs{\mathcal{{E}}}$. $\tilde{\bs{E}}$ is a sub-network of an observed directed network $\bs{E}\in \bs{\mathcal{E}}$, i.e., $\tilde{E}_{ij}\leq E_{ij}$.  For example, $\bs{E}$ may represent the social network on Twitter, and $E_{ij}=1$ if $i$ follows $j$ on Twitter. The interference between units is driven by whether $i$ and $j$ are friends offline, which is coded by $\tilde{\bs{E}}$. Or $\bs{E}$ may represent the joint network of several candidate networks which might drive the interference and $\tilde{\bs{E}}$ is one of the candidate networks.

 We are interested in the global average treatment effect (GATE), defined as the difference between the average of outcomes if all units are in the treatment group and that if all units are in the control group:
\[
\tau_{\tot} = \frac{1}{N}\sumi \bigl\{Y_i(\bs{1}_N)-Y_i(\bs{0}_N)\bigr\},
\]
where  $\bs{1}_N$ and $\bs{0}_N$ are vectors of all $1$'s and $0$'s of length $N$, respectively.

Let $\bs{e}_{i,M}$ be the standard basis of $\mathbb R^M$ with $1$ in the $i$th coordinate and $0$ elsewhere. For tractable inference of the estimand, we assume throughout that the outcome follows a heterogeneous additive treatment effect (HATE) model \citep{sussman2017elements} with a hidden network  $\tilde{\bs{E}}$, that is, 
\begin{equation}
\label{eq:HATE-model}
    Y_i = \alpha_i + \theta_i Z_i + \sum_{j=1}^N \tilde{E}_{ij}\gamma_{ij} Z_j,
\end{equation}
where $\alpha_i = Y_i(\bs{0}_N)$ is the baseline potential outcome if all units are in the control group, $\theta_i = Y_i(\bs{e}_{i,N})-Y_i(\bs{0}_N)$ is the individual level direct effect of unit $i$'s treatment, and $\gamma_{ij} = Y_i(\bs{e}_{j,N})-Y_i(\bs{0}_{N})$, $j \neq i$, is the indirect effect by his friend $j$'s treatment. For $\tilde{E}_{ij}=0$, we have ${\gamma}_{ij}=0$. Here, $\alpha_i$, $\theta_i$ and $\gamma_{ij}$ are unobserved latent parameters.

When $E_{ij}=\tilde{E}_{ij}$ and $\gamma_{ij} = \beta_i v_{ij}$ with $v_{ij}$ known, the HATE model simplifies to the linear exposure model introduced by \cite{harshaw2023design}. Importantly, inferring the GATE under the HATE model does not require consistent estimation of the model parameters such as $\alpha_i$, $\theta_i$, and $\gamma_{ij}$; see the next section for details.

\begin{remark}
In \cite{yu2022estimating}, the authors addressed the estimation of GATE under the HATE model without access to any information regarding $\tilde{\bs{E}}$. They opted to define $\bs{E}$ as the most extensive network possible: the complete graph. Their approach necessitates knowledge of the baseline average potential outcome, $N^{-1}\sumi \bs{Y}_i(\bs{0}_N)$, serving as a surrogate to derive an unbiased estimator for GATE.
\end{remark}

\begin{remark}
There are generally two approaches to formalizing interference in the potential outcome framework. The first approach involves constraining the function class of potential outcomes, which often permits a broader range of functions but imposes stringent restrictions on the network structure and interference magnitude (see, for instance, Assumption 6 in \cite{leung2022causal} and Assumption 4 in \cite{viviano2023causal}). The second approach entails structural assumptions on the potential outcome model, which incorporate specific domain knowledge regarding the interference pattern. For instance, \cite{aronow2017estimating} and \cite{harshaw2023design} assumed that potential outcomes depend solely on a low-dimensional exposure mapping vector. \cite{manski1993} introduced the linear-in-means model, which suggests a framework where people compare their performance to their peers. \cite{munro2021treatment} explored the scenario of a rival market, where interference is mediated through market prices.  Structural assumptions impose tighter constraints on potential outcome models as a trade-off for more intricate network representations or magnitude of interference. The HATE model aligns with the latter approach, offering flexibility to accommodate the complexities of real-world network structures.
\end{remark}

Under the HATE model, the GATE is $$\tau_{\tot} = \frac{1}{N}\sum_{i=1}^N \theta_i + \frac{1}{N}\sum_{i=1}^N \sum_{j=1}^N \tilde{E}_{ij}\gamma_{ij}.$$
Let $\tilde{\gamma}_{ij} = \tilde{E}_{ij}\gamma_{ij}$. Recall that $\tilde E_{ii} = E_{ii} = 0$, then $\tilde \gamma_{ii} = 0$. Let $\tau_{\dir} $ and $\tau_{\ind} $ be the average direct and indirect effects, defined as 
\[
\tau_{\dir}  = \frac{1}{N}\sum_{i=1}^N \theta_i,\quad \tau_{\ind}  =  \frac{1}{N}\sum_{i=1}^N\sum_{j=1}^N \tilde{\gamma}_{ij}\quad \text{with}\quad\tau_{\tot}  = \tau_{\dir} +\tau_{\ind} .
\]

Under the HATE model, $\tau_{\dir} $ can be interpreted as the average effect of changing ego treatment from $0$ to $1$ while keeping alter treatments fixed, while $\tau_{\ind} $ can be interpreted as the average effect of changing alter treatments from zero treated to all treated while keeping ego treatment fixed.

{\bf Notation.} For a vector $\bs{a} = (a_1,\ldots,a_n)$, let $\diag(\bs{a})$ be the diagonal matrix with $a_i$'s being its diagonal elements. Let $\bs{1}_S$ and $\bs{0}_S$ be the vectors of all $1$'s and $0$'s of length $S$, respectively. Let $Y_i\sim \bs{X}_i$ represent the ordinary least squares fit of $Y_i$ on $\bs{X}_i$ and we use ``$+$" to add more covariate in the regression. Note that we use the regression formula solely to generate the treatment effect estimators. We will evaluate their properties under the design-based inference framework, without imposing linear regression model assumptions.

 \section{Estimating the average treatment effect}
 \label{sec:estimating ATE}
We can estimate the DATE by 
\[
\hat{\tau}_{\dir}   = \frac{1}{N}\sumi \Bigl\{\frac{Y_iZ_i}{r_1} - \frac{Y_i(1-Z_i)}{r_0}\Bigr\}.
\]
Note that $\hat{\tau}_{\dir} $ is the classic Horvitz--Thompson estimator of the average treatment effect in the no-interference setting \citep{li2022random}.
Let $Y_{Z_i=z} = \operatorname{E}(Y_i|Z_i=z)$. Define
\[
\sigma_{\dir,1}^2=\frac{1}{N^2r_1r_0} \sumi (r_0Y_{Z_i=1}+r_1Y_{Z_i=0})^2,\quad \sigma_{\dir,2}^2=\frac{1}{N^2}  \sum_{1\leq i\ne j \leq N} \tilde{\gamma}_{ij}^2 + \frac{1}{N^2}  \sum_{1\leq i\ne j \leq N} \tilde{\gamma}_{ij}\tilde{\gamma}_{ji}.
\]

 \begin{proposition}
 \label{prop:E-and-Var-DIR}
     $\operatorname{E} (\hat{\tau}_{\dir} ) =\tau_{\dir} $ and $\Var(\hat{\tau}_{\dir} ) = \sigma_{\dir,1}^2 + \sigma_{\dir,2}^2$.
 \end{proposition}
 
We now discuss the magnitude of the variance of $\hat{\tau}_{\dir} $. Let $\tilde{N}_i = \sumj \tilde{E}_{ij}$ and ${N}_i = \sumj {E}_{ij}$ be the out-degrees of unit $i$ in the latent network and observed network, respectively.  We posit the following assumption on the magnitude of $\theta_i$ and $\tilde{\gamma}_{ij}$.
\begin{assumption}
\label{a:bounded-parameter}
    There exists a constant $C$ independent of $N$ such that \[\max \Big\{\maxi |\alpha_i|,\maxi |\theta_i|, \max_{1\leq i,j \leq N}\tilde{N}_i|\tilde{\gamma}_{ij}| \Big\} \leq C.\]
\end{assumption}

\Cref{a:bounded-parameter} bounds the magnitude of the model parameters. Assuming anonymous interference, i.e., $\tilde{\gamma}_{ij} = \tilde{\gamma}_{i1}$ for all $j$ with $\tilde{E}_{ij} =1$, \Cref{a:bounded-parameter} simplifies to the assumption that $\maxi\max_{z\in\{0,1\}^N}|Y_i(z)|$ is bounded, aligning with the standard bounded outcome assumption \citep{aronow2017estimating,leung2022causal,li2022random}. Furthermore, \Cref{a:bounded-parameter} restricts $\max_{1\leq i,j \leq N}\tilde{N}_i|\tilde{\gamma}_{ij}|\leq C$, bounding the magnitude of the indirect effect of unit $j$ on unit $i$. This constraint is a relaxation of Assumption 4 in \cite{viviano2023causal}, where they required that $\max_{1\leq i,j \leq N} {N}_i|\tilde{\gamma}_{ij}|\min\{N,$ $\maxi N_i^2\} = o(1)$, translated into our framework.

Define the normalized adjacency matrix of the latent network by ${\bs{Q}} = (Q_{ij})_{1 \leq i,j\leq N}$ with $Q_{ij}= \tilde{E}_{ij}/\tilde{N}_i$ if $\tilde E_{ij}=1$ and $0$ otherwise. Under Assumption~\ref{a:bounded-parameter}, we have $\tilde{\gamma}_{ij}\leq CQ_{ij}$. Define a density parameter by $\rhon = \sumi N_i/N^2$. Then, $N\rhon$ is the average number of neighbors (out-degrees) for the units in the observed network. Let $\|\cdot\|_{\opnorm}$ be the operator norm of a squared matrix. We consider networks satisfying the following assumptions.
\begin{assumption}
\label{a:density-rho_N}
 $\rhon=o(1)$ and there exists a constant $C_{+}>0$, such that $N\rhon\geq C_{+}$.
\end{assumption}

Assumption \ref{a:density-rho_N} excludes extremely dense networks such as the complete graph and extremely sparse networks such as the edgeless graph. Assumption \ref{a:opnorm-EE^T} below imposes regularity conditions on $\bs{E}$ and $\bs{Q}$.  

\begin{assumption}
\label{a:opnorm-EE^T}
     (i) $\|\bs{E}\|_{\textnormal{op}} = O(N\rhon)$ and (ii) $\|\bs{Q}\|_{\textnormal{op}} = O(1)$.
\end{assumption}

\begin{remark}
\label{rmk:some-sufficient-condition-for-a-3}
    When the adjacency matrix is symmetric, i.e., $\bs{E}=\bs{E}^\top$ and $\tilde{\bs{E}}=\tilde{\bs{E}}^\top$, a sufficient condition for Assumption \ref{a:opnorm-EE^T}(i) is that there exists a constant $C_{+}>0$, such that $\maxi N_i \leq C_+ N\rhon$, i.e., the maximum degree is bounded by the average degree times a constant.  A sufficient condition for Assumption \ref{a:opnorm-EE^T}(ii) is that $\max_i \tilde{N}_i / \min_i \tilde{N}_i \leq C_+$, a uniform degree condition on the latent network; see the Supplementary Material for details.
\end{remark}
 
 Proposition~\ref{prop:order-of-tau-dir} below demonstrates the magnitude of the variance of $\hat{\tau}_{\dir} $.

\begin{proposition}
\label{prop:order-of-tau-dir}
   Under Assumptions \ref{a:bounded-parameter} and \ref{a:opnorm-EE^T}, we have $\Var(\hat{\tau}_{\dir} ) = O(N^{-1})$. 
\end{proposition}

We use the following modified Horvitz--Thompson estimator from \cite{hu2022average} for estimating $\tau_{\ind} $:
\[
\hat{\tau}_{\ind}  = \frac{1}{N}\sumi\sumj  E_{ij} \Bigl\{\frac{Y_iZ_j}{r_1} - \frac{Y_i(1-Z_j)}{r_0}\Bigr\}.
\]
 Let $Y_{Z_i=z_i}^{Z_j=z_j} = \operatorname{E}(Y_i|Z_i=z_i,Z_j=z_j)$, for $1\leq i,j\leq N$ and $(z_i,z_j)\in\{0,1\}^2$. Define
\begin{align*}
     &\sigma_{\ind,1}^2= \frac{1}{N^2r_1r_0}\sumj \Bigl\{\sumi E_{ij}(r_1r_0 Y_{Z_i=1}^{Z_j=1} +r_0^2Y_{Z_i=0}^{Z_j=1}+r_1^2Y_{Z_i=1}^{Z_j=0}+r_1r_0Y_{Z_i=0}^{Z_j=0})\Bigr\}^2,\\
     &\sigma_{\ind,2}^2=  \frac{1}{N^2}  \sum_{1\leq i\ne j \leq N} \Big(E_{ij}\theta_{i}+\sumk E_{kj}\tilde{\gamma}_{ki} \Big)^2 \\
     &\qquad \qquad + \frac{1}{N^2}  \sum_{1\leq i\ne j \leq N}  \Big(E_{ji}\theta_{j}+\sumk E_{ki}\tilde{\gamma}_{kj}\Big)\Big(E_{ij}\theta_{i}+\sumk E_{kj}\tilde{\gamma}_{ki}\Big).
\end{align*}

\begin{proposition}
\label{prop:order-and-formula-of-tau-ind}
    $\operatorname{E}(\hat{\tau}_{\ind} )=\tau_{\ind} $ and $\Var(\hat{\tau}_{\ind} ) =  \sigma_{\ind,1}^2+\sigma_{\ind,2}^2$. Under Assumptions \ref{a:bounded-parameter} and \ref{a:opnorm-EE^T}, we have $\sigma_{\ind,1}^2 =  O(N\rhon^2)$ and $\sigma_{\ind,2}^2 =  O(\rhon)$.
\end{proposition}

\Cref{prop:order-and-formula-of-tau-ind} shows the unbiasedness of the IATE estimator $\hat{\tau}_{\ind} $ and presents the magnitudes of the two components of the variance of $\hat{\tau}_{\ind} $. When $N\rhon\rightarrow \infty$, $\sigma_{\ind,1}^2$ is the dominant source of the variance.

We can estimate $\tau_{\tot} $ by a plug-in method: $\hat{\tau}_{\tot}  = \hat{\tau}_{\dir} +\hat{\tau}_{\ind} $. 
Define
\begin{align*}
\sigma_{\tot,1}^2 = & \frac{1}{N^2r_1r_0}\sumi\Bigl\{r_0Y_{Z_i=1}+r_1Y_{Z_i=0}\\
&+\sumj  E_{ji}(r_1r_0 Y_{Z_j=1}^{Z_i=1} +r_0^2Y_{Z_j=0}^{Z_i=1}+r_1^2Y_{Z_j=1}^{Z_i=0}+r_1r_0Y_{Z_j=0}^{Z_i=0})\Bigr\}^2,
\\
\sigma_{\tot,2}^2 = &  \frac{1}{N^2}  \sum_{1\leq i\ne j \leq N}  \Big(\tilde{\gamma}_{ij}+E_{ji}\theta_{j}+\sumk E_{ki}\tilde{\gamma}_{kj}\Big)\Big(\tilde{\gamma}_{ji}+E_{ij}\theta_{i}+\sumk E_{kj}\tilde{\gamma}_{ki}\Big) \\
&+ \frac{1}{N^2}  \sum_{1\leq i\ne j \leq N} \Big(\tilde{\gamma}_{ji}+E_{ij}\theta_{i}+\sumk E_{kj}\tilde{\gamma}_{ki}\Big)^2.
\end{align*}

\begin{proposition}
\label{prop:order-of-tau-tot}
 $\operatorname{E}(\hat{\tau}_{\tot} )=\tau_{\tot} $ and $ \Var(\hat{\tau}_{\tot} ) =  \sigma_{\tot,1}^2 + \sigma_{\tot,2}^2.$
    Under Assumptions \ref{a:bounded-parameter}--\ref{a:opnorm-EE^T}, we have $\sigma^2_{\tot,1} = O(N\rhon^2)$, and $\sigma^2_{\tot,2} = O(\rhon).$
\end{proposition}

Proposition~\ref{prop:order-of-tau-tot} implies the unbiasedness of the plug-in estimator $\hat{\tau}_{\tot} $ and provides the magnitudes of the two components of its variance. By \Cref{prop:order-of-tau-dir}, $\hat{\tau}_{\dir} $ is always  $N^{1/2}$-consistent. In other words, one can consistently estimate the DATE regardless of the network density. However, the IATE and GATE estimators may be inconsistent for dense networks. Propositions~\ref{prop:order-and-formula-of-tau-ind} and \ref{prop:order-of-tau-tot} implies that
a sufficient condition for the estimators $\hat{\tau}_{\ind} $ and $\hat{\tau}_{\tot} $ to 
be consistent is $\rhon = o(N^{-1/2})$. Those results echo the findings of \cite{li2022random}.

\section{Asymptotic normality and confidence interval}
\label{sec:CLT-and-variance-estimator}
\subsection{Central limit theorem}
We require the following assumptions to establish the asymptotic normality of $\hat{\tau}_{\dir} $, $\hat{\tau}_{\ind} $, and $\hat{\tau}_{\tot} $. Let $M_i = \sumj E_{ji}$ be the in-degree of unit $i$.
\begin{assumption} 
\label{a:Lindberg-condition-unadj}
    $\sumj Q_{ji} = o(N^{1/2})$, $\maxi N_i^2 = o(N^3\rhon^2)$, $\maxi M_i^2 = o(N^3\rhon^2)$. 
\end{assumption}

By Assumption~\ref{a:opnorm-EE^T}, we have $\sumi (\sumj Q_{ji})^2 = \bs{1}^\top \bs{Q}\bs{Q}^\top \bs{1} \leq N\|\bs{Q}\|_\opnorm^2 = O(N)$,
$
\sumi N_i^2 = \sumi (\sumj E_{ij})^2  \leq N\|\bs{E}\|_\opnorm^2 = O(N^3\rhon^2) $ and $ \sumi M_i^2 = \sumi (\sumj E_{ji})^2$ $\leq N\|\bs{E}\|_\opnorm^2 = O(N^3\rhon^2).
$ \Cref{a:Lindberg-condition-unadj} is therefore the Lindberg-type condition, which is a standard condition in the finite-population or design-based inference framework.

\begin{assumption}
\label{a:assumption-CLT}
    \begin{align*}
       &(i)  \text{ Either}\quad \liminf_N N\sigma^2_{\dir,1} >0\quad \text{or}\quad \liminf_N N\sigma^2_{\dir,2} >0.\\
        &(ii) \text{ Either}\quad \liminf_N \frac{\sigma^2_{\ind,1}}{N\rhon^2} >0\quad \text{or}\quad \liminf_N \frac{\sigma^2_{\ind,2}}{N\rhon^2} >0.\\
       &(iii) \text{ Either}\quad \liminf_N \frac{\sigma^2_{\tot,1}}{N\rhon^2} >0\quad \text{or}\quad \liminf_N \frac{\sigma^2_{\tot,2}}{N\rhon^2} >0.
    \end{align*}
\end{assumption}

\Cref{a:assumption-CLT} includes three parallel assumptions for deriving the asymptotic normality of the DATE, IATE, and GATE estimators, respectively.  \Cref{a:assumption-CLT}(i) requires that at least one component of $\sigma^2_{\dir} $ does not vanish asymptotically. Since $\sigma^2_{\dir}  = O(N^{-1})$, \Cref{a:assumption-CLT}(i) implies that $\liminf_N N\sigma^2_{\dir} >0$, and then the asymptotic distribution is not degenerated. \Cref{a:assumption-CLT}(ii) and (iii) have similar interpretations, but for the indirect and global effects, respectively.

\begin{theorem}
\label{thm:CLT-no-adjust}
    Under Assumptions~\ref{a:bounded-parameter}--\ref{a:assumption-CLT}, we have
    \[
    \frac{\hat{\tau}_{\dir} -\tau_{\dir} }{\Var(\hat{\tau}_{\dir} )^{1/2}} \xrightarrow{d} \mathcal{N}(0,1),~~\frac{\hat{\tau}_{\ind} -\tau_{\ind} }{\Var(\hat{\tau}_{\ind} )^{1/2}}\xrightarrow{d} \mathcal{N}(0,1),~~\frac{\hat{\tau}_{\tot} -\tau_{\tot} }{\Var(\hat{\tau}_{\tot} )^{1/2}}\xrightarrow{d} \mathcal{N}(0,1).
    \]
\end{theorem}

\Cref{thm:CLT-no-adjust} establishes the central limit theorem for $\hat{\tau}_{\dir} $, $\hat{\tau}_{\ind} $, and $\hat{\tau}_{\tot} $. Noteworthy, the CLT holds if $\rhon = o(1)$ (by \Cref{a:density-rho_N}), including the regime where $N^{1/2}\rhon$ tends to infinity such that $\hat{\tau}_{\ind} $ and $\hat{\tau}_{\tot} $ are inconsistent. Our proof is based on the newly developed central limit theory for homogeneous sums \citep{koike2022high}, while the traditional central limit theory relies on the approach of dependency graph \citep{aronow2017estimating,ogburn2024causal,harshaw2023design,chin2018central}. As far as we know, the best result obtained by the dependency graph approach requires a strong assumption that the maximum node degree is of magnitude $o(N^{1/2})$, as mentioned by \cite{ogburn2024causal}. 

\subsection{Confidence interval}
Based on \Cref{thm:CLT-no-adjust}, we can construct Wald-type confidence intervals for the average treatment effects as long as we can consistently or conservatively estimate the asymptotic variances. First, we try to estimate $\sigma_{\dir,1}^2$, which has a nice representation of conditional expectation of outcomes. We ignore the second term for the moment because it involves $N^2$ nonidentifiable terms and is bounded by or a smaller order of the first term under certain conditions.  Note that $Y_{Z_i=1}$ and $Y_{Z_i=0}$ are the expected outcomes of unit $i$ conditional on it being assigned to the treatment and control groups, respectively. Since we can only observe the outcome of unit $i$ under one treatment arm, we cannot consistently estimate the asymptotic variances. Instead, by Cauchy--Schwarz inequality, $(r_1 a_i + r_0 b_i)^2 \leq r_1 a_i^2 + r_0 b_i^2$, we obtain an upper bound of the first term of $\Var(\hat{\tau}_{\dir} )$,  which is given by 
\[
\sigma^2_{\dir,1}\leq \sumi \frac{Y_{Z_i=1}^2}{N^2r_1}  + \sumi   \frac{Y_{Z_i=0}^2}{N^2r_0}.
\]
Naturally, we use $Y_i^2I(Z_i=z)/\operatorname{E}\{I(Z_i=z)\}$ to estimate  $Y_{Z_i=z}^2$, which leads to a conservative estimator of $\Var(\hat{\tau}_{\dir} )$:
\[
\hat{V}_{\dir}  = \sumi \frac{Z_i Y_i^2}{N^2r_1^2}  + \sumi   \frac{(1-Z_i)Y_i^2}{N^2r_0^2}.
\]
Noteworthy, $\hat{V}_{\dir} $ is also the variance estimator for the average treatment effect estimator in a no-interference setting.
\begin{theorem} 
\label{thm:variance-estimator-dir}
(i) $\operatorname{E} (\hat{V}_{\dir} ) - \Var(\hat{\tau}_{\dir} )  = N^{-2}\sumi \theta_i^2 - N^{-2}\sumi\sumj  \tilde{\gamma}_{ij}\tilde{\gamma}_{ji},\\ \operatorname{E} (2\hat{V}_{\dir} )- \Var(\hat{\tau}_{\dir} ) \geq 0,$ 
and (ii) if Assumptions \ref{a:bounded-parameter} and \ref{a:opnorm-EE^T} holds, then $\hat{V}_{\dir} -\operatorname{E} \hat{V}_{\dir} = \op(N^{-1})$.
\end{theorem}

Theorem \ref{thm:variance-estimator-dir}(i) characterizes the bias of the variance estimator $\hat{V}_{\dir} $ with respect to $\Var(\hat{\tau}_{\dir} )$. The estimator $2\hat{V}_{\dir} $, i.e., the classic variance estimator of the Horvitz-Thompson estimator in the no-interference setting multiplied by a factor $2$, serves as a conservative estimator of $\Var(\hat{\tau}_{\dir} )$. Theorem \ref{thm:variance-estimator-dir}(ii) shows the consistency of $\hat{V}_{\dir} $ with respect to its expectation $\operatorname{E}\hat{V}_{\dir} $. Equipped with a central limit theorem and a variance estimator, we can construct a Wald-type confidence interval.
\begin{remark}
\label{rem:conservative-variance-estimator}
Under the condition discussed below, the second term in $\operatorname{E} (\hat{V}_{\dir} )-\Var(\hat{\tau}_{\dir} )$ is negligible compared to the first term, and thus $\hat{V}_{\dir} $ itself is already a conservative estimator. Specifically, if $\tilde{\bs{E}}$ is dense such that each unit is indirectly affected by its sufficiently many neighbors, i.e., $\mini \tilde{N}_i \rightarrow \infty$, then
 \[
 \operatorname{E} (\hat{V}_{\dir} ) - \Var(\hat{\tau}_{\dir} )  = \frac{1}{N^2}\sumi \theta_i^2  + o(N^{-1}).
 \]
 Thus, $\hat{V}_{\dir} $ is a conservative estimator of $\Var(\hat{\tau}_{\dir} )$. Moreover, since the first term of the bias of $\hat{V}_{\dir}$ is always positive, as shown in our numerical studies, $\hat{V}_{\dir}$ is overall conservative even when the assumption $\mini \tilde{N}_i \rightarrow \infty$ is violated. Therefore, we recommend $\hat{V}_{\dir}$ as the variance estimator to improve the power.
\end{remark}

Since $\sigma^2_{\ind,1}$ is the main source of variation and it is difficult to estimate $\sigma^2_{\ind,2}$, we derive an estimator for $\sigma^2_{\ind,1}$ and ignore $\sigma^2_{\ind,2}$ for the moment. By applying Cauchy--Schwarz inequality twice, we derive an upper bound for $\sigma^2_{\ind,1}$,
\begin{align*}
   \sigma_{\ind,1}^2 &\leq 
    \frac{1}{N^2 r_1}\sumi\Bigl\{\sumj  E_{ji}(r_1 Y_{Z_j=1}^{Z_i=1} +r_0Y_{Z_j=0}^{Z_i=1})\Bigr\}^2\\
    & \qquad \qquad \qquad\qquad + \frac{1}{N^2r_0}\sumi\Bigl\{\sumj  E_{ji}(r_1Y_{Z_j=1}^{Z_i=0}+r_0Y_{Z_j=0}^{Z_i=0})\Bigr\}^2\\
    & \leq  \frac{1}{N^2}\sumi\biggl(\sumj  E_{ji}Y_{Z_j=1}^{Z_i=1} \biggr)^2+\frac{r_0}{N^2r_1}\sumi\biggl(\sumj  E_{ji}Y_{Z_j=0}^{Z_i=1}\biggr)^2  + \\
    &  \qquad\qquad\qquad\qquad\frac{r_1}{N^2r_0}\sumi\biggl(\sumj  E_{ji}Y_{Z_j=1}^{Z_i=0}\biggr)^2  + \frac{1}{N^2}\sumi\biggl(\sumj  E_{ji}Y_{Z_j=0}^{Z_i=0}\biggr)^2.
\end{align*}
Then, we use
$$
\quad\sumi\biggl\{\sumj  E_{ji}Y_{j}\frac{I(Z_j=z)}{\operatorname{E}I(Z_j=z)} \biggr\}^2 \frac{I(Z_i=z^\prime)}{\operatorname{E }I(Z_i=z^\prime)}
$$
to estimate $ \sumi(\sumj  E_{ji}Y_{Z_j=z}^{Z_i=z^\prime})^2,$ which leads to an estimator of $\Var(\hat{\tau}_{\ind} )$,
\begin{align*}
    \hat{V}_{\ind} = &\sumi\biggl(\sumj  E_{ji}Y_{j}\frac{Z_j}{r_1}\biggr)^2 \frac{Z_i}{N^2r_1}+\sumi\biggl(\sumj  E_{ji}Y_{j} \frac{1-Z_j}{r_0}\biggr)^2 \frac{r_0Z_i}{N^2r_1^2}  + \\
        &  \qquad\qquad\sumi\biggl(\sumj  E_{ji}Y_{j} 
\frac{Z_j}{r_1}\biggr)^2 \frac{r_1(1-Z_i)}{N^2r_0^2} + \sumi\biggl(\sumj  E_{ji}Y_{j}\frac{1-Z_j}{r_0}\biggr)^2 \frac{1-Z_i}{N^2r_0}.
\end{align*}
Note that, $\hat{V}_{\ind} $ is the weighted sum of squares of treated/control neighbors' outcome totals. 

\begin{theorem}
\label{thm:variance-estimator-ind}
(i) $\operatorname{E} (2\hat{V}_{\ind} ) -\Var(\hat{\tau}_{\ind} ) \geq 0$.
If further assume Assumptions~\ref{a:bounded-parameter}--\ref{a:opnorm-EE^T} holds and $N\rhon \rightarrow \infty$, then, 
\begin{align*}
    \operatorname{E} \hat{V}_{\ind} - \Var(\hat{\tau}_{\ind} )  = \frac{1}{N^2} \sumi\left(\sumj  E_{ji}\theta_{j}\right)^2  + o(N\rhon^2).
\end{align*}
(ii) Under Assumptions~\ref{a:bounded-parameter}--\ref{a:Lindberg-condition-unadj}, we have $\hat{V}_{\ind} -\operatorname{E} (\hat{V}_{\ind} )  = \op(N\rhon^2)$.
\end{theorem}

Theorem \ref{thm:variance-estimator-ind}(i) characterizes the bias of $\hat{V}_{\ind} $ with respect to $\Var(\hat{\tau}_{\ind} )$ and demonstrates the conservativeness of $2\hat{V}_{\ind} $. Moreover, when $N\rhon\rightarrow \infty$,  $\hat{V}_{\ind} $ is also a conservative estimator of $\Var(\hat{\tau}_{\ind} )$. \Cref{thm:variance-estimator-ind}(ii) shows the consistency of $\hat{V}_{\ind} $ with respect to its expectation $\operatorname{E}\hat{V}_{\ind} $. By \Cref{thm:CLT-no-adjust} and \Cref{thm:variance-estimator-ind}, the Wald-type confidence intervals based on $2\hat{V}_{\ind}$ and $\hat{V}_{\ind}$ are asymptotically valid even when $N^{1/2}\rhon \rightarrow \infty$ and $\rhon = o(1)$, the asymptotic regime that $\hat{\tau}_{\ind}$ is inconsistent. However, the confidence intervals may be too wide and lose precision in this case.

Our estimator of $\Var(\hat{\tau}_{\tot} )$ is motivated by the expression of $\sigma_{\tot,1}^2$. Analogously, we derive an upper bound for $\sigma_{\tot,1}^2$,
\begin{align*}
     &  \frac{1}{N^2}\sumi\biggl(Y_{Z_i=1}+\sumj  E_{ji}Y_{Z_j=1}^{Z_i=1} \biggr)^2+\frac{r_0}{N^2r_1}\sumi\biggl(Y_{Z_i=1}+\sumj  E_{ji}Y_{Z_j=0}^{Z_i=1}\biggr)^2  + \\
    &  \qquad\qquad\qquad\frac{r_1}{N^2r_0}\sumi\biggl(Y_{Z_i=0}+\sumj  E_{ji}Y_{Z_j=1}^{Z_i=0}\biggr)^2  + \frac{1}{N^2}\sumi\biggl(Y_{Z_i=0}+\sumj  E_{ji}Y_{Z_j=0}^{Z_i=0}\biggr)^2,
\end{align*}
and an estimator for this upper bound is
\begin{align*}
    \hat{V}_{\tot} = &\sumi\biggl(Y_i+\sumj  E_{ji}Y_{j}\frac{Z_j}{r_1}\biggr)^2 \frac{Z_i}{N^2r_1}+\sumi\biggl(Y_i+\sumj  E_{ji}Y_{j} \frac{1-Z_j}{r_0}\biggr)^2 \frac{r_0Z_i}{N^2r_1^2}  + \\
        &  \qquad\sumi\biggl(Y_i+\sumj  E_{ji}Y_{j} 
\frac{Z_j}{r_1}\biggr)^2 \frac{r_1(1-Z_i)}{N^2r_0^2} + \sumi\biggl(Y_i+\sumj  E_{ji}Y_{j}\frac{1-Z_j}{r_0}\biggr)^2 \frac{1-Z_i}{N^2r_0}.
\end{align*}
Note that, $\hat{V}_{\tot} $ is the weighted sum of squares of treated/control neighbors' outcome totals plus ego outcomes of treated/control units.

\begin{theorem}
\label{thm:variance-estimator-tot}
(i) $\operatorname{E}(2\hat{V}_{\tot} )-\Var(\hat{\tau}_{\tot} ) \geq 0$.
If further assume Assumptions~\ref{a:bounded-parameter}--\ref{a:opnorm-EE^T} hold and $N\rhon \rightarrow \infty$, then, 
\begin{align*}
    \operatorname{E} (\hat{V}_{\tot} ) - \Var(\hat{\tau}_{\tot} )  = \frac{1}{N^2} \sumi\left(\theta_i +\sumj  E_{ji}\theta_{j}\right)^2  + o(N\rhon^2).
\end{align*}
(ii) Under Assumptions~\ref{a:bounded-parameter}--\ref{a:Lindberg-condition-unadj}, we have $\hat{V}_{\tot} -\operatorname{E} \hat{V}_{\tot}   = \op(N\rhon^2)$.
\end{theorem}

Theorem \ref{thm:variance-estimator-tot}(i) characterizes the bias of $\hat{V}_{\tot} $ with respect to $\Var(\hat{\tau}_{\tot} )$ and the conservativeness of $2\hat{V}_{\tot} $. If $N\rho_N$, the average number of neighbors, tends to infinity, $\hat{V}_{\tot}$ is also a conservative estimator. Theorem \ref{thm:variance-estimator-tot}(ii) shows the consistency of $\hat{V}_{\tot} $ with respect to its expectation $\operatorname{E}\hat{V}_{\tot} $. Similar to IATE, the Wald-type confidence interval based on $2\hat{V}_{\tot}$ and $\hat{V}_{\tot}$  is asymptotically valid even in the regime where  $\hat{\tau}_{\tot} $ is inconsistent.

\begin{remark}
\label{rmk:comparing-with-variance-estimator-in-the-literature}

Our variance estimators rely on the HATE model. \cite{aronow2017estimating} developed a general approach for estimating variance, which has been adopted by \cite{cortez2023exploiting,bojinov2023design}, among others. Specifically, let $\mathcal{N}_i$ be the set of neighbors of unit $i$, $\bs{z}_{\mathcal{N}_i}$ and $\bs{Z}_{\mathcal{N}_i}$ represent the subvectors of assignment vectors $\bs{z}$ and $\bs{Z}$ corresponding to $i$'s neighbors. If we assume one's potential outcomes only depend on its neighbors' treatment vector, their approach involves expressing the estimand and estimators as:
    \[
    N^{-1}\sumi\sum_{\bs{Z}_{\mathcal{N}_i}\in \{0,1\}^{|\mathcal{N}_i|}} Y_i(\bs{Z}_{\mathcal{N}_i})w_i(\bs{Z}_{\mathcal{N}_i})
    \]
and expanding the variance as:
$$
   \begin{aligned}
       &\frac{1}{N^2}\sum_{1 \leq i,j\leq N}\sum_{\bs{z}_i\in\{0,1\}^{|\mathcal{N}_i|}} \sum_{\bs{z}_j\in\{0,1\}^{|\mathcal{N}_j|}} Y_i(\bs{z}_i)w_i(\bs{z}_i)\:Y_j(\bs{z}_j)w_j(\bs{z}_j) \cdot \\     &\qquad \qquad \qquad \qquad \qquad \qquad \qquad \qquad \qquad   \mathrm{Cov}\big(I(\bs{Z}_{\mathcal{N}_i}=\bs{z}_{\mathcal{N}_i}),I(\bs{Z}_{\mathcal{N}_j}=\bs{z}_{\mathcal{N}_j})\big),
   \end{aligned}
   $$
where $w_i(\bs{Z}_{\mathcal{N}_i})$ is the sample weight  of unit $i$ with neighbor assignment $\bs{Z}_{\mathcal{N}_i}$ used in the estimand. They then apply the Cauchy--Schwarz inequality and inverse probability weighting to derive a conservative variance estimator. Unlike ours, their estimator lacks a closed-form expression and is computationally inefficient due to the summation of $N^2 2^{2N\rhon}$ terms. Moreover, their estimator may exhibit poor performance since the inverse probability weighting relies on the probability $\operatorname{E}(I(\bs{Z}_{\mathcal{N}_j}=\bs{z}_{\mathcal{N}_j})I(\bs{Z}_{\mathcal{N}_i}=\bs{z}_{\mathcal{N}_i}))$, which tends to $0$ as the average degree approaches infinity. \cite{harshaw2023design} derived an unbiased variance estimator for the linear exposure model by solving $N^2$ linear equations. However, their method does not extend to the HATE model and is not applicable when there exists linear equations with no solution.
\end{remark}

\section{Eigenvector-adjusted estimators of indirect and \\global average treatment effect}
\label{sec:eigen-value-adjustment}

The consistency of $\hat{\tau}_{\dir} $ and $\hat{\tau}_{\ind} $ requires $\rhon = o(N^{-1/2})$, i.e., the network is not too dense. To address this issue, we propose an eigenvector regression adjustment method for $\hat{\tau}_{\ind} $ and $\hat{\tau}_{\tot} $. 
Under appropriate conditions, the proposed method ensures consistency even when $\rhon = o(1)$.

Let's first find the reason for the inconsistency of $\hat{\tau}_{\ind} $ and $\hat{\tau}_{\tot} $ when $N\rhon\rightarrow\infty$. Let $\bs{c} = (r_1Y_{Z_1=1}+r_0Y_{Z_1=0},\ldots,r_1Y_{Z_N=1}+r_0Y_{Z_N=0})^\top$. When $N\rhon\rightarrow\infty$ and Assumption~\ref{a:bounded-parameter} holds, we can prove that (see Supplementary Materials for details)
$$
\begin{aligned}
    \Var(\hat{\tau}_{\ind} ) &= \frac{1}{N^2r_1r_0}\sumi\Bigl\{\sumj  E_{ji}(r_1Y_{Z_j=1}+r_0Y_{Z_j=0})\Bigr\}^2+O(\rhon)\\
&= \frac{1}{N^2r_1r_0} \bs{c}^\top\bs{E}\bs{E}^\top \bs{c} +O(\rhon).
\end{aligned}
$$
Let $\lambda_1\geq \ldots \geq \lambda_N$ be the ordered eigenvalues of $\bs{E}\bs{E}^\top$ with corresponding eigenvectors $\bs{V} = (\bs{V}_1,\ldots,\bs{V}_N)$, that is, 
 \[
 \bs{E}\bs{E}^\top = (\bs{V}_1,\ldots,\bs{V}_N)\diag(\lambda_1,\ldots,\lambda_N)(\bs{V}_1,\ldots,\bs{V}_N)^\top.
 \]

If the eigenvalues of $\bs{E}\bs{E}^\top$ are spiked, then the variance explosion is primarily caused by $\bs{c}$ that is correlated with the dominant eigenvectors of $\bs{E}\bs{E}^\top$, i.e., the eigenvectors corresponding to large eigenvalues. Meanwhile, note that $\bs{1}^\top \bs{E}\bs{E}^\top \bs{1}/N^2 = N\rhon^2$. Strong correlation between $\bs{c}$ and $\bs{1}$ also induces a large variance. Similar discussion applies to $\Var(\hat{\tau}_{\tot} )$.
Based on the above observation, we propose to adjust the correlation between $Y_{Z_i=z}$ and the dominant eigenvectors of $\bs{E}\bs{E}^\top$ as well as $\bs{1}$ to reduce the variance of the IATE estimator. Specifically, let $(\bs{V}_1,\ldots\bs{V}_K)\in \mathbb{R}^{N\times K}$ be the $K$ dominant eigenvectors of $\bs{E}\bs{E}^\top$. We regress $Y_i$ on  $(\bs{1},\bs{V}_1,\ldots,\bs{V}_K)$ for $\{i:Z_i=1\}$ and $\{i:Z_i=0\}$, respectively. Let $V_{ik}$ be the $i$th element of $\bs{V}_{k}$ for $k=1,\ldots,K$, and let $\bs{W}_i = (1,V_{i1},\ldots,V_{iK})^\top$. The Ordinary Least Squares (OLS) regression coefficient is  
\[
\hat{\bs{\beta}}_z= \argmin_{\bs{\beta}} \sumi (Y_i-\bs{\beta}^\top \bs{W}_i)^2I(Z_i=z).
\]
Let $\hat{e}_i = Y_i-\sum_{z\in \{0,1\}}I(Z_i=z)\hat{\bs{\beta}}_z^\top \bs{W}_i $ be the residual of the OLS regression. The EV-adjusted IATE and GATE estimators are defined as
\[\hat{\tau}_{\ind} ^{\ev} = \frac{1}{N}\sumi\sumj  E_{ij} \Bigl\{\frac{\hat{e}_iZ_j}{r_1} - \frac{\hat{e}_i(1-Z_j)}{r_0}\Bigr\},\quad \hat{\tau}_{\tot} ^{\ev} = \hat{\tau}_{\dir} +\hat{\tau}_{\ind} ^{\ev}.\]
\begin{remark}
\label{rmk:use-strata-indicator-as-regressor}
$\bs{V}_k$ represents the estimated membership indicator vector for community $k$ in the stochastic block model and the estimated eigencomponent of a graphon model \citep{li2022random}.  We can interpret the regression with $\bs{W}_i$ as adjusting the outcome for the main effect and the community effect. Importantly, we adopt the design-based inference framework, and thus, the theoretic results allow the stochastic block model to be misspecified. Let $\mathcal{N} = \{1,\ldots,N\}$ be the population set. If the network has several components with no between-cluster edges and $\cup_{h=1}^H \mathcal{S}_h = \mathcal{N}$ with $\mathcal{S}_h\cap \mathcal{S}_{h^\prime} = \emptyset$, for all $h\ne h^\prime$. We can also adjust the effect of the components by using $(I(i\in \mathcal{S}_1),\ldots, I(i\in \mathcal{S}_H))_{i=1}^N$ in the regression.
\end{remark}

Since linear regression yields the same output $\hat{e}_i$ under linear transformation of $\bs{W}_i$, $i=1,\ldots,N$, without loss of generality, we normalize $\bs{W}_i$ such that $N^{-1}\sum_{i=1}^{N} \bs{W}_i\bs{W}_i^\top = \bs{I}_{K+1}$. Next, we derive the asymptotic distributions of $\hat{\tau}_{\ind} ^{\ev}$ and $\hat{\tau}_{\tot} ^{\ev}$. Define the oracle regression coefficient for $z\in\{0,1\}$ as $\bs{\beta}_{z,\ora} = \sumi \bs{W}_i Y_{Z_i=z}/N,$ with the oracle residual 
\begin{align*}
    e_i &= Y_i-Z_i\bs{\beta}_{1,\ora}^\top \bs{W}_i-(1-Z_i)\bs{\beta}_{0,\ora}^\top \bs{W}_i\\
        &=(\alpha_i-\bs{\beta}_{0,\ora}^\top \bs{W}_i) + (\theta_i-\bs{\beta}_{1,\ora}^\top \bs{W}_i+\bs{\beta}_{0,\ora}^\top \bs{W}_i) Z_i + \sumj  \tilde{\gamma}_{ij}Z_j.
\end{align*}
Note that
\begin{align*}
    \bs{\beta}_{1,\ora}-\bs{\beta}_{0,\ora} &= \frac{1}{N}\sumi \bs{W}_i (Y_{Z_i=1}-Y_{Z_i=0})= \frac{1}{N}\sumi \bs{W}_i \theta_i.
\end{align*}
Therefore, $\theta_i-\bs{\beta}_{1,\ora}^\top \bs{W}_i+\bs{\beta}_{0,\ora}^\top \bs{W}_i$ is the OLS residual of regressing $\theta_i$ on $\bs{W}_i$. Let $\theta_{\backslash \bs{W},i} = \theta_i-\bs{\beta}_{1,\ora}^\top \bs{W}_i+\bs{\beta}_{0,\ora}^\top \bs{W}_i$, $e_{Z_i=z_i}^{Z_j=z_j} = \operatorname{E}(e_i\mathrel{|}Z_i=z_i,Z_j=z_j)$, and $e_{Z_i=z_i} = \operatorname{E}(e_i\mathrel{|}Z_i=z_i)$.
The oracle EV-adjusted IATE and GATE are defined as
\[
\tilde{\tau}_{\ind} ^{\ev} = \frac{1}{N}\sumi\sumj  E_{ij} \Bigl\{\frac{{e}_iZ_j}{r_1} - \frac{{e}_i(1-Z_j)}{r_0}\Bigr\},\quad \tilde{\tau}_{\tot} ^{\ev} = \hat{\tau}_{\dir} +\tilde{\tau}_{\ind} ^{\ev}.
\]
It's easy to see that $\Var(\tilde{\tau}_{\ind} ^{\ev})$ is the same as $\Var(\hat{\tau}_{\ind} )$ appeared in \Cref{prop:order-and-formula-of-tau-ind} except that $\theta_i$ and $Y_{Z_j=z_j}^{Z_i=z_i}$ are replaced with $\theta_{\backslash \bs{W},i}$ and $e_{Z_j=z_j}^{Z_i=z_i}$, respectively. The same is true for $\tilde{\tau}_{\tot} ^{\ev}$. Let $\bs{c}_e =(r_1 e_{Z_1=1}+r_0e_{Z_1=0},\ldots,r_1e_{Z_N=1}+r_0e_{Z_N=0})^\top$ and $\alpha_{\backslash \bs{W},i}$ be the OLS residual of regressing $\alpha_i$ on $\bs{W}_i$, 
\begin{align*}
    \alpha_{\backslash \bs{W},i} = \alpha_i - \bs{\beta}_\alpha^\top \bs{W}_i,\quad \text{where} \quad \bs{\beta}_\alpha = \frac{1}{N}\sumi \bs{W}_i \alpha_i.
\end{align*}
We define $h_{\backslash \bs{W},i }$ analogously as the OLS residual of regressing $h_i = \sumj  {\tilde{\gamma}}_{ij}$ on $\bs{W}_i$. Let
\[
\Deltan = N^{-1} \max\Big\{\sumi\Big(\sumj E_{ji}e_{Z_j=1}\Big)^2, \sumi\Big(\sumj E_{ji}e_{Z_j=0}\Big)^2\Big\}.
\]
Here, $\Deltan$ plays an important role in the variance of EV-adjusted estimators. We can show that $\Deltan=O(N^2\rhon^2)$ under \Cref{a:opnorm-EE^T}.

\begin{assumption}
\label{a:bounded-parameter-adjusted}
There exists a constant $C$ independent of $N$, such that $$\max\big\{ \max_{1 \leq i \leq N}|\alpha_{\backslash \bs{W},i}|, \max_{1 \leq i \leq N}|\theta_{\backslash \bs{W},i}|, \max_{1 \leq i \leq N}|h_{\backslash \bs{W},i}| \big\} \leq C.$$
\end{assumption}

\Cref{a:bounded-parameter-adjusted} is the analog of \Cref{a:bounded-parameter} extended to the adjusted parameters.
Define $(\sigma^{\ev}_{\ind,1})^2$ and $(\sigma^{\ev}_{\tot,1})^2$ similarly to $\sigma^2_{\ind,1}$ and $\sigma^2_{\tot,1}$ with  $Y_{Z_j=z}^{Z_i=z^\prime}$ ($(z,z^\prime)\in \{0,1\}^2$) replaced by $e_{Z_j=z}^{Z_i=z^\prime}$. Similarly, define $(\sigma^{\ev}_{\ind,2})^2$ and $(\sigma^{\ev}_{\tot,2})^2$ by replacing  $\theta_i$ with $\theta_{\bsw,i}$ in $\sigma^2_{\ind,2}$ and $\sigma^2_{\tot,2}$, respectively.
Proposition~\ref{prop:order-of-adjusted-estimators} below gives the formulas and magnitudes of the variances of EV-adjusted IATE and GATE estimators.
\begin{proposition}
\label{prop:order-of-adjusted-estimators}
Under Assumptions~\ref{a:bounded-parameter}--\ref{a:opnorm-EE^T} and \ref{a:bounded-parameter-adjusted}, we have
\begin{align*}
        &\Var(\tilde{\tau}_{\ind} ^{\ev}) =(\sigma_{\ind,1}^{\ev})^2+(\sigma_{\ind,2}^{\ev})^2  = O(N^{-1}\Deltan + \rhon),\\
        &\Var(\tilde{\tau}_{\tot} ^{\ev}) = (\sigma_{\tot,1}^{\ev})^2 + (\sigma_{\tot,2}^{\ev})^2 = O(N^{-1}\Deltan + \rhon).
    \end{align*}
\end{proposition}

\begin{remark} 
\label{rmk:estimate-of-deltan-eigenvalue}
Given that $e_{Z_i=z}$ are orthogonal to the first $K$ eigenvectors of $\bs{E}\bs{E}^\top$, and $\max_i\max_{z\in{0,1}}|e_{Z_i=z}|\leq C_e$ as implied by \Cref{a:bounded-parameter-adjusted}, where $C_e$ is a constant independent of $N$,  we can derive $N^{-1}\Deltan = O(\lambda_{K+1}/N)$.
\end{remark}

\begin{assumption}
\label{a:control-the-order-of-two-component}
\[
(i) \text{ Either}\quad \liminf \frac{(\sigma_{\ind,1}^{\ev})^2}{N^{-1}\Deltan +\rhon} >0\quad \text{or} \quad \liminf \frac{(\sigma_{\ind,2}^{\ev})^2 }{N^{-1}\Deltan +\rhon} >0.
\]
\[
(ii) \text{ Either}\quad \liminf \frac{(\sigma_{\tot,1}^{\ev})^2}{N^{-1}\Deltan +\rhon} >0\quad \text{or} \quad \liminf \frac{(\sigma_{\tot,2}^{\ev})^2 }{N^{-1}\Deltan +\rhon} >0.
\]
\end{assumption}
 
\begin{assumption}
\label{a:lind-berg-condition-for-adjusted-estimator}
$$N^{-1}  \max\{\maxi(\sumj E_{ji}e_{Z_j=1})^2, \maxi(\sumj E_{ji} e_{Z_j=0})^2\} = o(\Deltan ).$$
\end{assumption}

\begin{assumption}
\label{a:bounded-eigenvectors}
There exists a constant $C_{\bs{W}}$ independent of $N$, such that $ \|\bs{W}_i\|_\infty \leq C_{\bs{W}}$.
\end{assumption}

Assumption~\ref{a:control-the-order-of-two-component} is analogous to \Cref{a:assumption-CLT} on the variances of the EV-adjusted estimators. Assumption~\ref{a:lind-berg-condition-for-adjusted-estimator} is a Lindeberg-type condition for deriving the asymptotic normality of the EV-adjusted estimators. \Cref{a:bounded-eigenvectors} requires $\bs{W}_i$ to be uniformly bounded. Informally, if we view $\bs{W}_i$ as the estimated community membership or the component membership, \Cref{a:bounded-eigenvectors} requires the size of each estimated community or component used in the regression to be comparable to the population size $N$. For example, as in \Cref{rmk:use-strata-indicator-as-regressor}, if we have a component $\mathcal{S}_h$ of size $o(N)$, we should not directly use $(I(i\in\mathcal{S}_h))_{i=1}^N$ in the regression.

\begin{theorem}
\label{thm:CLT-EV-adj}
    Under Assumptions~\ref{a:bounded-parameter}--\ref{a:Lindberg-condition-unadj} and \ref{a:bounded-parameter-adjusted}--\ref{a:bounded-eigenvectors}, we have $(\hat{\tau}_{\ind} ^{\ev}-\tau_{\ind} )/\Var(\tilde{\tau}_{\ind} ^{\ev})^{1/2}\xrightarrow{d}\mathcal{N}(0,1)$ and $(\hat{\tau}_{\tot} ^{\ev}-\tau_{\tot} )/\Var(\tilde{\tau}_{\tot} ^{\ev})^{1/2}\xrightarrow{d}\mathcal{N}(0,1)$.
\end{theorem}

\Cref{thm:CLT-EV-adj} establishes the asymptotic normality of the EV-adjusted estimators $\hat{\tau}_{\ind} ^{\ev}$ and $\hat{\tau}_{\tot} ^{\ev}$. Their asymptotic variances are of order $N^{-1}\Deltan+\rhon$ (see \Cref{prop:order-of-adjusted-estimators}), representing an improvement in convergence rates over their unadjusted counterparts. In cases where $N^{-1}\Deltan=o(1)$, as demonstrated in the examples in the subsequent section, we can ensure the consistency of $\hat{\tau}_{\ind} ^{\ev}$ and $\hat{\tau}_{\tot} ^{\ev}$ provided that $\rhon = o(1)$.  Consequently, the EV-adjusted estimators relax the assumptions on the network structure compared to the unadjusted Horvitz--Thompson estimators.

To estimate the asymptotic variances and construct confidence intervals, we modify $\hat{V}_{\ind} $ and $\hat{V}_{\tot} $ by replacing the treated/control neighbors' outcome ($Y_i$) totals with their residual 
($\hat{e}_i$) totals, and we denote the resulting variance estimators as $\hat{V}_{\ind}^{\ev}$ and $\hat{V}_{\tot}^{\ev}$.

\begin{theorem}
\label{thm:asymptotic-variance-estimator-adjusted}
    Under Assumptions~\ref{a:bounded-parameter}--\ref{a:Lindberg-condition-unadj}, \ref{a:bounded-parameter-adjusted}, and \ref{a:lind-berg-condition-for-adjusted-estimator}--\ref{a:bounded-eigenvectors}, we have (i) $\hat{V}_{\ind} ^{\ev} = \Op(N^{-1}\Deltan + \rhon)$ and $\hat{V}_{\tot} ^{\ev} = \Op(N^{-1}\Deltan + \rhon)$, (ii) $ 2\hat{V}_{\ind} ^{\ev} -\Var(\tilde{\tau}_{\ind} ^{\ev}) = R_N^{(1)} + \op(N^{-1}\Deltan +\rhon)$ and $2\hat{V}_{\tot} ^{\ev}-\Var(\tilde{\tau}_{\tot} ^{\ev}) = R_N^{(2)}+\op(N^{-1}\Deltan +\rhon)$, 
 where $R_N^{(1)} \geq 0$ and $R_N^{(2)} \geq 0$,
 and (ii) if further assume that $\liminf N^{-1}\Deltan/\rhon \rightarrow \infty$, then  
\begin{align*}
   &\hat{V}_{\ind} ^{\ev}- \Var(\tilde{\tau}_{\ind} ^{\ev})  = \frac{1}{N^2} \sumi\left(\sumj  E_{ji}\theta_{\backslash \bs{W},j}\right)^2  + \op\Big(N^{-1}\Deltan\Big),\\
   &\hat{V}_{\tot} ^{\ev}-\Var(\tilde{\tau}_{\tot} ^{\ev})  = \frac{1}{N^2} \sumi\left(\theta_i +\sumj  E_{ji}\theta_{\backslash \bs{W},j}\right)^2  + \op\Big(N^{-1}\Deltan\Big).
\end{align*}
\end{theorem}

\Cref{thm:asymptotic-variance-estimator-adjusted}(i) presents the magnitude of the variance estimators.
Theorem \ref{thm:asymptotic-variance-estimator-adjusted}(ii) demonstrates that $2\hat{V}_{\ind} ^{\ev}$ and $2\hat{V}_{\tot} ^{\ev}$ are conservative variance estimators. Moreover, $\hat{V}_\star^{\ev}$ ($\star\in \{\ind,\tot\}$) exhibits a smaller size of $\Op(N^{-1}\Deltan +\rhon)$ compared to the $\Op(N\rhon^2)$ of $\hat{V}_\star$, thereby enhancing the precision of inference. Theorem \ref{thm:asymptotic-variance-estimator-adjusted}(iii) implies that the second term in the variances $\Var(\tilde{\tau}_{\ind} ^{\ev})$ and $\Var(\tilde{\tau}_{\tot} ^{\ev})$ becomes negligible compared to the first term if $N^{-1}\Deltan$ tends to $0$ slower than $\rhon$. In such cases, $\hat{V}_{\ind} ^{\ev}$ and $\hat{V}_{\tot} ^{\ev}$ already serve as conservative variance estimators.

\section{Application example}
\label{sec:application-in-some-networks}
In this section, we illustrate, by three concrete examples, that the EV-adjusted estimators can significantly improve the estimation precision and the inference efficiency, demonstrating their flexibility to handle various interference structures.

\subsection{Partial interference}
\label{sec:partial}
Partial interference refers to interference occurring solely among units within the same group \citep{sobel2006randomized,imai2021causal,baird2018optimal}. Under partial interference, saturation designs, an extension of cluster-randomized experiments, prove useful for estimating both direct and spillover effects \citep{baird2018optimal,imai2021causal}. This involves randomly assigning clusters with varying treatment proportions. We explore the use of Bernoulli trials for estimating treatment effects and demonstrate their efficiency, showing that the EV-adjusted GATE estimator under Bernoulli trials can achieve the same convergence rate as the Horvitz--Thompson estimator of the GATE under cluster-randomized experiments.

Suppose the target population is divided into $M$ non-overlapping, equally-sized groups (clusters), each containing $N_M$ units, resulting in a total population size of $N = MN_M$. For ease of representation, let $\{(m-1)N_M+1,\ldots,mN_M\}$ denote group $m$. We denote by $C_i\in \{1,\ldots,M\}$ the group indicator of unit $i$, where $C_i=m$ if $i\in\{(m-1)N_M+1,\ldots,mN_M\}$. Each group forms a complete graph, and thus the whole network is composed of components of complete graphs, each of size $N_M$, with no between-group edges. Partial interference aligns with Assumption~\ref{a:stratified-interference} below, where each group corresponds to a complete graph.

\begin{assumption}\label{a:stratified-interference}
    For $i\ne j\in\{1,\ldots,N\}$, $E_{ij}=1$ if $C_i=C_j$.
\end{assumption}

Let $\bs{J}_{N_M}$ denote the square matrix of all ones with dimensions $N_M \times N_M$. Under partial interference, the adjacency matrix takes the form of a block diagonal matrix $\diag(\{\bs{J}_{N_M}-\bs{I}_{N_M}\}_{m=1}^M)$. The density of the network is $\rhon = M\binom{N_M}{2}/\binom{N}{2}= O(M^{-1})$. Corollaries \ref{cor:variance-estimator-adj-stratified} and \ref{cor:variance-estimator-adj-stratified2} below provide the convergence rates for the unadjusted and EV-adjusted IATE and GATE estimators.

\begin{corollary}
\label{cor:variance-estimator-adj-stratified}
    Under Assumptions~\ref{a:bounded-parameter}--\ref{a:opnorm-EE^T} and \ref{a:stratified-interference}, we have $\hat{\tau}_{\ind}-\tau_{\ind}  = \Op(N^{1/2}M^{-1})$ and $\hat{\tau}_{\tot}-\tau_{\tot}  = \Op(N^{1/2}M^{-1})$.
\end{corollary}

Recall from Theorem~\ref{prop:order-of-tau-dir} that the convergence rate of $\hat{\tau}_{\dir} $ is $N^{-1/2}$, regardless of the network structure. Therefore, to consistently estimate $\tau_{\ind} $ and $\tau_{\tot} $ using $\hat{\tau}_{\ind} $ and $\hat{\tau}_{\tot} $, we require $N_M = o(N^{1/2})$. An alternative approach to estimating GATE under partial interference involves conducting cluster-randomized experiments, where each group (cluster) receives the same treatment. The difference in means of outcomes under the treatment and control is an unbiased estimator of GATE, with a convergence rate of $\Op(M^{-1/2})$. This rate is faster than that of $\hat{\tau}_{\tot} $ unless $N_M = O(1)$.

It is well known that the spectrum of a complete graph is $(N_M-1)^1$ and $(-1)^{N_M-1}$, where the exponent denotes the multiplicity. Therefore, the spectrum of this block diagonal matrix is $(N_M-1)^M$ and $(-1)^{N-M}$. Let $\bs{e}_{m, M}$ be the $m$-th canonical basis of $\mathbb R^M$. The $M$ dominant eigenvectors are $\{\bs{e}_{m,M}\otimes\bs{1}_{N_M}\}_{m=1}^M$, where $\otimes$ is the Kronecker product. Here, $\bs{e}_{m,M}\otimes\bs{1}_{N_M}$ is the vector of the group membership indicator of group $m$. The variances of $\hat{\tau}_{\ind} $ and $\hat{\tau}_{\tot} $ are primarily attributed to the correlation between $Y_{i}$ and these eigenvectors. Therefore, it is intuitive to use $\bs{W}_i = (I(C_i=1),\ldots,I(C_i=M))^\top$ in the regression.

We cannot adjust for all the $M$ eigenvectors because our theory restricts us to using a fixed number of eigenvectors. Moreover, adjusting the effect of all groups violates \Cref{a:bounded-eigenvectors}, which requires the size of each group to be comparable with $N$, i.e., $N_M \geq C N$ for some constant $C > 0$. To address this issue, we assume the between-groups homogeneity (see \Cref{a:no-heterogeneity-between-groups} below) and then we can merge $M$ group indicators into a single indicator, leading to $\hat{e}_i$ obtained from
\begin{align}
\label{eq:regression-reduced-group-indicator}
    Y_{i}\sim  I(Z_{i}=0) + I(Z_{i}=1).
\end{align}

\begin{assumption}
\label{a:no-heterogeneity-between-groups}
   $N_M^{-1} \sum_{i:C_i=m} Y_{Z_i=z} = N_M^{-1}\sum_{i:C_i=1} Y_{Z_i=z}$, for all $m=1,\ldots,M$, $z=0,1$.
\end{assumption}

\begin{remark}
\label{rmk:stratify-the-group-partial-interference}
    We can further relax \Cref{a:no-heterogeneity-between-groups}. Specifically, if we can partition the $M$ groups into several strata and assume between-group homogeneity within each stratum, we can merge the cluster indicators into a stratum indicator and perform regression using the stratum indicator; see the Supplementary Material for details. 
\end{remark}

Using the residual of regression \eqref{eq:regression-reduced-group-indicator} to construct $\hat{\tau}^{\ev}_{\ind} $ and $\hat{\tau}^{\ev}_{\tot} $, we obtain their convergence rates, as shown in \Cref{cor:variance-estimator-adj-stratified2} below.
\begin{corollary}
\label{cor:variance-estimator-adj-stratified2}
       Under Assumptions~\ref{a:bounded-parameter}--\ref{a:Lindberg-condition-unadj} and \ref{a:bounded-parameter-adjusted}--\ref{a:no-heterogeneity-between-groups}, we have $\hat{\tau}_{\ind} ^{\ev}-\tau_{\ind} = \Op(M^{-1/2})$ and $\hat{\tau}_{\tot} ^{\ev} - \tau_{\tot} = \Op(M^{-1/2})$.
\end{corollary}

By \Cref{cor:variance-estimator-adj-stratified2}, $\hat{\tau}_{\tot} ^{\ev}$ achieves the same convergence rates as the difference in means estimator under cluster-randomized experiments. Meanwhile, we can estimate the DATE and IATE, which are not identifiable under cluster-randomized experiments.

Under \Cref{a:no-heterogeneity-between-groups}, the oracle residual $e_i$ of $\bs{W}_i \equiv 1$ is the same as that of $\bs{W}_i \equiv (I(C_i=1),\ldots,I(C_i=M))^\top$. Thereby, we actually adjust the first $M$ eigenvectors of the graph using \Cref{eq:regression-reduced-group-indicator} (See Supplementary Materials for more details). 
Subsequently, following the argument of \Cref{rmk:estimate-of-deltan-eigenvalue}, we find that $\Deltan = O(\lambda_{M+1}) = O(1)  = o(N_M^{-1})= o(N\rhon)$.
By \Cref{prop:order-of-adjusted-estimators}, $\hat{\tau}_{\star}^{\ev}$, $\star\in\{\ind,\tot\}$ improves the convergence rate of $\hat{\tau}_{\star}$, from $N^{1/2}\rhon$ to $\rhon^{1/2}$.

\subsection{Local interference in a two-sided marketplace}
\label{sec:local}
There is a growing interest in estimating GATE in experiments carried out in a two-sided marketplace \citep{masoero2024multiple,holtz2024reducing,harshaw2023design}. The target population has a matrix form. Consider a $N_R \times N_C$ matrix with $N=N_RN_C$ where each row represents a buyer and each column represents a seller. We observe outcomes for the $(r,c)$-th pair under a treatment targeted to that specific pair, such as the amount of money the $r$-th buyer spends with the $c$-th seller, with or without the implementation of an encouragement policy. Equipped with the EV-adjusted estimators, we aim for valid inferences in Bernoulli trials under local interference assumption in a two-sided marketplace, where the interference exists within rows and columns.

Let $C_i$ and $R_i$ be the column and row numbers of the $i$-th pair, where $C_i=c$ and $R_i=r$ if $i=(r-1)N_R+c$, $1\leq i\leq N$. Local interference assumes that interference occurs solely within units in the same column (same seller) or row (same buyer). We then translate this assumption into the language of networks.
\begin{assumption}
\label{a:local-interference-for-marketplace}
    For $i\ne j\in\{1,\ldots,N\}$, $E_{ij}=1$ if $C_i=C_j$ or $R_i=R_j$.
\end{assumption}

Under \Cref{a:local-interference-for-marketplace}, $\rhon = \big\{N_R\binom{N_C}{2}+ N_C\binom{N_R}{2}\big\}/N^2 = O(N_R^{-1}+N_C^{-1})$. Hence, Assumption~\ref{a:density-rho_N} requires that $N_R,N_C\rightarrow \infty$. \Cref{cor:local-interference-for-marketplace} below presents the convergence rates of $\hat{\tau}_{\ind} $ and $\hat{\tau}_{\tot} $.

\begin{corollary}
\label{cor:local-interference-for-marketplace}
    Under Assumptions~\ref{a:bounded-parameter}, \ref{a:opnorm-EE^T}, and \ref{a:local-interference-for-marketplace}, we have $\hat{\tau}_{\ind} - \tau_{\ind}  = \Op(N^{1/2}N_R^{-1}+N^{1/2}N_C^{-1})$ and $\hat{\tau}_{\tot} - \tau_{\tot}  = \Op(N^{1/2}N_R^{-1}+N^{1/2}N_C^{-1})$.
\end{corollary}

The consistency of $\hat{\tau}_{\ind} $ and $\hat{\tau}_{\tot} $ requires that $\max\{N_R,N_C\}=o(N^{1/2})$, which is impossible. Therefore, we need to use the EV-adjusted estimators. It is not difficult to show that
\[
\bs{E} = \sum_{r=1}^{N_R}(\bs{e}_{r,N_R}\otimes\bs{1}_{N_C})^\top (\bs{e}_{r,N_R}\otimes\bs{1}_{N_C})+ \sum_{c=1}^{N_C} (\bs{1}_{N_R}\otimes \bs{e}_{c,N_C})^\top(\bs{1}_{N_R}\otimes \bs{e}_{c,N_C}) -2\bs{I}_N.
\]
Therefore, the dominant eigenvectors of $\bs{E}$ are in the space spanned by $\{\bs{e}_{r,N_R}\otimes\bs{1}_{N_C}\}_{r=1}^{N_R}\cup\{\bs{1}_{N_R}\otimes \bs{e}_{c,N_C}\}_{c=1}^{N_c}$, which are the column and row indicator vectors. Choosing the first row and first column as the reference level, it is intuitive to use $\bs{W}_i = (1,I(C_i=2),\ldots, I(C_i=N_C),I(R_i=2),\ldots,I(R_i=N_C))^\top$ in the regression adjustment. This regression contains $2(N_R+N_C-1)$ parameters with $2(N_R+N_C-1)\rightarrow \infty$ and the corresponding eigenvectors violate \Cref{a:bounded-eigenvectors}. To reduce the number of regressors and fulfill \Cref{a:bounded-eigenvectors}, we assume homogeneity between columns and rows and perform regression \eqref{eq:regression-reduced-group-indicator}. Then, we use the regression residual to construct $\hat{\tau}^{\ev}_{\ind} $ and $\hat{\tau}^{\ev}_{\tot} $ with convergence rates being presented in \Cref{cor:local-interference-for-marketplace-adj} below.

\begin{assumption}
    \label{a:no-heterogeneity-between-columns-and-rows}
    For all $c=1,\ldots,N_C$ and $r=1,\ldots,N_R$,
    \begin{align*}
        N_R^{-1} \sum_{i:C_i=c} Y_{Z_i=z} = N_R^{-1}\sum_{i:C_i=1} Y_{Z_i=z},\quad N_C^{-1} \sum_{i:R_i=r} Y_{Z_i=z} = N_C^{-1}\sum_{i:R_i=1} Y_{Z_i=z}.
    \end{align*}
\end{assumption}

\begin{remark}
Similar to \Cref{rmk:stratify-the-group-partial-interference}, we can relax \Cref{a:no-heterogeneity-between-columns-and-rows}, if we can group sellers and buyers, and assume both between-group buyer-homogeneity and between-group seller-homogeneity. The corresponding regression is a 
two way effect linear regression; see the Supplementary Materials for more details. 
\end{remark}

\begin{corollary}
\label{cor:local-interference-for-marketplace-adj}
    Under Assumptions~\ref{a:bounded-parameter}--\ref{a:Lindberg-condition-unadj}, \ref{a:bounded-parameter-adjusted}--\ref{a:bounded-eigenvectors}, and \ref{a:local-interference-for-marketplace}--\ref{a:no-heterogeneity-between-columns-and-rows}, we have $\hat{\tau}_{\ind} ^{\ev} - \tau_{\ind} = \Op(N_R^{-1/2}+N_C^{-1/2})$ and $\hat{\tau}_{\tot} ^{\ev} - \tau_{\tot} = \Op(N_R^{-1/2}+N_C^{-1/2})$. 
\end{corollary}

\Cref{cor:local-interference-for-marketplace-adj} indicates that the EV-adjusted estimators $\hat{\tau}_{\ind} ^{\ev}$ and $\hat{\tau}_{\tot} ^{\ev}$ are still consistent if $N_R, N_C \rightarrow \infty$. The oracle residual $e_i$ of $\bs{W}_i \equiv 1$ is the same as that of $\bs{W}_i \equiv (1,I(C_i=2),\ldots, I(C_i=N_C),I(R_i=2),\ldots,I(R_i=N_C))^\top$. Therefore, regression \eqref{eq:regression-reduced-group-indicator} adjusts the first $N_R+N_C-1$ top eigenvectors of $\bs{E} \bs{E}^\top$ and $\lambda_{N_R+N_C}=(-2)^2=O(1)$, and we have $\Deltan = O(\lambda_{N_R+N_C}) = O(1) = o(N\rhon)$. Therefore, $\hat{\tau}_{\star}^{\ev}$, $\star\in\{\ind,\tot\}$ improves the convergence rate of $\hat{\tau}_{\star}$, from $N^{1/2}\rhon$ to $\rhon^{1/2}$.

\subsection{Random graph drawn from a graphon model}
\cite{li2022random} considered random graphs generated from a graphon model. To see the magnitude of $\Deltan$ and the convergence rates of the IATE and GATE estimators, we assume that the observed network comes from a graphon model in this subsection.

In the graphon model, let $U_i\stackrel{\textnormal{iid}}{\sim}\textnormal{Uniform}[0,1]$ be a latent random variable associated with the $i$-th unit. Assume that $E_{ij} = E_{ji}$ and conditional on $U_i$'s, $E_{ij}$ is generated through $E_{ij}\sim \textnormal{Bernoulli}(G_N(U_i,U_j))$ for $i<j$, where $G_N:[0,1]^2\rightarrow [0,1]$ is called the graphon function. We consider $G_N(U_i,U_j) = \min\{1,\rhon^\star G(U_i,U_j)\}$ and focus on a setting where $G(U_i,U_j)$ has a low rank representation:
\[
G(U_i,U_j) = \sum_{k=1}^r \lambda_k \psi_k(U_i)\psi_k(U_j).
\]

\cite{li2022random} proposed the following regularity assumptions on the graphon function.
\begin{assumption}
\label{a:regularity-conditons-for-graphon-model}
   (i) The function $g_1(u)=\int_0^1 \min(1,G(u,t))dt$ is bounded away from $0$, i.e., $g_1(u)\geq C_{+}$ for a constant $C_{+}>0$ and any $u \in [0,1]$; 
   (ii) $|\lambda_1|\geq |\lambda_2|\geq\cdots\geq |\lambda_r|>0$, $\operatorname{E}\{\psi_k(U_1)^2\}=1$ and $\operatorname{E}\{\psi_k(U_1)\psi_l(U_1)\}=0$ for $k\ne l$;  (iii) $\psi_k(U_1)$ satisfies the Bernstein condition with parameter $b$, i.e., for $k=2,3,4,\ldots$,
   \[
   |\operatorname{E}\{\psi_k(U_1)-\operatorname{E}\psi_k(U_1)\}^k| \leq \frac{1}{2}k!\Var\{\psi_k(U_1)\}b^{k-2};
   \]
   and (iv) $\liminf \log \rhon^\star/\log N > -1/2$ and $\liminf \log \rhon^\star/\log N < 0$.
\end{assumption}

Assume that the rank $r$ of the graphon is known. Then we use the dominant $r$ eigenvectors of $\bs{E}$ as $\bs{W}_i$ in the regression adjustment.
Then, we use the residual of the above regression to construct $\hat{\tau}_{\ind} $ and $\hat{\tau}_{\tot} $.

\begin{proposition}
\label{prop:graphon-model}
    Under Assumptions \ref{a:bounded-parameter}, \ref{a:bounded-parameter-adjusted}, and \ref{a:regularity-conditons-for-graphon-model}, with probability tending to one, we have, for any $\delta>0$,
    $
    N^{-1}\Deltan \leq C (\log N)^{4+\delta}\rhon^\star$ and $\underline{C}\rhon^\star \leq \rhon\leq \bar{C}\rhon^\star
    $ for constant $C,\underline{C},\bar{C}>0$.
\end{proposition}

Recall that conditional on the graph the EV-adjusted estimators have the magnitude of $\Op(\rhon^{1/2}+N^{-1/2}\Deltan^{1/2})$. Therefore, 
by \Cref{prop:graphon-model}, under such a graphon model, $\hat{\tau}^{\ev}_{\ind} $ and $\hat{\tau}^{\ev}_{\tot} $ improve the convergence rates of $\hat{\tau}_{\ind} $ and $\hat{\tau}_{\tot} $, from $N^{1/2}\rhon^\star$ to $(\rhon^\star)^{1/2}$ (up to a log factor) if we use the top $r$ eigenvectors of $\bs{E}$.

\section{Numerical study}
\label{sec:sim}

\subsection{Numerical experiments}
In this section, we conduct numerical simulations to evaluate the performance of the proposed estimators for DATE, IATE, and GATE, and to verify the improvement in estimation accuracy and inference efficiency of the EV-adjusted estimators. We consider three different scenarios: partial interference, local interference, and interference based on a real-world network. For comparison, we also present the results of methods from existing literature under these scenarios.

Consider a Bernoulli trial with treatment probability $r_1=0.5$. The outcomes are generated by the heterogeneous additive treatment effect model (\ref{eq:HATE-model}), that is,
$$Y_i = \alpha_i + \theta_i Z_i + \sum_{j=1}^N \tilde{E}_{ij}\gamma_{ij} Z_j.$$
The parameters are generated from normal distribution and $t$-distribution. Formally, we generate $\alpha_i \stackrel{\text { i.i.d. }}{\sim} \mathcal{N}(\mu_i,1)$, $\theta_i \stackrel{\text { i.i.d. }}{\sim} 0.8 \cdot t_3(\mu_i,0.5)$ and $\gamma_{i j}^{\prime} \stackrel{\text { i.i.d. }}{\sim}  1.8 \cdot t_3((\mu_i+\mu_j)/2,0.5)$, where $t_3(\mu, \sigma)$ denotes the $t$-distribution with 3 degrees of freedom, location parameter $\mu$ and scaling parameter $\sigma$. To meet Assumption \ref{a:bounded-parameter}, we normalize the indirect effect parameters by $\gamma_{i j}=\gamma_{i j}^{\prime} / \sum_{j=1}^N \tilde{E}_{i j}$. The sample size $N$, the network structure, and the parameters $\mu_i$'s are specified differently in the following three scenarios:

\textbf{Scenario 1} (Partial interference): Consider the network structure described in Section \ref{sec:partial}. The outcomes are generated through a hidden network. Specifically, we generate the observed network $\bs{E}$ by Assumption \ref{a:stratified-interference}, and then independently remove each edge in $\bs{E}$ with a probability of 0.75 to obtain $\tilde{\bs{E}}$. We set the cluster size $N_M$ and the number of clusters $M$ to be $(10,45)$ and $(20,180)$. The numbers of observed neighbors for each unit are 9 and 19, and the average numbers of true neighbors for each unit are 2.25 and 4.75, respectively. The observed network densities $\rhon$ are approximately 0.020 and 0.005. As mentioned in Remark \ref{rmk:stratify-the-group-partial-interference}, suppose that all clusters can be equally divided into 3 strata, where the parameter $\mu_i$ is set to be the same within each stratum, and takes values $0, 1, 2$ for different strata. In this setting, $\mu_i$ actually represents the stratum effect. We use the cluster stratum indicators as the eigenvectors in regression adjustment. As mentioned in Remark \ref{rem:conservative-variance-estimator}, we employ the variance estimators without multiplication by 2 for inference (the same in the other two scenarios).

In partial interference case, a common method to estimate the GATE is to perform cluster randomization, which means a Bernoulli trial at the cluster level. The simple difference in means between treated and control clusters (denoted as $\hat{\tau}_{\tot}^{\textnormal{CL}}=\sumi Z_iY_i / (Nr_1) - \sumi (1-Z_i)Y_i / (Nr_0)$) can be used as the estimator \citep{viviano2023causal}. If we consider clusters as aggregated units and the cluster-averaged outcomes as corresponding outcomes, then this estimator is equivalent to our estimator for the direct effect (defined on aggregated units). Therefore, similar normal approximation and Neyman-type variance estimators can be used for inference. We set the treatment probability of cluster randomization to 0.5 as well. To fit the finite-population framework, we generate the observed network, the hidden network, and the parameters $(\alpha_i, \theta_i, \gamma_{ij})$ once, then separately repeat the Bernoulli trial and cluster randomization 10000 times to evaluate the bias, standard deviation (SD), root mean squared error (RMSE) of the point estimators, and empirical coverage probability (CP) and mean confidence interval length (Length) of $95\%$ confidence intervals.

\begin{table}[ht]
	\centering
	\caption{Simulation results for scenario 1}
	\begin{threeparttable}
		\resizebox{\textwidth}{!}{
			\begin{tabular}{lrrrrrrrrrrrr}
				\toprule
				& \multicolumn{6}{c}{$(N_M,M)=(10,45)$}                     & \multicolumn{6}{c}{$(N_M,M)=(20,180)$} \\
				\cmidrule{2-13}          & \multicolumn{1}{l}{True} & \multicolumn{1}{l}{Bias} & \multicolumn{1}{l}{SD} & \multicolumn{1}{l}{RMSE} & \multicolumn{1}{l}{CP} & \multicolumn{1}{l}{Length} & \multicolumn{1}{l}{True} & \multicolumn{1}{l}{Bias} & \multicolumn{1}{l}{SD} & \multicolumn{1}{l}{RMSE} & \multicolumn{1}{l}{CP} & \multicolumn{1}{l}{Length} \\
				\cmidrule{2-13}    $\hat{\tau}_{\dir} $   & 0.831  & 0.003  & 0.319  & 0.319  & 0.953  & 1.266  & 0.796  & -0.001  & 0.108  & 0.108  & 0.953  & 0.431  \\
				$\hat{\tau}_{\ind} $   & 1.687  & 0.008  & 2.621  & 2.621  & 0.965  & 10.974  & 1.797  & -0.006  & 1.919  & 1.919  & 0.957  & 7.768  \\
				$\hat{\tau}_{\ind} ^{\ev}$ & 1.687  & -0.115  & 0.561  & 0.573  & 0.980  & 2.690  & 1.797  & -0.029  & 0.296  & 0.297  & 0.967  & 1.269  \\
				$\hat{\tau}_{\tot} $   & 2.518  & 0.011  & 2.904  & 2.904  & 0.963  & 12.010  & 2.593  & -0.007  & 2.018  & 2.018  & 0.956  & 8.141  \\
				$\hat{\tau}_{\tot} ^{\ev}$ & 2.518  & -0.112  & 0.672  & 0.681  & 0.954  & 2.793  & 2.593  & -0.030  & 0.326  & 0.327  & 0.952  & 1.298  \\
				$\hat{\tau}_{\tot} ^{\textnormal{CL}}$ & 2.518  & 0.008  & 0.885  & 0.885  & 0.967  & 3.970  & 2.593  & 0.005  & 0.444  & 0.444  & 0.974  & 1.987  \\
				\bottomrule
			\end{tabular}
		}
		\begin{tablenotes}
			\footnotesize
			\item Note: True, true value of estimand; SD, standard deviation; RMSE, root mean square error; CP, \\coverage probability; Length, mean length of confidence interval.
		\end{tablenotes}
	\end{threeparttable}
	\label{tab:partial}
\end{table}

\textbf{Scenario 2} (Local interference): Consider the network structure described in Section \ref{sec:local}. The hidden network generation mechanism is similar to that in Scenario 1, except that the probability of removing edges in the observed network is changed to 0.95. The numbers of buyers and sellers $(N_R,N_C)$ are set as $(20,20)$ and $(60,60)$, to simulate $N_R$ and $N_C$ tending to infinity at the same rate as $N$ increases. The numbers of observed neighbors for each unit are 38 and 118, and the average numbers of true neighbors for each unit are 1.9 and 5.9, respectively. The network densities $\rhon$ are approximately 0.095 and 0.033. Similarly, suppose that both buyers and sellers can be equally divided into 2 strata (totally 4 strata for transaction pairs), where the parameter $\mu_i = 1 + \delta_{i,R} + \delta_{i,C}$, with $\delta_{i, R}$ and $\delta_{i, C}$ sharing the same value for transaction pairs within the same buyer stratum and the same seller stratum, respectively. Let $\delta_{i, R}$, $\delta_{i, C}$ take values in $\{-1,1\}$. In this setting, $\mu_i$ represents a linear combination of the buyer-stratum effect and the seller-stratum effect. We use the buyer stratum indicators and the seller stratum indicators as the eigenvectors in regression adjustment.

We compare our method with the approach proposed by \cite{masoero2024multiple} in this scenario, which advocates for a simple multiple randomization design (SMRD). We choose proper contrast vector to ensure their estimands correspond to the DATE, IATE, and GATE (denote the corresponding estimators as $\hat{\tau}_{\dir} ^{\textnormal{SMRD}}$, $\hat{\tau}_{\ind} ^{\textnormal{SMRD}}$, $\hat{\tau}_{\tot} ^{\textnormal{SMRD}}$) as defined in our paper, and assign half of the buyers and half of the sellers to treatment. Similarly, the generation of networks and parameters is performed once, and the two different designs are repeated 10000 times.

\begin{table}[ht]
	\centering
	\caption{Simulation results for scenario 2}
	\begin{threeparttable}
		\resizebox{\textwidth}{!}{
			\begin{tabular}{lrrrrrrrrrrrr}
				\toprule
				& \multicolumn{6}{c}{$(N_R,N_C)=(20,20)$}                     & \multicolumn{6}{c}{$(N_R,N_C)=(60,60)$} \\
				\cmidrule{2-13}          & \multicolumn{1}{l}{True} & \multicolumn{1}{l}{Bias} & \multicolumn{1}{l}{SD} & \multicolumn{1}{l}{RMSE} & \multicolumn{1}{l}{CP} & \multicolumn{1}{l}{Length} & \multicolumn{1}{l}{True} & \multicolumn{1}{l}{Bias} & \multicolumn{1}{l}{SD} & \multicolumn{1}{l}{RMSE} & \multicolumn{1}{l}{CP} & \multicolumn{1}{l}{Length} \\
				\cmidrule{2-13}    $\hat{\tau}_{\dir}$   & 0.825  & -0.008  & 0.389  & 0.389  & 0.952  & 1.543  & 0.786  & 0.000  & 0.130  & 0.130  & 0.957  & 0.525  \\
				$\hat{\tau}_{\dir} ^{\textnormal{SMRD}}$ & 0.825  & 0.005  & 0.515  & 0.515  & 1.000  & 6.703  & 0.786  & -0.004  & 0.249  & 0.249  & 1.000  & 4.313  \\
				$\hat{\tau}_{\ind}$   & 1.465  & -0.183  & 9.494  & 9.496  & 0.962  & 38.794  & 1.778  & -0.098  & 10.550  & 10.550  & 0.957  & 42.769  \\
				$\hat{\tau}_{\ind}^{\ev}$ & 1.465  & -0.218  & 1.076  & 1.098  & 0.990  & 5.810  & 1.778  & -0.097  & 0.564  & 0.572  & 0.988  & 2.811  \\
				$\hat{\tau}_{\ind} ^{\textnormal{SMRD}}$ & 1.465  & 0.002  & 1.494  & 1.494  & 0.999  & 10.687  & 1.778  & -0.016  & 0.897  & 0.897  & 1.000  & 6.681  \\
				$\hat{\tau}_{\tot}$   & 2.290  & -0.191  & 9.819  & 9.821  & 0.961  & 39.993  & 2.565  & -0.099  & 10.665  & 10.665  & 0.957  & 43.215  \\
				$\hat{\tau}_{\tot}^{\ev}$ & 2.290  & -0.226  & 1.512  & 1.529  & 0.986  & 5.840  & 2.565  & -0.098  & 0.582  & 0.590  & 0.985  & 2.819  \\
				$\hat{\tau}_{\tot} ^{\textnormal{SMRD}}$ & 2.290  & 0.007  & 1.866  & 1.866  & 0.985  & 9.670  & 2.565  & -0.020  & 1.124  & 1.124  & 0.988  & 5.875  \\
				\bottomrule
			\end{tabular}
		}
		\begin{tablenotes}
			\footnotesize
			\item Note: True, true value of estimand; SD, standard deviation; RMSE, root mean square error; CP, \\coverage probability; Length, mean length of confidence interval.
		\end{tablenotes}
	\end{threeparttable}
	\label{tab:twosided}
\end{table}

\textbf{Scenario 3} (A real-world network): Consider a real network generated by email data from a large European research institution. The nodes represent institution members and the edges represent their communication, more specifically, whether one person has sent at least one email to another. Since the act of sending an email is unidirectional, we observe a directed network in this scenario. The network contains 1005 nodes and 25571 edges, with an average number of observed neighbors of approximately 25.4 and an average density $\rhon$ of about 0.025. Similarly, we set a probability of 0.75 to generate the hidden network with an average number of true neighbors of approximately 6.4. Here, we do not consider the stratification structure, so we set all $\mu_i$'s to 1.

An important hyperparameter in the EV-adjusted estimators is the number of eigenvectors used for regression adjustment. We plot the eigenvalue decay of the observed network (see figure \ref{fig:ev-plot}) and find that the eigenvalues decrease rapidly for the largest 10. In detail, the largest eigenvalue is 4098, and the 5th and 10th largest eigenvalues are 634 and 355, respectively. To illustrate the differences, we present the results of $\hat{\tau}_{\ind}^{\ev}$ and $\hat{\tau}_{\tot}^{\ev}$ for varying numbers of eigenvectors ($\hat{\tau}_{\dir}$ exhibits similar performance to the other two scenarios, so we omit it). Let $K$ denote the number of eigenvectors used for adjustment. In particular, we use $K=0$ to indicate adjusting intercept only and $K=-1$ to indicate no adjustment.

\begin{figure}
    \centering
    \includegraphics[width=0.9\textwidth]{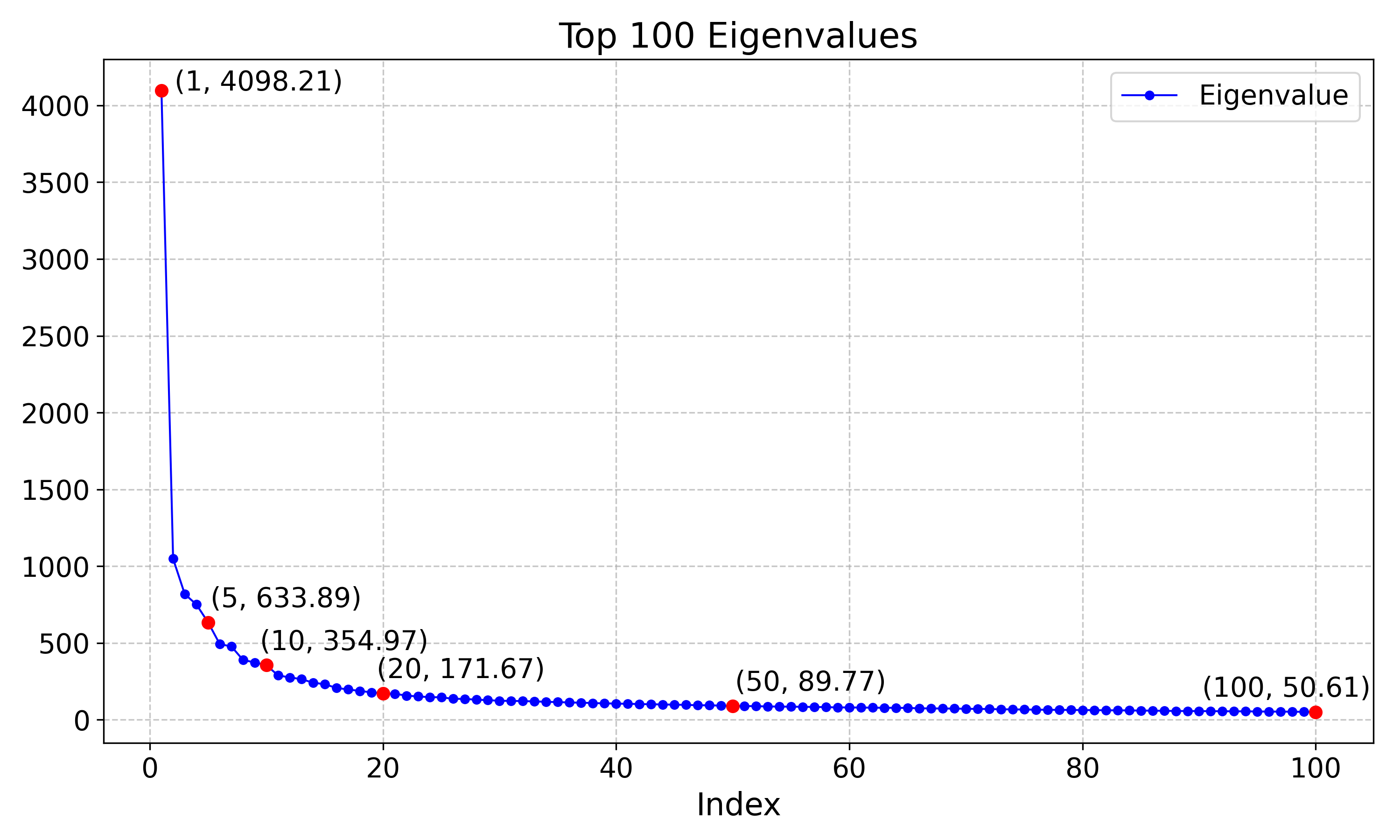}
    \caption{Eigenvalues of the email network}
    \label{fig:ev-plot}
\end{figure}

\begin{table}[ht]
	\centering
	\caption{Simulation results for scenario 3}
	\begin{threeparttable}
		\resizebox{\textwidth}{!}{
			\begin{tabular}{lrrrrrrrrrrrr}
				\toprule
				& \multicolumn{6}{c}{$\hat{\tau}_{\ind}^{\ev}$}                      & \multicolumn{6}{c}{$\hat{\tau}_{\tot}^{\ev}$} \\
				\cmidrule{2-13}          & \multicolumn{1}{l}{True} & \multicolumn{1}{l}{Bias} & \multicolumn{1}{l}{SD} & \multicolumn{1}{l}{RMSE} & \multicolumn{1}{l}{CP} & \multicolumn{1}{l}{Length} & \multicolumn{1}{l}{True} & \multicolumn{1}{l}{Bias} & \multicolumn{1}{l}{SD} & \multicolumn{1}{l}{RMSE} & \multicolumn{1}{l}{CP} & \multicolumn{1}{l}{Length} \\
				\cmidrule{2-13}    $K=-1$  & 1.293  & -0.081  & 5.272  & 5.273  & 0.958  & 21.206  & 2.110  & -0.081  & 5.366  & 5.367  & 0.958  & 21.567  \\
				$K=0$   & 1.293  & -0.077  & 0.606  & 0.611  & 0.974  & 2.571  & 2.110  & -0.077  & 0.697  & 0.701  & 0.954  & 2.792  \\
				$K=1$   & 1.293  & -0.111  & 0.386  & 0.402  & 0.986  & 2.003  & 2.110  & -0.111  & 0.433  & 0.447  & 0.976  & 2.039  \\
				$K=5$   & 1.293  & -0.252  & 0.344  & 0.426  & 0.969  & 1.824  & 2.110  & -0.252  & 0.387  & 0.462  & 0.953  & 1.856  \\
				$K=10$  & 1.293  & -0.349  & 0.297  & 0.458  & 0.939  & 1.640  & 2.110  & -0.349  & 0.337  & 0.485  & 0.922  & 1.666  \\
				\bottomrule
			\end{tabular}
		}
		\begin{tablenotes}
			\footnotesize
			\item Note: True, true value of estimand; SD, standard deviation; RMSE, root mean square error; CP, \\coverage probability; Length, mean length of confidence interval.
		\end{tablenotes}
	\end{threeparttable}
	\label{tab:email}
\end{table}

Tables \ref{tab:partial}-\ref{tab:email} show the results for the three scenarios. Overall, the biases of the proposed estimators are much smaller compared to the standard deviations in all scenarios, except when a considerably large number of eigenvectors are used for adjustment. Meanwhile, we observe that the ratio of bias to standard deviation decreases as the sample size increases in scenarios 1 and 2, which confirms our theoretical conclusion that bias can be negligible if the sample size is sufficiently large. Moreover, the EV-adjusted estimators exhibit significantly smaller standard deviations and shorter confidence intervals compared to the original ones. This demonstrates a substantial improvement in estimation precision. In addition, even though the average degrees of the hidden networks are relatively small, all empirical coverage probabilities of variance estimators without multiplication by 2 are still around or over 95\%, implying that the Type I error rates are well controlled. The variance estimators of the EV-adjusted estimators seem to be conservative in some cases but still have sufficient power to reject the null hypothesis of no effects when the sample size is large. 

In more detail, in scenarios 1 and 2, our method exhibits lower standard deviations compared to existing methods, along with narrower confidence intervals. This confirms the conclusion in Corollaries \ref{cor:variance-estimator-adj-stratified2} and \ref{cor:local-interference-for-marketplace-adj} that $\hat{\tau}_{\ind} ^{\ev}$ and $\hat{\tau}_{\tot} ^{\ev}$ achieve the same convergence rates as the existing estimators, with actually higher precision and greater efficiency. One interesting phenomenon observed in additional simulations is that both the cluster randomization method in scenario 1 and the SMRD method in scenario 2 are more sensitive to the setting of $\mu_i$. The standard deviations increase rapidly as the variation of $\mu_i$ becomes larger. Therefore, if researchers have prior knowledge that the population consists of several groups with significant differences, our method will work much better than existing methods.

For a real-world network, especially irregular ones like in scenario 3, we find that bias may be non-negligible when we use too many eigenvectors for adjustment. In fact, as we theorized, bias compared to the standard deviation is only off by an order of $\rhon^{1/2}$, which can be not significantly small for networks that are not too sparse in finite samples. Figure \ref{fig:ev-estimators-plot} show the results of $\hat{\tau}_{\ind} ^{\ev}$ and $\hat{\tau}_{\tot} ^{\ev}$ with different numbers of eigenvectors for adjustment. Although the standard deviation decreases with increasing number of eigenvectors, the absolute bias increases. This is because the largest eigenvalue accounts for a substantial portion of the network information, but subsequent eigenvalues contribute minimally (as shown in Figure \ref{fig:ev-plot}), which results in using additional eigenvectors for adjustment not significantly reducing the mean squared error. Therefore, for practical networks, we recommend using a limited number of eigenvectors for adjustment, unless there is a clear natural stratified structure.

\begin{figure}
    \centering
    \includegraphics[width=0.9\textwidth]{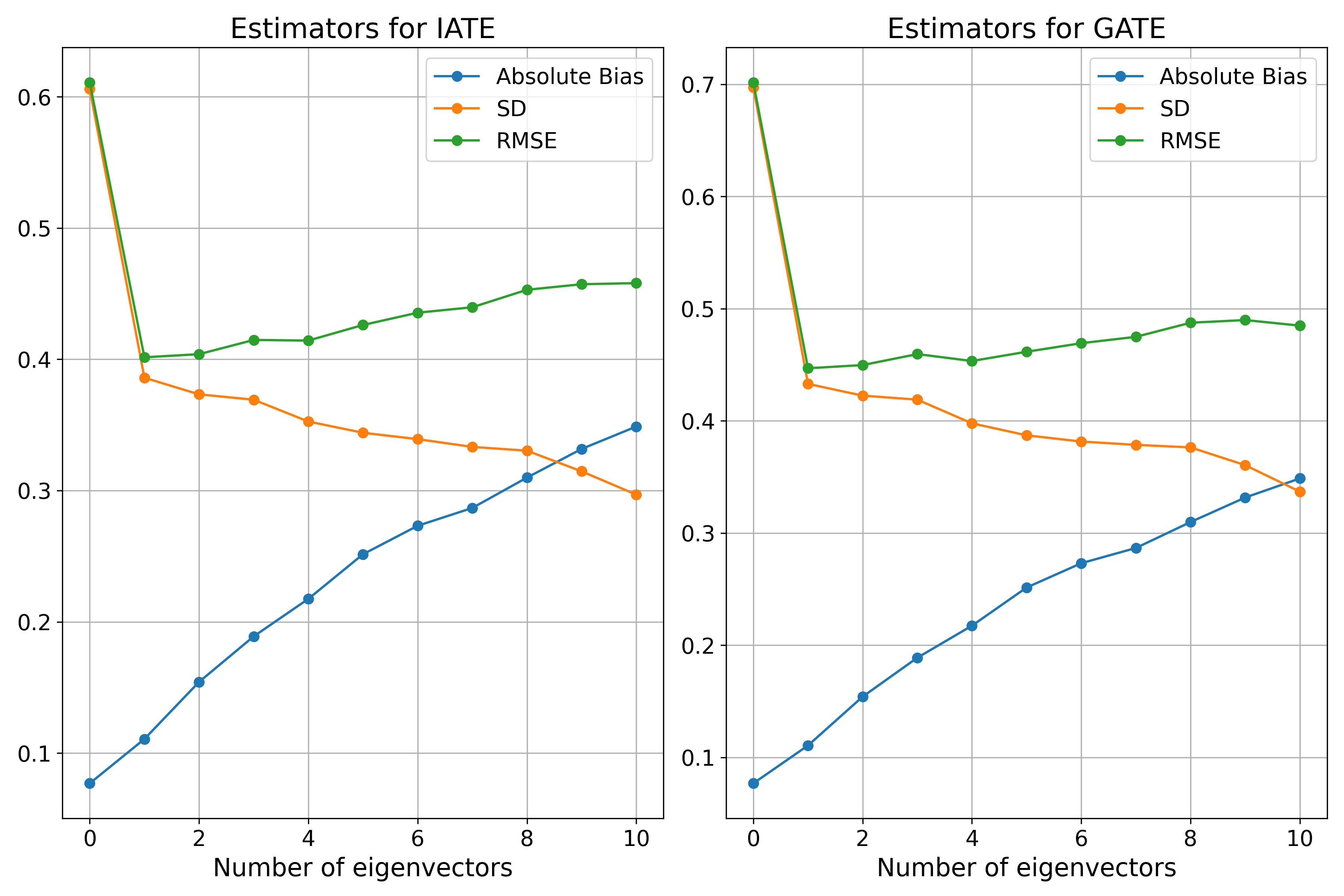}
    \caption{Eigenvector-adjusted estimators with different numbers of adjusted eigenvectors}
    \label{fig:ev-estimators-plot}
\end{figure}

\subsection{Real data analysis}
An experiment on the effect of anti-conflict intervention was conducted in the 2012-2013 school year across 56 New Jersey public middle schools, with 24191 students \citep{paluck2016changing}. The primary objective of the experiment was to investigate how altering social norms within school environments can reduce conflict and bullying behaviors. The social network was measured by asking students which 10 students they chose to spend time with in the last few weeks. We set $E_{ij} = 1$ if student $i$ listed student $j$ on their friendship questionnaire, and thus, we observe a directed graph. In this experiment, 28 of 56 schools were randomly assigned to receive the intervention in the first stage. The researchers selected 64 seed-eligible students using a deterministic algorithm in each treated school, and then randomly assigned half of them to be invited to participate in the anti-conflict program.

We focus on the five largest treated schools and consider the experiment to approximate a Bernoulli trial on the seed-eligible students, as in \cite{leung2022causal} and \cite{gao2023causal}. We choose whether or not students wore the anti-conflict wristband, which is the reward for good behavior, as the outcome of interest. After removing the missing data, we finally obtain a dataset of 265 students with 144 assigned to the treatment group. The average degree of network is 0.544. Although it is relatively small, based on our simulated data findings, we have reason to believe that the variance estimators without multiplication by 2 can be used for inference. Since the inherent stratified structure of students, we simply use the school indicators for eigenvector adjustment. 

Table \ref{tab:real-school} presents the results. As we can see, the direct effect and the total effect exhibit significance, which indicates that anti-conflict interventions, whether on individual students or all students, contribute to raising students' awareness. This conclusion coincides with the findings in \cite{aronow2017estimating}. The indirect effect is not statistically significant, but still positive, which means the primary effect on a student originates from the anti-conflict intervention applied to himself. One noteworthy observation is that despite the average number of neighbors being modest, the adjusted estimators for IATE and GATE have narrower confidence intervals, approximately 39.4\% and 38.9\% shorter than those of the unadjusted estimators, respectively. For efficiency considerations, we prefer the EV-adjusted estimators.

\begin{table}
\caption{\centering \label{tab:real-school} Results for anti-conflict intervention}
  \centering  
  \begin{threeparttable}
    \begin{tabular}{lccc}    
    \toprule
          & \multicolumn{1}{c}{Estimator} & \multicolumn{1}{c}{95\% CI} & \multicolumn{1}{c}{Length} \\
    \cmidrule{2-4}
    $\hat{\tau}_{\dir} $ & 0.215  & (0.054, 0.377) & 0.321  \\
    $\hat{\tau}_{\ind} $ & 0.062  & (-0.133, 0.256) & 0.388  \\
    $\hat{\tau}_{\ind} ^{\ev}$ & 0.036  & (-0.080, 0.153) & 0.235  \\
    $\hat{\tau}_{\tot} $ & 0.277  & (-0.015, 0.569) & 0.584  \\
    $\hat{\tau}_{\tot} ^{\ev}$ & 0.252  & (0.074, 0.429) & 0.357  \\
    \bottomrule
    \end{tabular}
    \begin{tablenotes}[flushleft]
      \footnotesize
      \item Note: $95\%$ CI, 95\% confidence interval; length, length of the confidence interval.
    \end{tablenotes}
  \end{threeparttable}
\end{table}

\section{Conclusion}
\label{sec:con}

We provide a comprehensive framework for analyzing the DATE, IATE, and GATE in a Bernoulli trial under the HATE model. Our framework is built upon the design-based inference framework and the conclusions are drawn without making additional assumptions on the data-generating process of the network and outcome-related parameters. It accommodates complex real-world network structures, allowing the average node degree to tend to infinity. Our framework exhibits the ingredient of a hidden network, permitting the existence of sparse network interference where among many neighbors of a unit, only a few have an indirect effect on that unit. Moreover, we can significantly improve the estimation and inference efficiencies by conducting the eigenvector adjustment on the IATE and GATE, which exploits the clustering structure of the network.

This work suggests several open questions for further exploration. Firstly, in practice, baseline covariates can potentially improve estimation precision, such as the individual-level and cluster-level covariates for clustered data as well as buyer-level and seller-level covariates in marketplace data. It would be intriguing to explore methods for conducting covariate adjustment to improve the estimation and inference. Secondly, many networks feature weighted edges that measure the strength of the connection between nodes. Exploring how to utilize this information would be valuable. Lastly, there are nonstandard Bernoulli trials where the units receiving treatment differ from those whose outcomaes are measured, for example, the bipartite experimental designs. Extending the framework to incorporate these designs would be worthwhile.

\bibliographystyle{agsm}
\bibliography{causal}

\newpage

\appendix

\centerline{ \Large\bf SUPPLEMENTARY MATERIAL}
\vspace{2mm}

\spacingset{1.5}

Section~\ref{sec:A} introduces the notation used throughout the supplementary material.

Section~\ref{sec:B} presents proofs about the unbiasedness and variances of the treatment effect estimators, as detailed in Propositions~\ref{prop:E-and-Var-DIR}--\ref{prop:order-of-adjusted-estimators}.

Section~\ref{sec:C} details the proofs supporting the results on asymptotic normality outlined in Theorem~\ref{thm:CLT-no-adjust}.

Section~\ref{sec:D} covers the proofs validating the consistency of variance estimators, as per Theorems~\ref{thm:variance-estimator-dir}--\ref{thm:variance-estimator-tot}. 

Section~\ref{sec:E} provides proofs for the asymptotic properties of eigenvector-adjusted estimators, presented in Theorems~\ref{thm:CLT-EV-adj}--\ref{thm:asymptotic-variance-estimator-adjusted}.

Section~\ref{sec:F} elaborates on the proofs for the results discussed in \Cref{sec:application-in-some-networks} of the main text, specifically in Corollaries~\ref{cor:variance-estimator-adj-stratified}--\ref{cor:local-interference-for-marketplace-adj} and Proposition~\ref{prop:graphon-model}.

\appendix

\section{Notation}\label{sec:A}
If the summation or maximization symbol is not annotated with a range, it is understood by default to range from $1$ to $N$. For example, we write 
$\sum_{i=1}^N$ and $\max_{i=1}^N$ as $\sum_i$ and $\max_i$ for short. We write $\sum_{i_1=1}^N\cdots\sum_{i_l=1}^N$ as $\sum_{i_1,\ldots,i_l}$ for short. We denote $\sum_{i_1\ne \cdots \ne i_l}$ as the summation over all tuples of $(i_1,\ldots,i_l)$ with mutually different $1\leq i_1,\ldots,i_l\leq N$. We write $a_n\lesssim b_n$ if there exists a constant $C$ independent $N$ such that $|a_n|\leq C b_n$ when $n$ is large enough. We write $\Var(\hat{\tau}_\star) = \sigma^2_{\star}$ and $\Var(\tilde{\tau}^{\ev}_\star) = (\sigma^{\ev}_{\star})^2$, for $\star \in \{\dir,\ind,\tot\}$. Let $\Cov(\cdot)$ denote the covariance. We use $\E$ to denote the expectation and $E_{ij}$ to denote the $(ij)$th element of the adjacency matrix of the network.

\section{Unbiasedness and variance}\label{sec:B}

\subsection{Proof of \Cref{prop:E-and-Var-DIR}}

\begin{proof}

	Recall that
	\begin{align*}
		\hat{\tau}_{\dir}  = \frac{1}{N}\sum_i  \Bigl\{\frac{Y_iZ_i}{r_1} - \frac{Y_i(1-Z_i)}{r_0}\Bigr\}, 
	\end{align*}
	$$
	Y_i = \alpha_i + \theta_i Z_i + \sum_{j=1}^N \tilde{\gamma}_{ij} Z_j.
	$$
	Thus,
	$$
	Y_{Z_i=z} = E(Y_i \mid Z_i  = z) = \alpha_i + \theta_i z + \sum_{j=1}^N \tilde{\gamma}_{ij} r_1.
	$$
	Therefore,
	\[
	\E \Bigl\{\frac{Y_iZ_i}{r_1} - \frac{Y_i(1-Z_i)}{r_0}\Bigr\} = Y_{Z_i=1}-Y_{Z_i=0} = \theta_i,
	\]
	which implies that $\E \hat{\tau}_{\dir} = N^{-1} \sum_{i} \theta_i =  \tau_{\dir}$.

	To derive the formula of $\Var(\hat{\tau}_{\dir})$, we see that 
	\begin{align*}
		Y_i(Z_i-r_1) =& \alpha_i(Z_i-r_1) + \theta_i Z_i(Z_i-r_1) + \sum_j  \tilde{\gamma}_{ij} Z_j(Z_i-r_1)\\
		=&\theta_i r_0r_1 + (\alpha_i+ r_0\theta_i+\sum_j  r_1\tilde{\gamma}_{ij})(Z_i-r_1) + \sum_j  \tilde{\gamma}_{ij} (Z_j-r_1)(Z_i-r_1).
	\end{align*}
	As a consequence, we have
	\begin{align}
		\hat{\tau}_{\dir}-\tau_{\dir} = &\sum_{i}  \frac{Y_i(Z_i-r_1)}{Nr_1r_0}-\tau_{\dir} \nonumber \\
		=& \sum_i  \frac{(Z_i-r_1)}{Nr_1r_0} (\alpha_i+r_0\theta_i+\sum_j   r_1\tilde{\gamma}_{ij}) + \sum_{i,j}  \tilde{\gamma}_{ij}\frac{(Z_i-r_1)(Z_j-r_1)}{Nr_1r_0} \nonumber \\
		= :& B_{\dir,1} + B_{\dir,2}. \label{eq:decompose-direct-effect}
	\end{align}
	Note that 
	\begin{align*}
		& r_1Y_{Z_i=0}+r_0Y_{Z_i=1} = \alpha_i+r_0\theta_i+\sum_j   r_1\tilde{\gamma}_{ij},\quad \Cov(B_{\dir,1},B_{\dir,2}) = 0,\\
		&\Var(Z_i-r_1) = r_1r_0,\quad \Var\{(Z_i-r_1)(Z_j-r_1)\} = (r_1r_0)^2, \quad i \neq j.
	\end{align*}
	Then, we have 
	\begin{align*}
		\Var(\hat{\tau}_{\dir}) = \Var(B_{\dir,1}) + \Var(B_{\dir,2}) = \sigma_{\dir,1}^2 + \sigma_{\dir,2}^2.
	\end{align*}
	
\end{proof}

\subsection{Proof of \Cref{rmk:some-sufficient-condition-for-a-3}}
We show the following proposition mentioned in \Cref{rmk:some-sufficient-condition-for-a-3}.
\begin{proposition}
	\label{prop:rmk-2-proposition}
	When $\bs{E}=\bs{E}^\top$ and $\tilde{\bs{E}}=\tilde{\bs{E}}^\top$, (i) if there exists a constant $C_{+}>0$, such that $\maxi N_i \leq C_+ N\rhon$, then $\|\bs{E}\|_{\oprtnorm} = O(N\rhon)$; (ii) if $\max_i \tilde{N}_i / \min_i \tilde{N}_i \leq C_+$, then $\|\bs{Q}\|_{\oprtnorm} = O(1)$.
\end{proposition}
Before proving \Cref{prop:rmk-2-proposition}, we need \Cref{lem:circle-theorem} below, coming from Theorem 0 of \cite{bell1965gershgorin}.

\begin{lemma}[Gershgorin Circle Theorem]
	\label{lem:circle-theorem}
	Let $\bs{A} = (A_{ij})$ be a complex $n \times n$ matrix. For $i \in\{1, \ldots, N\}$, let $R_i$ be the sum of the absolute values of the nondiagonal entries in the $i$th row:
	$$
	R_i=\sum_{j: j \neq i}\left|A_{i j}\right|.
	$$
	Let $D\left(A_{i i}, R_i\right) \subseteq \mathbb{C}$ be a closed disc centered at $A_{i i}$ with radius $R_i$. Then, every eigenvalue of $\bs{A}$ lies within at least one of the Gershgorin discs $D\left(A_{i i}, R_i\right)$, $i=1,\ldots,N$.
\end{lemma}

\begin{proof}[Proof of \Cref{prop:rmk-2-proposition}]
	For the first part of \Cref{prop:rmk-2-proposition}, when $\bs{E} = \bs{E}^\top$, we have $\|\bs{E}\|_{\oprtnorm}$ equals to the maximum of absolute eigenvalues of $\bs{E}$. Thus, by \Cref{lem:circle-theorem}, we have
	\begin{align*}
		\|\bs{E}\|_{\oprtnorm} \leq \max_{i} \sum_{j} E_{ij} = \max_i N_i  = O(N\rhon).
	\end{align*}
	For the second part of \Cref{prop:rmk-2-proposition}, when $\tilde{\bs{E}} = \tilde{\bs{E}}^\top$ and $\max_i \tilde{N}_i / \min_i \tilde{N}_i \leq C_+$, we have $\|\diag(\tilde{N}_i^{-1})\|_{\oprtnorm} = O((N\rhon)^{-1})$. Again by \Cref{lem:circle-theorem}, we have
	\begin{align*}
		\|\bs{Q}\|_{\oprtnorm} \leq \|\tilde{\bs{E}}\|_{\oprtnorm}\|\diag(\tilde{N}_i^{-1})\|_{\oprtnorm} \leq \max_i \tilde{N}_i \cdot \big(\min_i \tilde{N}_i\big)^{-1}\leq C_+ = O(1).
	\end{align*}
\end{proof}

\subsection{Proof of \Cref{prop:order-of-tau-dir}}

\begin{proof}

	By Assumption \ref{a:bounded-parameter}, $\max_i\max_{z\in\{0,1\}^N}|Y_i(z)|$ has an upper bound $C_Y$. Consequently, we have
	$$ \Var(B_{\dir,1}) = \frac{1}{N^2r_1r_0} \sum_i   (r_0Y_{Z_i=1}+r_1Y_{Z_i=0})^2\leq \frac{1}{N^2r_1r_0} \sum_i   C_Y^2 = O(N^{-1}).$$
	Recall that ${\bs{Q}} = (Q_{ij})_{1 \leq i,j\leq N}$ is the normalized adjacency matrix of the latent network with $Q_{ij}= \tilde{E}_{ij}/\tilde{N}_i$ if $\tilde E_{ij}=1$ and $0$ otherwise. Under Assumption~\ref{a:bounded-parameter}, we have $\tilde{\gamma}_{ij}\leq CQ_{ij}$. By Assumption \ref{a:opnorm-EE^T} and $\tilde{\gamma}_{ii} = 0$, we have
	\begin{align*}
		\Var(B_{\dir,2}) &= \frac{1}{N^2}\sum_{i,j}  \tilde{\gamma}_{ij}^2 + \frac{1}{N^2}  \sum_{i,j}  \tilde{\gamma}_{ij}\tilde{\gamma}_{ji}\\
		&\leq \frac{2}{N^2}\sum_{i,j}  \tilde{\gamma}_{ij}^2 \lesssim \frac{1}{N^2}\bs{1}^\top \bs{Q}\bs{Q}^\top \bs{1}\\
		&\leq N^{-1}\|\bs{Q}\|_{\oprtnorm}^2 = O(N^{-1}).
	\end{align*}
	
\end{proof}

\subsection{Proof of \Cref{prop:order-and-formula-of-tau-ind}}

\begin{proof}
	Recall that
	\begin{align*}
		\hat{\tau}_{\ind} = \frac{1}{N}\sum_{i,j} E_{ij} \Bigl\{\frac{Y_iZ_j}{r_1} - \frac{Y_i(1-Z_j)}{r_0}\Bigr\}.
	\end{align*}
	Simple calculation gives 
	\begin{align*}
		&Y_i(Z_j-r_{1}) \\
		=& \alpha_i (Z_j-r_{1})+ \theta_i Z_i (Z_j-r_{1})+ \tilde{\gamma}_{ij}Z_j(Z_j-r_{1}) + \sum_{k:k\ne j}\tilde{\gamma}_{ik}Z_k(Z_j-r_{1})\\
		=&\alpha_i (Z_j-r_{1}) + r_{1} \theta_i (Z_j-r_{1}) + \theta_i (Z_i-r_{1}) (Z_j-r_{1})+ \tilde{\gamma}_{ij}Z_j r_{0} + \sum_{k:k\ne j}\tilde{\gamma}_{ik}Z_k(Z_j-r_{1})\\
		=&r_{1}r_{0} \tilde{\gamma}_{ij}+(\alpha_i  + r_{1} \theta_i + \tilde{\gamma}_{ij} r_{0} +  r_{1}\sum_{k:k\ne j}\tilde{\gamma}_{ik})(Z_j-r_{1}) +  \\
		&\sum_{k:k\ne j}\tilde{\gamma}_{ik}(Z_k-r_{1})(Z_j-r_{1}) +\theta_i (Z_i-r_{1}) (Z_j-r_{1}).
	\end{align*}
	Therefore,
	\begin{align}
		&\hat{\tau}_{\ind}-\tau_{\ind} = \sum_{i,j} \frac{E_{ij}Y_i(Z_j-r_{1})}{Nr_{1}r_{0}}-\tau_{\ind} \nonumber\\
		= & \sum_{i,j}\frac{E_{ij}}{Nr_{1}r_{0}} (\alpha_i  + r_{1} \theta_i + \tilde{\gamma}_{ij} r_{0} +  r_{1}\sum_{k:k\ne j}\tilde{\gamma}_{ik})(Z_j-r_1) + \nonumber\\
		&\qquad \qquad \sum_{i,j}\frac{E_{ij}}{Nr_{1}r_{0}}\sum_{k:k\ne j}\tilde{\gamma}_{ik}(Z_k-r_{1})(Z_j-r_{1}) + \sum_{i,j} \frac{E_{ij}}{Nr_{1}r_{0}} \theta_i (Z_i-r_{1}) (Z_j-r_{1}) \nonumber\\
		=& \sum_{i,j}\frac{E_{ij}}{Nr_{1}r_{0}} (\alpha_i  + r_{1} \theta_i + \tilde{\gamma}_{ij} r_{0} +  r_{1}\sum_{k:k\ne j}\tilde{\gamma}_{ik})(Z_j-r_1) +  \nonumber \\
		&\qquad \qquad\sum_{j,k}\frac{E_{kj}}{Nr_{1}r_{0}}\sum_{i:i\ne j}\tilde{\gamma}_{ki}(Z_i-r_1)(Z_j-r_1) + \sum_{i,j} \frac{E_{ij}}{Nr_{1}r_{0}} \theta_i (Z_i-r_1) (Z_j-r_1) \nonumber\\
		= & \sum_{i,j}\frac{E_{ij}}{Nr_{1}r_{0}} (\alpha_i  + r_{1} \theta_i + \tilde{\gamma}_{ij} r_{0} +  r_{1}\sum_{k:k\ne j}\tilde{\gamma}_{ik})(Z_j-r_1) + \nonumber \\
		&\qquad \qquad \qquad \qquad\frac{1}{Nr_{1}r_{0}} \sum_{i\ne j}\big(E_{ij}\theta_i+\sum_{k}E_{kj}\tilde{\gamma}_{ki}\big)(Z_i-r_1)(Z_j-r_1) \nonumber \\
		=: & B_{\ind,1} + B_{\ind,2} \label{eq:decompose-indirect-effect}.
	\end{align}
	
	Noting that $Y_{Z_i=z_i}^{Z_j=z_j} = \operatorname{E}(Y_i|Z_i=z_i,Z_j=z_j)$ and
	\begin{equation}\label{eqn:equality1}
		\alpha_i  + r_{1} \theta_i + r_{0}\tilde{\gamma}_{ij} +  r_{1}\sum_{k:k\ne j}\tilde{\gamma}_{ik} = r_1r_0 Y_{Z_i=1}^{Z_j=1} +r_0^2Y_{Z_i=0}^{Z_j=1}+r_1^2Y_{Z_i=1}^{Z_j=0}+r_1r_0Y_{Z_i=0}^{Z_j=0},
	\end{equation}
	for $i\neq j$, then we have $\Var(\hat{\tau}_{\ind}) =\Var(B_{\ind,1})+\Var(B_{\ind,2}) = \sigma^2_{\ind,1}+\sigma^2_{\ind,2}.$ By \Cref{a:bounded-parameter} and \Cref{a:opnorm-EE^T}, we have
	$$
	\begin{aligned}
		\sigma^2_{\ind,1} &= \frac{1}{N^2r_1r_0}\sum_i  \Bigl\{\sum_j   E_{ji}(r_1r_0 Y_{Z_j=1}^{Z_i=1} +r_0^2Y_{Z_j=0}^{Z_i=1}+r_1^2Y_{Z_j=1}^{Z_i=0}+r_1r_0Y_{Z_j=0}^{Z_i=0})\Bigr\}^2 \\
		& \lesssim \frac{1}{N^2r_1r_0}\sum_i  \Big(\sum_j   E_{ji}\Big)^2 \lesssim N^{-2} \bs{1}^\top\bs{E}\bs{E}^\top\bs{1} \lesssim N^{-1} \|\bs{E}\|_{\oprtnorm}^2 = O(N\rhon^2),
	\end{aligned}$$
	and
	$$
	\begin{aligned}
		\sigma^2_{\ind,2} &\leq \frac{2}{N^2}   \sum_{j \ne i}(E_{ij}\theta_{i}+\sum_{k}E_{kj}\tilde{\gamma}_{ki})^2 \lesssim \frac{2}{N^2}     \sum_{j \ne i}(E_{ij}\theta_{i})^2 + \frac{2}{N^2}  \sum_{j \ne i}(\sum_{k}E_{kj}\tilde{\gamma}_{ki})^2 \\
		&\lesssim \frac{1}{N^2}    \sum_{j \ne i}E_{ij} +  \frac{2}{N^2}  \sum_{j \ne i}(\sum_{k}E_{kj}Q_{ki})^2 = \rhon + \frac{1}{N^2}  \sum_{j \ne i}(\sum_{k}E_{kj}Q_{ki})^2.
	\end{aligned}
	$$
	On the other hand, we have 
	\begin{align*}
		& \sum_{j \ne i}(\sum_{k}E_{kj}Q_{ki})^2\leq  \sum_{i}  \sum_{j} \sum_{k} \sum_{k^\prime} E_{kj}Q_{ki}E_{k^\prime j}Q_{k^\prime i} \leq \sum_{i} \sum_{j} \sum_{k} \sum_{k^\prime} Q_{ki} E_{k^\prime j}Q_{k^\prime i}\\
		= &\bs{1}^\top \bs{Q}\bs{Q}^\top \bs{E} \bs{1} \leq N \|\bs{Q}\|^2_{\oprtnorm}\|\bs{E}\|_{\oprtnorm} = O(N^2\rhon).
	\end{align*}
	Putting together the pieces, we have $ \sigma^2_{\ind,2} = O(\rhon).$
	
\end{proof}

\subsection{Proof of \Cref{prop:order-of-tau-tot}}

\begin{proof}
	It is easy to see that $E(\hat{\tau}_{\tot}) = E( \hat{\tau}_{\dir}) + E(\hat{\tau}_{\ind}) = \tau_{\dir} + \tau_{\ind} =\tau_{\tot}$. Moreover, by the proof of Propositions \ref{prop:E-and-Var-DIR}--\ref{prop:order-and-formula-of-tau-ind}, we can decompose $\hat{\tau}_{\tot}$ as $\hat{\tau}_{\tot} = B_{\tot,1} + B_{\tot,2}$, where $B_{\tot,1} = B_{\dir,1} + B_{\ind,1}$ and $B_{\tot,2} = B_{\dir,2} + B_{\ind,2}$. Simple calculation gives $\sigma_{\tot,1}^2 = \Var(B_{\tot,1})$, $\sigma_{\tot,2}^2 = \Var(B_{\tot,2})$, and $\Cov(B_{\tot,1}, B_{\tot,2})=0$. Therefore, $\Var(\hat{\tau}_{\tot})= \sigma_{\tot,1}^2 + \sigma_{\tot,1}^2$.

	By Assumptions~\ref{a:bounded-parameter}--\ref{a:opnorm-EE^T} and Cauchy--Schwarz inequality, we have
	\begin{align*}
		&\sigma_{\tot,1}^2 \lesssim \sigma_{\dir,1}^2 + \sigma_{\ind,1}^2 = O(N\rhon^2 + N^{-1}) = O(N\rhon^2);\\
		&\sigma_{\tot,2}^2 \lesssim \sigma_{\dir,2}^2 + \sigma_{\ind,2}^2 = O(\rhon + N^{-1}) = O(\rhon).
	\end{align*}

\end{proof}

\subsection{Proof of \Cref{prop:order-of-adjusted-estimators}}

\begin{proof}
	By replacing  $Y_{Z_j=z}^{Z_i=z^\prime}$, $(z,z^\prime)\in \{0,1\}^2$, with $e_{Z_j=z}^{Z_i=z^\prime}$ in $\sigma^2_{\ind,1}$ and $\sigma^2_{\tot,1}$, and  replacing  $\theta_i$ with $\theta_{\bsw,i}$ in $\sigma^2_{\ind,2}$ and $\sigma^2_{\tot,2}$, we obtain that
	\begin{align*}
		&\Var(\tilde{\tau}_{\ind}^{\ev}) =(\sigma_{\ind,1}^{\ev})^2+(\sigma_{\tot,2}^{\ev})^2 ,\quad\Var(\tilde{\tau}_{\tot}^{\ev}) = (\sigma_{\tot,1}^{\ev})^2 + (\sigma_{\tot,2}^{\ev})^2, 
	\end{align*}
	where
	\begin{align*}
		(\sigma^{\ev}_{\ind,1})^2  =&  \frac{1}{N^2r_1r_0}\sum_i \Bigl\{\sum_j   E_{ji}(r_1r_0 e_{Z_j=1}^{Z_i=1} +r_0^2e_{Z_j=0}^{Z_i=1}+r_1^2e_{Z_j=1}^{Z_i=0}+r_1r_0e_{Z_j=0}^{Z_i=0})\Bigr\}^2, \\
		(\sigma_{\ind,2}^{\ev})^2=& \frac{1}{N^2}  \sum_{ i\ne j } \Big(E_{ij}\theta_{\backslash \bs{W},i}+\sum_k  E_{kj}\tilde{\gamma}_{ki}\Big)^2 + \\
		&\frac{1}{N^2}  \sum_{ i\ne j }  \Big(E_{ji}\theta_{\backslash \bs{W},j}+\sum_k  E_{ki}\tilde{\gamma}_{kj}\Big)\Big(E_{ij}\theta_{\backslash \bs{W},i}+\sum_k  E_{kj}\tilde{\gamma}_{ki}\Big),\\
		(\sigma_{\tot,1}^{\ev})^2 =&  \frac{1}{N^2r_1r_0}\sum_i \Bigl\{r_0Y_{Z_i=1}+r_1Y_{Z_i=0}+ \\
		& \quad \sum_j   E_{ji}(r_1r_0 e_{Z_j=1}^{Z_i=1} +r_0^2e_{Z_j=0}^{Z_i=1}+r_1^2e_{Z_j=1}^{Z_i=0}+r_1r_0e_{Z_j=0}^{Z_i=0})\Bigr\}^2,\\
		(\sigma_{\tot,2}^{\ev})^2  =&  \frac{1}{N^2}  \sum_{ i\ne j }  \Big(\tilde{\gamma}_{ij}+E_{ji}\theta_{\backslash \bs{W},j}+\sum_k  E_{ki}\tilde{\gamma}_{kj}\Big)\Big(\tilde{\gamma}_{ji}+E_{ij}\theta_{\backslash \bs{W},i}+\sum_k  E_{kj}\tilde{\gamma}_{ki}\Big) + \\
		& \frac{1}{N^2}  \sum_{ i\ne j } \Big(\tilde{\gamma}_{ji}+E_{ij}\theta_{\backslash \bs{W},i}+\sum_k  E_{kj}\tilde{\gamma}_{ki}\Big)^2.
	\end{align*}
	Recall that $N^{-1}\sum_{i=1}^{N} \bs{W}_i\bs{W}_i^\top = \bs{I}_{K+1}$ and
	\begin{align*}
		e_i =(\alpha_i-\bs{\beta}_{0,\ora}^\top \bs{W}_i) + (\theta_i-\bs{\beta}_{1,\ora}^\top \bs{W}_i+\bs{\beta}_{0,\ora}^\top \bs{W}_i) Z_i + \sum_j   \tilde{\gamma}_{ij}Z_j.
	\end{align*}
	Define $\alpha_{\mid \bs{W},i}$ as the projection of $\alpha_i$ on $\bs{W}_i$,
	\begin{align*}
		\alpha_{\mid \bs{W},i} = \bs{\beta}_\alpha^\top \bs{W}_i,\quad \text{where} \quad \bs{\beta}_\alpha = \frac{1}{N}\sum_i  \bs{W}_i \alpha_i.
	\end{align*}
	Similarly, we define $\theta_{\mid \bs{W},i}$ and $h_{\mid \bs{W},i}$ as the linear projections of $\theta_i$ and $h_i = \sum_j   {\tilde{\gamma}}_{ij}$ on $ \bs{W}_i$. Then,
	\begin{align*}
		&\bs{\beta}_{0,\ora}^\top \bs{W}_i = \bs{W}_i^\top \sum_i  \bs{W}_i Y_{Z_i=0}/N = \alpha_{\mid \bs{W},i} + r_1 h_{\mid \bs{W},i}.
	\end{align*}
	Therefore,
	\[
	\alpha_i-\bs{\beta}_{0,\ora}^\top \bs{W}_i =  \alpha_{\bsw,i} - r_1 h_{\mid \bs{W},i}.
	\]
	Recall that
	\[
	\theta_i-\bs{\beta}_{1,\ora}^\top \bs{W}_i+\bs{\beta}_{0,\ora}^\top \bs{W}_i = \theta_{\bsw,i}.
	\]
	
	By \eqref{eqn:equality1} and replacing $\alpha_i$ with $\alpha_{\bsw,i} - r_1 h_{\mid \bs{W},i}$ and $\theta_i$ with $\theta_{\bsw,i}$ on the right hand side of
	\begin{align*}
		\sigma^2_{\ind,1} = \frac{1}{N^2r_1r_0}\sum_j \Bigl\{\sum_i   E_{ij}( \alpha_i  + r_{1} \theta_i + r_{0}\tilde{\gamma}_{ij} +  r_{1}\sum_{k\ne j}\tilde{\gamma}_{ik})\Bigr\}^2,
	\end{align*}
	we have
	\begin{align*}
		(\sigma^{\ev}_{\ind,1})^2 &= \frac{1}{N^2r_1r_0}\sum_j \Bigl\{\sum_i   E_{ij}( \alpha_{\bsw,i} - r_1 h_{\mid \bs{W},i}  + r_{1} \theta_{\bsw,i} + r_{0}\tilde{\gamma}_{ij} +  r_{1}\sum_{k\ne j}\tilde{\gamma}_{ik})\Bigr\}^2\\
		&= \frac{1}{N^2r_1r_0}\sum_j \Bigl\{\sum_i   E_{ij}( \alpha_{\bsw,i} + r_1 h_{\bsw,i}  + r_{1} \theta_{\bsw,i} + (r_{0}-r_1)\tilde{\gamma}_{ij} )\Bigr\}^2\\
		&= \frac{1}{N^2r_1r_0}\sum_j \Bigl\{\sum_i   E_{ij}( r_1 e_{Z_i=1}+ r_0 e_{Z_i=0}+ (r_{0}-r_1)\tilde{\gamma}_{ij} )\Bigr\}^2\\
		&\lesssim N^{-1}\Deltan + \frac{1}{N^2}\sum_j \Bigl\{\sum_i  Q_{ij}\Bigr\}^2 \\
		&\lesssim N^{-1}\Deltan + \frac{1}{N^2} \bs{1}^\top \bs{Q}\bs{Q}^\top \bs{1} = O(N^{-1}\Deltan + N^{-1}).
	\end{align*}
	Moreover, by \Cref{a:bounded-parameter-adjusted}, similar to the way that we obtain the order of $\sigma^2_{\ind,2}$, we have
	\begin{align*}
		(\sigma^{\ev}_{\ind,2})^2 &\leq \frac{2}{N^2}  \sum_{ i\ne j } \Big(E_{ij}\theta_{\backslash \bs{W},i}+\sum_k  E_{kj}\tilde{\gamma}_{ki}\Big)^2 \lesssim \rhon.
	\end{align*}
	
	Putting together and noting that $N^{-1} = O(\rhon)$ by \Cref{a:density-rho_N}, we have $(\sigma^{\ev}_{\ind})^2 = O(N^{-1}\Deltan + \rhon + N^{-1})=  O(N^{-1}\Deltan + \rhon)$. Moreover, 
	\begin{align*}
		&(\sigma_{\tot,1}^{\ev})^2 \lesssim \sigma_{\dir,1}^2 + (\sigma_{\ind,1}^{\ev})^2 = O(N^{-1}\Deltan + \rhon + N^{-1}),\\
		& (\sigma_{\tot,2}^{\ev})^2 \lesssim \sigma_{\dir,2}^2 + (\sigma_{\ind,2}^{\ev})^2 = O(\rhon + N^{-1}) = O(\rhon).
	\end{align*}
	Therefore, $(\sigma^{\ev}_{\tot})^2 = (\sigma_{\tot,1}^{\ev})^2 + (\sigma_{\tot,2}^{\ev})^2 =  O(N^{-1}\Deltan + \rhon)$.
\end{proof}

\section{Proof of the asymptotic normality}\label{sec:C}

\subsection{Notation and preliminary results}
For a random variable $U$, let $\kappa_4(U) = \E  (U^4) - 3\{\E(  U^2)\}^2$ denote its fourth-order cumulant. 
For one-variable function $f(i)$, let $\operatorname{Inf}_i(f) = f^2(i)$ and $\left\|f\right\|_{\ell_2} = \{\sum_i f^2(i) \}^{1/2}$. For two-variable symmetric function $f(i,j)$ that vanishes on diagonals, i.e., $f(i,j)=f(j,i)$ for $i \ne j$ and $f(i,i)=0$, let $\operatorname{Inf}_i(f) = \sum_j f^2(i,j)$ and $\left\|f\right\|_{\ell_2} = \{\sum_{i \ne j} f^2(i,j)\}^{1/2}$. Denote $\mathcal{M}(f) = \max_i \operatorname{Inf}_i(f)$.

\begin{lemma}
	\label{lem:kappa4-of-some-random-variables}
	Let $\{W_i\}_{i=1}^N$ be i.i.d random variables with $\E W_i = 0$ and $\E W_i^2 = 1$. Let $\{f_i\}_{1\leq i \leq N}$ and $\{f_{ij}\}_{1\leq i\ne j \leq N }$ be fixed arrays with $f_{ij} = f_{ji}$. Then, we have
	\begin{align*}
		\kappa_4\Big(\sum_i W_i f_i\Big) = \kappa_4(W_1) \sum_i f_i^4,\quad  \kappa_4\Big(\sum_{i\ne j} W_i W_j f_{ij}\Big) = O(|M_1|+|M_2|+|M_3|+|M_4|),
	\end{align*}
	where $M_1 = \sum_{i \ne j}f_{ ij}^4$, $M_2 = \sum_{i\ne j \ne k}f_{ ij}^2f_{ ik} f_{jk}$, $M_3 = \sum_{i\ne j \ne k}f_{ ij}^2f_{ ik}^2$, and $M_4  = \sum_{i \ne j \ne k \ne l}f_{ ik}f_{ il}$ $f_{ jk}f_{ jl}$.
\end{lemma}

\begin{proof}
	Note that
	\begin{align*}
		&\E \big(\sum_i W_i f_i\big)^4 = \E \sum_i f_i^4 W_i^4 + 3\E \sum_{i\ne j}f_i^2f_j^2 W_i^2W_j^2 = \E W_i^4 \sum_i f_i^4 + 3\sum_{i\ne j} f_i^2f_j^2,\\
		&\Big\{\E  \big(\sum_i W_i f_i\big)^2\Big\}^2 = \big(\E  \sum_i f_i^2 W_i^2\big)^2 = \big(\sum_i f_i^2\big)^2 =\sum_i f_i^4 +  \sum_{i\ne j} f_i^2f_j^2.
	\end{align*}
	As a consequence, we have
	\[
	\kappa_4\big(\sum_i W_i f_i\big)  = \kappa_4(W_1) \sum_i f_i^4.
	\]
	
	Moreover, we have
	\begin{align*}
		\E  \big(\sum_{i\ne j} W_i W_j f_{ij}\big)^4 =& C_1\sum_{i \ne j}f_{ ij}^4 + C_2\sum_{i\ne j \ne k}f_{ ij}^2f_{ ik}^2 + C_3\sum_{i\ne j \ne k}f_{ ij}^2f_{ ik}f_{ jk}  +\\
		&12\sum_{i \ne  j \ne k \ne l}f_{ ij}^2f_{ kl}^2 + C_4\sum_{i \ne j \ne k \ne l}f_{ ik}f_{ il}f_{ jk}f_{ jl},\\
		\Big\{\E  \big(\sum_{i\ne j} W_i W_j f_{ij}\big)^2\Big\}^2 =& C_5 \sum_{i \ne j}f_{ ij}^4 + C_6 \sum_{i \ne j \ne k}f_{ ij}^2f_{ ik}^2 + 4\sum_{i\ne j\ne k \ne l}f_{ ij}^2f_{ kl}^2.
	\end{align*}
	Here, $C_k$, $1\leq k \leq 6$, are constants that do not depend on $N$. Their expressions are not crucial for our proof, so we omit them.  As a consequence, we have
	\begin{align*}
		\kappa_4\big(\sum_{i\ne j} W_i W_j f_{ij}\big) = O(|M_1|+|M_2|+|M_3|+|M_4|).
	\end{align*}
\end{proof}

\begin{lemma}
	\label{lem:DeJong}
	Let $W_i = (Z_i-r_1)/\sqrt{r_1r_0}$ with $E(W_i) = 0$ and $E(W_i^2)=1$. Let $\{f_i\}_{1\leq i \leq N}$ and $\{f_{ij}\}_{1\leq i \ne j \leq N }$ be fixed arrays with $f_{ij} = f_{ji}$. Define $B_{1} = \sum_{i}f_{i}W_i$, $B_{2} =\sum_{i \ne j}f_{ij}W_iW_j $, $\sigma^2_1 = \Var(B_{1})$, and $\sigma^2_2 = \Var(B_{2})$. For $k=1,2$, if \\$\liminf_{N\rightarrow \infty}\sigma_{k}^2 >0$, $ \kappa_4(B_{k}) /\sigma_{k}^4 \rightarrow 0$ and $\mathcal{M}\left(f_{k}\right) / \sigma_{k}^2 \rightarrow 0$, then
	$
	B_{k} / \sigma_{k} \xrightarrow{d} \mathcal{N}(0,1).
	$
\end{lemma}

Lemma~\ref{lem:DeJong} is a direct result of Theorem 1 of \cite{de1990central}, so we omit its proof.

\subsection{Proof of Theorem~\ref{thm:CLT-no-adjust}}

Firstly, we prove \Cref{prop:marginal-CLTs-for-direct-effect-components} below, which is useful for deriving the asymptotic normality of the DATE estimator.

\begin{proposition}
	\label{prop:marginal-CLTs-for-direct-effect-components}
	Under Assumptions~\ref{a:bounded-parameter}--\ref{a:Lindberg-condition-unadj},
	\begin{itemize}
		\item[(i)] if $\liminf_{N\rightarrow \infty} N\sigma^2_{\dir,1} >0$, then $B_{\dir,1}/\sigma_{\dir,1} \xrightarrow{d} \mathcal{N}(0,1)$;
		\item[(ii)] if $\liminf_{N\rightarrow \infty} N\sigma^2_{\dir,2} >0$, then $B_{\dir,2}/\sigma_{\dir,2} \xrightarrow{d} \mathcal{N}(0,1).$
	\end{itemize}    
\end{proposition}
\begin{proof}
	Let $W_i = (Z_i-r_1)/\sqrt{r_1r_0}$ with $E(W_i) = 0$ and $E(W_i^2)=1$. Let  $f_{\dir,1}(i) = (r_0Y_{Z_i=1}+r_1Y_{Z_i=0}) / (N\sqrt{r_1r_0})$ and $f_{\dir,2}(i,j) = (\tilde{\gamma}_{ij}+\tilde{\gamma}_{ji})/(2N)$. Then $B_{\dir,1} = \sum_{i}f_{\dir,1}(i)W_i$ and $B_{\dir,2} =\sum_{i \ne j}f_{\dir,2}(i,j)W_iW_j $.

	{\bf Step 1.} Proof of (i). By \Cref{lem:DeJong}, we have
	$$
	B_{\dir,k} / \sigma_{\dir,k} \xrightarrow{d} \mathcal{N}(0,1), \quad k \in\{1,2\},
	$$
	provided that the following two conditions hold: (1) $  \kappa_4(B_{\dir,k}) /\sigma_{\dir,k}^4 \rightarrow 0$; and (2)\\ $\mathcal{M}\left(f_{\dir,k}\right)/ \sigma_{\dir,k}^2 \rightarrow 0$.
	
	By \Cref{a:bounded-parameter}, we have $\max_i |f_{\dir,1}(i)| =O(N^{-1})$. Moreover, by \Cref{lem:kappa4-of-some-random-variables}, we have $$\kappa_4(B_{\dir,1}) = O\big(\sum_i f_{\dir,1}^4(i)\big) = O(N^{-3}), \quad \mathcal{M}\left(f_{\dir,1}\right) = \max_i f_{\dir,1}^2(i) = O(N^{-2}).$$  If $\liminf_N N\sigma_{\dir,1}^2  >0$, then  $\mathcal{M}(f_{\dir,1})/\sigma_{\dir,1}^2 = O(N^{-1})$ and $\kappa_4(B_{\dir,1})/\sigma_{\dir,1}^4 = O(N^{-1})$. Therefore, $B_{\dir,1}/\sigma_{\dir,1} \xrightarrow{d} \mathcal{N}(0,1)$.

	{\bf Step 2.} Proof of (ii).
	Denote $\bar{E}_{ij} = E_{ij}+E_{ji} \leq 2$, $\bar{Q}_{ij} = Q_{ij}+Q_{ji}\leq 2$, and ${f}_{\dir,ij}: = (\tilde{\gamma}_{ij}+\tilde{\gamma}_{ji})/(2N) \lesssim \bar{Q}_{ij} / N$. We write $f_{\dir,2}(i,j)$ as $f_{\dir,ij}$ for short. Then $\|\bs{\bar{Q}}\|_{\oprtnorm}\leq 2 \|\bs{Q}\|_{\oprtnorm} = O(1)$ and 
	\begin{align*}
		\sum_i \sum_j f_{\dir,ij}^2 &\lesssim N^{-2}\sum_i\sum_j \bar{Q}_{ij}^2 \lesssim N^{-2}\bs{1}^\top \bs{\bar{Q}} \bs{\bar{Q}}^\top \bs{1}= O(N^{-1}),\\
		\max_{i,j} f_{\dir,ij}^2 &\lesssim N^{-2} \max_{i,j} \bar{Q}_{ij}^2 = O(N^{-2}),\\
		\max_{i} \sum_j f_{\dir,ij}^2 &\lesssim N^{-2}\max_i\sum_j \bar{Q}_{ij}^2 \leq N^{-2} \max_i(\bs{\bar{Q}}\bs{\bar{Q}}^\top)_{ii}=O(N^{-2}).
	\end{align*}
	By \Cref{lem:kappa4-of-some-random-variables}, we have
	$$
	\kappa_4\left({B}_{\dir,2}\right)=O\left(\left|M_{\dir,21}\right|+\left|M_{\dir,22}\right|+\left|M_{\dir,23}\right|+\left|M_{\dir,24}\right|\right),
	$$
	where 
	\begin{align*}
		M_{\dir,21}&=\sum_{i_1\ne i_2} f_{\dir,i_1 i_2}^4 \leq \max_{i_1,i_2} f_{\dir,i_1i_2}^2 \cdot \sum_{i_1}\sum_{i_2}f_{\dir,i_1i_2}^2 = O(N^{-3}), \\
		M_{\dir,22}&=\sum_{i_1\ne i_2\ne i_3} f_{\dir,i_1 i_2}^2 f_{\dir,i_2 i_3} f_{\dir,i_3 i_1} \lesssim N^{-4}\max_{i_1,i_2} \bar{Q}_{i_1i_2} \cdot \sum_{i_1 \ne i_2\ne i_3} \bar{Q}_{i_1 i_2} \bar{Q}_{i_2 i_3}\bar{Q}_{i_3 i_1}\\
		&\lesssim N^{-4}\max_{i_1,i_2} \bar{Q}_{i_1i_2} \cdot \sum_{i_1 \ne i_2\ne i_3}\bar{Q}_{i_2 i_3}\bar{Q}_{i_3 i_1} \leq N^{-4}\max_{i_1,i_2} \bar{Q}_{i_1i_2} \cdot \bs{1}^\top\bs{\bar{Q}}\bs{\bar{Q}}\bs{1} = O(N^{-3}),\\
		M_{\dir,23}&=\sum_{i_1 \ne i_2\ne i_3} f_{\dir,i_1 i_2}^2 f_{\dir,i_2 i_3}^2 \leq \max_{i_2} \sum_{i_3} f_{\dir,i_2i_3}^2 \cdot \sum_{i_1}\sum_{i_2}f_{\dir,i_1i_2}^2 = O(N^{-3}) , \\
		M_{\dir,24}&=\sum_{i_1\ne i_2\ne i_3\ne i_4} f_{\dir,i_1 i_2} f_{\dir,i_2 i_3} f_{\dir,i_3 i_4} f_{\dir,i_4 i_1} \\
		&\lesssim N^{-1}\sum_{i_1\ne i_2\ne i_3\ne i_4} f_{\dir,i_1 i_2} f_{\dir,i_2 i_3} f_{\dir,i_3 i_4} \\
		&\leq N^{-4}\bs{1}^\top\bs{\bar{Q}}\bs{\bar{Q}}\bs{\bar{Q}}\bs{1}= O(N^{-3}).
	\end{align*}
	Therefore, if $\liminf_{N\rightarrow\infty}N\sigma_{\dir,2}^2>0$, then $\kappa_4({B}_{\dir,2})/\sigma_{\dir,2}^4 = O(N^{-1}) = o(1)$.
	
	Moreover,
	\begin{align*}
		\mathcal{M}\left(f_{\dir,2}\right) &= \max_i \operatorname{Inf}_i\left(f_{\dir,2}\right) \lesssim N^{-2}\big(\max_i \sum_j Q_{ij}^2 + \max_i \sum_j Q_{ji}^2\big)\\
		& \lesssim N^{-2}\big(\|\bs{Q}\bs{Q}^\top\|_{\oprtnorm} +\|\bs{Q}^\top\bs{Q}\|_{\oprtnorm} \big)=O(N^{-2}).
	\end{align*}
	The conclusion follows. 
\end{proof}

Next, we prove a general result for deriving the asymptotic normality of the unadjusted and EV-adjusted IATE and GATE estimators. Define
$\hat{\tau}^\prime_{\tot} = \hat{\tau}_{\dir} + \hat{\tau}^\prime_{\ind}$ with 
\[
\hat{\tau}^\prime_{\ind}  =  \frac{1}{N}\sum_{i,j} E_{ij} \Bigl\{\frac{Y_i^\prime Z_j}{r_1} - \frac{Y_i^\prime (1-Z_j)}{r_0}\Bigr\},
\]
where $Y_i^\prime = \alpha_i^\prime + \theta_i^\prime Z_i +\sum_j  \tilde{E}_{ij} {\gamma}_{ij}^\prime Z_j$ with $\tilde{\gamma}^\prime_{ij} = \tilde{E}_{ij} {\gamma}_{ij}^\prime$. Here, $\{\alpha_i^\prime,\theta_i^\prime,{\gamma}_{ij}^\prime\}_{1 \leq i,j\leq N}$ are finite-population parameters like $\{\alpha_i,\theta_i,{\gamma}_{ij}\}_{1 \leq i,j\leq N}$. We define analogous quantites with respect to $\{\alpha_i^\prime,\theta_i^\prime,{\gamma}_{ij}^\prime\}_{1 \leq i,j\leq N}$. For example, we define $\tau_{\ind}^\prime$ and $\tau_{\tot}^\prime$ as the IATE and GATE for $\{\alpha_i^\prime,\theta_i^\prime,{\gamma}_{ij}^\prime\}_{1 \leq i,j\leq N}$.

Note that $ \hat{\tau}_{\tot}^\prime \equiv \hat{\tau}_{\tot}$ when $Y_i^\prime \equiv Y_i$, and $ \hat{\tau}_{\tot}^\prime \equiv \tilde{\tau}^{\ev}_{\tot}$ when $Y_i^\prime \equiv e_i$.  Analogously, we define $ \sigma_{\tot}^{\prime~2}:= \Var(\hat{\tau}^\prime_{\tot})$. Let  $\sigma_{\tot,1}^{\prime~2}$ and $\sigma_{\tot,2}^{\prime~2}$ be the variances of linear and quadratic components of $\hat{\tau}^\prime_{\tot}$, respectively, similar to the definitions of $\sigma_{\tot,1}^2$ and $\sigma_{\tot,2}^2$ with $Y_i$ replaced by $Y_i^\prime$. Then, $\sigma_{\tot}^{\prime ~ 2}= \sigma_{\tot,1}^{\prime~2} + \sigma_{\tot,2}^{\prime ~ 2}$. Similarly, define $\sigma_{\ind}^{\prime~2}: = \Var(\hat{\tau}^\prime_{\ind})  = \sigma_{\ind,1}^{\prime~2} + \sigma_{\ind,2}^{\prime~2}$.

\begin{proposition}
	\label{prop:marginal-CLTs-for-indirect-effect-components}
	Suppose that Assumption~\ref{a:bounded-parameter} holds for both $\{\alpha_i,\theta_i,{\gamma}_{ij}\}_{1 \leq i,j\leq N}$ and\\ $\{\alpha_i^\prime,\theta_i^\prime,{\gamma}_{ij}^\prime\}_{1 \leq i,j\leq N}$, Assumption \ref{a:density-rho_N}--\ref{a:Lindberg-condition-unadj} holds, and for a nonrandom sequence $\Lambdan$,
	\begin{align*}
		& N^{-1} \max\Big\{\sum_i \Big(\sum_j  E_{ji}Y^\prime_{Z_j=1}\Big)^2, \sum_i \Big(\sum_j  E_{ji}Y^\prime_{Z_j=0}\Big)^2\Big\} = O(\Lambdan ),\\
		&N^{-1} \max\Big\{\maxi\Big(\sum_j  E_{ji}Y^\prime_{Z_j=1}\Big)^2, \maxi\Big(\sum_j  E_{ji}Y^\prime_{Z_j=0}\Big)^2\Big\} = o(\Lambdan ).
	\end{align*}
	Then, we have
	\begin{itemize}
		\item[(i)] if $\liminf_{N\rightarrow \infty} \sigma^{\prime~2}_{\ind,1}/(N^{-1}\Lambdan+\rhon) >0$, we have $B_{\ind,1}^\prime/\sigma_{\ind,1}^\prime \xrightarrow{d} \mathcal{N}(0,1)$;
		\item[(ii)] if $\liminf_{N\rightarrow \infty} \sigma^{\prime~2}_{\ind,2}/(N^{-1}\Lambdan+\rhon) >0$, we have $B_{\ind,2}^\prime/\sigma_{\ind,2}^\prime \xrightarrow{d} \mathcal{N}(0,1)$;
		\item[(iii)]  if $\liminf_{N\rightarrow \infty} \sigma_{\tot,1}^{\prime~2}/(N^{-1}\Lambdan+\rhon) >0$, we have $B_{\tot,1}^\prime/\sigma_{\tot,1}^\prime \xrightarrow{d} \mathcal{N}(0,1)$;
		\item[(iv)] if $\liminf_{N\rightarrow \infty} \sigma_{\tot,2}^{\prime~2}/(N^{-1}\Lambdan+\rhon) >0$, we have $B_{\tot,2}^\prime/\sigma_{\tot,2}^\prime \xrightarrow{d} \mathcal{N}(0,1)$.
	\end{itemize}
\end{proposition}
\begin{proof}
	{\bf Step 1.} We prove (i). With a slight abuse of notation, let 
	$$f_{\ind,1}(i) =\sum_j E_{ji}\{r_1 Y_{Z_j=1}^\prime+ r_0 Y_{Z_j=0}^\prime+ (r_{0}-r_1)\tilde{\gamma}_{ji}^\prime\}/ (N\sqrt{r_1r_0}).$$
	Note that $\liminf_N \sigma_{\ind,1}^{\prime~2} / (N^{-1}\Lambdan+\rhon) >0$, $\sum_i f_{\ind,1}^2(i) = \sigma_{\ind,1}^{\prime~2}$, and
	\begin{align*}
		& N^{-1} \max \Big\{\sum_i \Big(\sum_j  E_{ji}Y_{Z_j=1}^\prime\Big)^2, \sum_i \Big(\sum_j  E_{ji}Y_{Z_j=0}^\prime\Big)^2\Big\} = O(\Lambdan ),\\
		&N^{-1} \max \Big\{\maxi\Big(\sum_j  E_{ji}Y_{Z_j=1}^\prime\Big)^2, \maxi\Big(\sum_j  E_{ji}Y_{Z_j=0}^\prime\Big)^2\Big\} = o(\Lambdan ).
	\end{align*}
	Then, we have
	\begin{align*}
		&\frac{\max_i f_{\ind,1}^2(i)}{\sigma_{\ind,1}^{\prime~2}} \lesssim \frac{N^{-2}\max_i \{ (\sum_{j}E_{ji}Y_{Z_j=1}^\prime)^2, (\sum_{j}E_{ji}Y_{Z_j=0}^\prime)^2,\sum_j Q_{ji}^2\}}{N^{-1}\Lambdan+\rhon} = o(1),\\
		&  \sum_i f_{\ind,1}^4(i) \leq \max_i f_{\ind,1}^2(i) \sum_i f_{\ind,1}^2(i) = o(\sigma_{\ind,1}^{\prime~4}).
	\end{align*}
	It follows that 
	$$
	\begin{aligned}
		\kappa_4(B_{\ind,1}^\prime) /\sigma_{\ind,1}^{\prime~4} \rightarrow 0,\quad \mathcal{M}\left(f_{\ind,1}\right) / \sigma_{\ind,1}^{\prime~2} \rightarrow 0.
	\end{aligned}
	$$
	As a consequence, by \Cref{lem:DeJong}, we have
	$$
	B^\prime_{\ind,1} / \sigma_{\ind,1}^\prime  \xrightarrow{d} \mathcal{N}(0,1).
	$$
	
	{\bf Step 2.} We prove (ii). With a slight abuse of notation, let $f_{\ind,2}(i,i) = 0$ and
	$$f_{\ind,2}(i,j) = (E_{ij}\theta^\prime_{i}+E_{ji}\theta^\prime_{j}+\sum_{k}E_{kj}\tilde{\gamma}^\prime_{ki}+\sum_{k}E_{ki}\tilde{\gamma}^\prime_{kj} ) / (2N), \quad i \neq j.$$
	Write $f_{\ind,2}(i,j)$ as ${f}_{\ind,ij}$ for short. Let $g_i = \sum_j (\sum_k E_{kj}\tilde{\gamma}^\prime_{ki} + \sum_k E_{ki}\tilde{\gamma}^\prime_{kj})^2$. Then, 
	\begin{align*}
		g_i &\lesssim \sum_j (\sum_k E_{kj}|\tilde{\gamma}^\prime_{ki}|)^2 + \sum_j (\sum_k E_{ki}|\tilde{\gamma}^\prime_{kj}|)^2  \\
		&\lesssim \sum_j \sum_k \sum_l E_{ki}E_{li}Q_{kj}Q_{lj} + \sum_j \sum_k \sum_l E_{kj}E_{lj}Q_{ki}Q_{li}\\
		&= (\bs{E}^\top\bs{Q}\bs{Q}^\top\bs{E})_{ii} +(\bs{Q}^\top\bs{E}\bs{E}^\top\bs{Q})_{ii}.
	\end{align*}
	Therefore, 
	\begin{align}
		\label{eq:max-gi}
		\max_i g_i \lesssim  2 \|\bs{Q}^\top\bs{E}\|_{\oprtnorm}^2 = O(N^2\rhon^2) = o(N^2\rhon),
	\end{align}
	where the last equality follows from \Cref{a:density-rho_N} ($\rhon = o(1)$). Moreover, we have
	\begin{align*}
		\sum_i g_i &\lesssim \sum_i \sum_j \sum_k \sum_l E_{ki}E_{li}Q_{kj}Q_{lj} + \sum_i\sum_j \sum_k \sum_l E_{kj}E_{lj}Q_{ki}Q_{li} \\
		&\lesssim \sum_i \sum_j \sum_k \sum_l E_{li}Q_{kj}Q_{lj} + \sum_i \sum_j \sum_k \sum_l E_{lj}Q_{ki}Q_{li} \\
		& \lesssim N\|\bs{E}\|_{\oprtnorm}\|\bs{Q}\|_{\oprtnorm}^2 = O(N^2\rhon).
	\end{align*}
	Consequently, we have
	\begin{align}
		\label{eq:sum-of-square-of-gi}
		\sum_i g_i^2 \leq (\max_i g_i) \Big(\sum_i g_i\Big) = O(N^4\rhon^3) = o(N^4\rhon^2).
	\end{align}
	If follows from \Cref{a:Lindberg-condition-unadj} that
	\begin{align*}
		\mathcal{M}\left({f}_{\ind,2} \right) &= \max_i \operatorname{Inf}_i\left({f}_{\ind,2}\right) \lesssim \max_i N^{-2}\big(\sum_{j}E_{ij}\theta_{ i}^{\prime~2} +\sum_{j}E_{ji}\theta_{ j}^{\prime~2}+ g_i\big)\\
		&\leq N^{-2} \big\{O(\max_i N_i) + O(\max_i M_i) + O(\max_i g_i)\big\}\\
		&\leq N^{-2} \big\{o(N^{3/2}\rhon) + o(N^2\rhon)\big\} = o(\rhon).
	\end{align*}
	Therefore,
	\begin{align*}
		\max_i \sum_j {f}_{\ind,ij}^2 &= \mathcal{M}\left({f}_{\ind,2}\right) = o(\rhon),\quad 
		\max_{i,j} {f}_{\ind,ij}^2 \leq \max_i \sum_j {f}_{\ind,ij}^2  = o(\rhon),\\
		\max_{i,j} \sum_k |{f}_{\ind,jk}||{f}_{\ind,ki}| &\leq \max_{i,j} \sum_k ({f}_{\ind,jk}^2 + {f}_{\ind,ki}^2)/2 \\
		& \leq \max_i \sum_k {f}_{\ind,ki}^2 + \max_j \sum_k {f}_{\ind,jk}^2 = o(\rhon),\\
		\sum_i \sum_j {f}_{\ind,ij}^2 &= \sigma_{\ind,2}^{\prime~2}  = O(\rhon).
	\end{align*}
	
	By \Cref{lem:kappa4-of-some-random-variables}, we have $\kappa_4\left(B_{\ind,2}^\prime\right)=O\left(\left|M_{\ind,21}\right|+\left|M_{\ind,22}\right|+\left|M_{\ind,23}\right|+\left|M_{\ind,24}\right|\right)$, where
	\begin{align*}   
		M_{\ind,21}&=\sum_{i_1\ne i_2} {f}_{\ind,i_1 i_2}^4 \leq \max_{i_1,i_2} {f}_{\ind,i_1i_2}^2 \cdot \sum_{i_1}\sum_{i_2}{f}_{\ind,i_1i_2}^2 = o(\rhon^2), \\
		M_{\ind,22}&=\sum_{i_1\ne i_2\ne i_3} {f}_{\ind,i_1 i_2}^2 {f}_{\ind,i_2 i_3} {f}_{\ind,i_3 i_1} \leq \max_{i_1,i_2} \sum_{i_3} |{f}_{\ind,i_2i_3}| |{f}_{\ind,i_3i_1}| \cdot \sum_{i_1}\sum_{i_2}{f}_{\ind,i_1i_2}^2 \\
		& = o(\rhon^2),\\
		M_{\ind,23}&=\sum_{i_1\ne i_2\ne i_3} {f}_{\ind,i_1 i_2}^2 {f}_{\ind,i_2 i_3}^2 \leq \max_{i_2} \sum_{i_3} {f}_{\ind,i_2i_3}^2 \cdot \sum_{i_1}\sum_{i_2}{f}_{\ind,i_1i_2}^2 = o(\rhon^2) , \\
		M_{\ind,24}&=\sum_{i_1\ne i_2\ne i_3\ne i_4} {f}_{\ind,i_1 i_2} {f}_{\ind,i_2 i_3} {f}_{\ind,i_3 i_4} {f}_{\ind,i_4 i_1}.
	\end{align*}
	
	Next, we derive the magnitude of $M_{\ind,24}$. Let $\bs{G} = \bs{Q}^\top\bs{E}$, $\bar{G}_{ij} = {G}_{ij}+{G}_{ji}$, $\bar{E}_{ij} = {E}_{ij}+{E}_{ji}$, $\bar{\bs{G}} = (\bar{G}_{ij})_{i,j}$ and  $\bar{\bs{E}} = (\bar{E}_{ij})_{i,j}$, then we have $|{f}_{\ind,i j}| \leq CN^{-1}(\bar{G}_{ij} + \bar{E}_{ij})$, for some constant $C$, and
	\begin{align*}   
		M_{\ind,24}&=\sum_{i_1,i_2,i_3, i_4} {f}_{\ind,i_1 i_2} {f}_{\ind,i_2 i_3} {f}_{\ind,i_3 i_4} {f}_{\ind,i_4 i_1} \\
		&\lesssim N^{-4}\tr\big\{(\bs{\bar{E}}+\bs{\bar{G}})(\bs{\bar{E}}+\bs{\bar{G}})(\bs{\bar{E}}+\bs{\bar{G}})(\bs{\bar{E}}+\bs{\bar{G}})\big\}.
	\end{align*}
	The above trace can be decomposed into
	$$\sum_{\bs{D}_1,\bs{D}_2,\bs{D}_3,\bs{D}_4 \in \{\bs{\bar{E}},\bs{\bar{G}}\}^4} \tr(\bs{D}_1\bs{D}_2\bs{D}_3\bs{D}_4).$$
	\begin{itemize}
		\item[(i)] At least one $\bs{D}_i = \bs{\bar{E}}$ for $i=1,2,3,4$. Without loss of generality, supposing $\bs{D}_1 = \bs{\bar{E}}$, then we have
		$$\tr(\bar{\bs{E}}\bs{D}_2\bs{D}_3\bs{D}_4) \lesssim \bs{1}^\top\bs{D}_2\bs{D}_3\bs{D}_4\bs{1} \leq N\|\bs{D}_2\|_{\oprtnorm}\|\bs{D}_3\|_{\oprtnorm}\|\bs{D}_4\|_{\oprtnorm}.$$
		Note that $\|\bs{\bar{E}}\|_{\oprtnorm} \leq 2 \|\bs{{E}}\|_{\oprtnorm} =O(N\rhon)$, $\|\bs{\bar{G}}\|_{\oprtnorm}\leq 2\|\bs{{Q}}\|_{\oprtnorm}\|\bs{{E}}\|_{\oprtnorm}=O(N\rhon)$, then
		$$N\|\bs{D}_2\|_{\oprtnorm}\|\bs{D}_3\|_{\oprtnorm}\|\bs{D}_4\|_{\oprtnorm}= NO(N\rhon)\cdot O(N\rhon) \cdot O(N\rhon)= O(N^4\rhon^3).$$
		\item[(ii)] No term in the summation equals $\bs{\bar{E}}$. In this case,
		\begin{align*}
			\tr(\bs{\bar{G}}\bs{\bar{G}}\bs{\bar{G}}\bs{\bar{G}}) &\lesssim \bs{1}^\top \bs{\bar{G}}\bs{\bar{G}}\bs{\bar{G}}\bs{\bar{Q}}\bs{1}\\
			& \leq N\|\bs{\bar{G}}\|_{\oprtnorm}\|\bs{\bar{G}}\|_{\oprtnorm}\|\bs{\bar{G}}\|_{\oprtnorm}\| \bs{\bar{Q}}\|_{\oprtnorm}= O(N^4\rhon^3).
		\end{align*}
	\end{itemize}
	
	In summary, we obtain that
	$$\tr\big((\bs{\bar{E}}+\bs{\bar{G}})(\bs{\bar{E}}+\bs{\bar{G}})(\bs{\bar{E}}+\bs{\bar{G}})(\bs{\bar{E}}+\bs{\bar{G}})\big) = O(N^4\rhon^3).$$ Therefore, $M_{\ind,24} = O(\rhon^3) = o(\rhon^2)$, and thus, $\kappa_4(B_{\ind,2}^\prime) = o(\rhon^2)$. If $\liminf_N \sigma_{\ind,2}^{\prime~2} / (N^{-1}\Lambdan$ $+\rhon) >0$,  we have $\kappa_4(B_{\ind,2}^\prime)/\sigma_{\ind,2}^{\prime~4} =o(1) $. As a consequence, by \Cref{lem:DeJong}, we have
	$$
	B_{\ind,2}^\prime / \sigma_{\ind,2}^\prime  \xrightarrow{d} \mathcal{N}(0,1).
	$$
	
	{\bf Step 3.} The proofs of (iii) and (iv) are similar to those of (i) and (ii), so we omit them. 
\end{proof}

Based on these preliminary results, we can prove Theorem~\ref{thm:CLT-no-adjust}.

\begin{proof}[Proof of Theorem~\ref{thm:CLT-no-adjust}]
	{\bf Step 1.} We prove the asymptotic normality of $\hat{\tau}_{\dir}$. Recall the decomposition \eqref{eq:decompose-direct-effect}, we have
	\begin{align*}
		\hat{\tau}_{\dir}-\tau_{\dir} = B_{\dir,1} + B_{\dir,2}.
	\end{align*}
	Moreover,
	$ \Var(B_{\dir,1}) = \sigma^2_{\dir,1}= O(N^{-1})$ and $  \Var(B_{\dir,2}) = \sigma^2_{\dir,2}= O(N^{-1}).$

	Consider the decomposition $\sd^{-1}(\hat{\tau}_{\dir}-\tau_{\dir}) = B_{\dir,1} / \sd + B_{\dir,2} / \sd$, where $\sd^2 = \Var(\hat{\tau}_{\dir})$. Let $\tilde{B}_{\dir,1} = B_{\dir,1} / \sd$, $\tilde{B}_{\dir,2} = B_{\dir,2} / \sd$, $\tilde{f}_{\dir,1}(i) = f_{\dir,1}(i) / \sd$, and $\tilde{f}_{\dir,2}(i,j) = f_{\dir,2}(i,j) / \sd$, then we have $\tilde{B}_{\dir,1} = \sum_{i}\tilde{f}_{\dir,1}(i)W_i$, $\tilde{B}_{\dir,2} =\sum_{i \ne j}\tilde{f}_{\dir,2}(i,j)$ $W_iW_j $.
	
	Let $\mathbf{G}_{\dir}= (G_{\dir,1}, G_{\dir,2})$ be a normal approximation of $(\tilde{B}_{\dir,1}, \tilde{B}_{\dir,2})$, i.e., $G_{\dir,1}$ and $G_{\dir,2}$ are two independent  zero-mean normal random variables with variances $\sigma_{\dir,1}^2/ \sd^2$ and $\sigma_{\dir,2}^2/ \sd^2$, respectively.
	
	By Theorem~2.1 of \cite{koike2022high}, we have
	$$
	\begin{aligned}
		& \sup _{\left(x_1, x_2\right) \in \mathbb{R}^2}\left|\textnormal{P}\left(\tilde{B}_{\dir,1} \leq x_1 ; \tilde{B}_{\dir,2} \leq x_2\right)-\textnormal{P}\left(G_{\dir,1} \leq x_1 ; G_{\dir,2} \leq x_2\right)\right| \\
		\leq &C_{\dir,0}\left(\delta_0(\bs{\tilde{B}}_{\dir})^{\frac{1}{3}}+\delta_1(\bs{\tilde{B}}_{\dir})^{\frac{1}{3}}+\max _{k=1,2}\left(\mathcal{M}(\tilde{f}_{\dir,k})\right)^{1 / 2}\right)\\
		&\cdot \left(1+\frac{1}{\min \left\{\Var(\tilde{B}_{\dir,1})^{1/2}, \Var(\tilde{B}_{\dir,2})^{1/2}\right\}}\right),
	\end{aligned}
	$$
	where $C_{\dir,0}$ is a constant that does not depend on $N$ and
	$$
	\begin{aligned}
		\delta_0(\bs{\tilde{B}}_{\dir})= & \|\operatorname{cov}(\bs{\tilde{B}}_{\dir})-\operatorname{cov}(\bs{G}_{\dir})\|_{\infty}, \\
		\delta_1(\bs{\tilde{B}}_{\dir})= & \left(\left|\kappa_4\left(\tilde{B}_{\dir,1}\right)\right|+\sum_i \operatorname{Inf}_i\left(\tilde{f}_{\dir,1}\right)^2\right)^{1 / 2} \\
		& +\left(\left|\kappa_4\left(\tilde{B}_{\dir,2}\right)\right|+\sum_i \operatorname{Inf}_i\left(\tilde{f}_{\dir,2}\right)^2\right)^{1 / 2} \\
		& +\left\|\tilde{f}_{\dir,1}\right\|_{\ell_2}\left(\left|\kappa_4\left(\tilde{B}_{\dir,2}\right)\right|+\sum_i \operatorname{Inf}_i\left(\tilde{f}_{\dir,2}\right)^2\right)^{1 / 4} .
	\end{aligned}
	$$
	
	Note that $\Var(\tilde{B}_{\dir,k}) = \Var(G_{\dir,k})$ for $k=1,2$, and $\operatorname{Cov}(\tilde{B}_{\dir,1},\tilde{B}_{\dir,2}) = 0 = \operatorname{Cov}(G_{\dir,1},G_{\dir,2})$. Therefore,
	\[
	\delta_0(\bs{\tilde{B}}_{\dir})=0.
	\]
	
	We have shown in \Cref{prop:marginal-CLTs-for-direct-effect-components} that $\mathcal{M}({f}_{\dir,2})=O(N^{-2})$, $\mathcal{M}({f}_{\dir,1})=O(N^{-2})$, $\kappa_4(B_{\dir,1})= O(N^{-3})$, and $\kappa_4({B}_{\dir,2})=O(N^{-3})$. Since $\tilde{f}_{\dir,1} = {f}_{\dir,1}/\sigma_{\dir}$, it follows that $\mathcal{M}(\tilde{f}_{\dir,1})  =O(N^{-1})$, $ \mathcal{M}(\tilde{f}_{\dir,2}) =O(N^{-1})$,  $ \kappa_4(\tilde{B}_{\dir,1}) =O(N^{-1})$, and $ \kappa_4(\tilde{B}_{\dir,2}) =O(N^{-1})$.
	
	Moreover, we have
	\begin{align*}
		\sum_i \operatorname{Inf}_i\left(\tilde{f}_{\dir,1}\right)^2 &= \sum_i \tilde{f}_{\dir,1}^4(i) = O\big(\kappa_4(\tilde{B}_{\dir,1})\big)=O(N^{-1}),\\
		\left\|\tilde{f}_{\dir,1}\right\|_{\ell_2} &= \Big(\sum_i \tilde{f}_{\dir,1}^2(i)\Big)^{1/2} = \left(\sigma_{\dir,1}^2/\sd^{2}\right)^{1/2}=O(1),\\
		\sum_i \operatorname{Inf}_i\left(\tilde{f}_{\dir,2}\right)^2 &= \sum_i\Big(\sum_j \tilde{f}_{\dir,2}^2(i,j)\Big)^2 \leq \sum_i\left(\frac{\sum_{j}\tilde{\gamma}_{ij}^2 + \sum_{j}\tilde{\gamma}_{ji}^2 }{2N^2\sd^2}\right)^2 \\
		&\lesssim N^{-4}\sd^{-4} \Big\{\sum_i \big(\sum_j Q_{ij}^2\big)^2 + \sum_i \big(\sum_j Q_{ji}^2\big)^2\Big\}\\
		&\leq N^{-4}\sd^{-4}\Big(\sum_i \sum_j \sum_k Q_{ij} Q_{ik} + \sum_i \sum_j \sum_k Q_{ji} Q_{ki} \Big)\\
		&\leq N^{-4}\sd^{-4}\big(\bs{1}^\top \bs{Q}^\top\bs{Q}\bs{1} +\bs{1}^\top \bs{Q}\bs{Q}^\top\bs{1}\big) =O(N^{-1}).
	\end{align*}
	
	In summary, we have 
	$$\left(\delta_0(\bs{\tilde{B}}_{\dir})^{\frac{1}{3}}+\delta_1(\bs{\tilde{B}}_{\dir})^{\frac{1}{3}}+\max _{k=1,2}\left(\mathcal{M}(\tilde{f}_{\dir,k})\right)^{1 / 2}\right) = O(N^{-1/12}).$$
	
	By Assumption \ref{a:assumption-CLT}, we have either $\liminf_N N\sigma_{\dir,1}^2  >0$ or $\liminf_N N\sigma_{\dir,2}^2  >0$. Without loss of generality, suppose that $\liminf_N N\sigma_{\dir,1}^2  >0$. Then we split $\{\Var(\tilde{B}_{\dir,2})\}_{N=1,2,...}$ into two subsequences. The first subsequence satisfies $\Var(\tilde{B}_{\dir,2}) < N^{-1/7}$, such that $\Var(B_{\dir,2}) = o(\sd^2)=o(N^{-1})$. The CLT, i.e., $\tilde{B}_{\dir,1}+\tilde{B}_{\dir,2} \xrightarrow{d} \mathcal{N}(0,1)$, holds in this case by Slutsky Theorem and \Cref{prop:marginal-CLTs-for-direct-effect-components}. The second subsequence satisfies $\Var(\tilde{B}_{\dir,2}) \geq N^{-1/7}$, and then 
	\begin{align*}
		&\sup _{\left(x_1, x_2\right) \in \mathbb{R}^2}\left|\textnormal{P}\big(\tilde{B}_{\dir,1} \leq x_1 ; \tilde{B}_{\dir,2} \leq x_2\big)-\textnormal{P}\big(G_{\dir,1} \leq x_1 ; G_{\dir,2} \leq x_2\big)\right| \\
		= &O(N^{-1/12}) \cdot O\bigg(1+\frac{1}{N^{-1/14}}\bigg) = o(1),
	\end{align*}
	which implies 
	$$\sup _{x \in \mathbb{R}}\left|\textnormal{P}\big(\tilde{B}_{\dir,1} +\tilde{B}_{\dir,2} \leq x\big)-\textnormal{P}\big(G_{\dir,1}+G_{\dir,2} \leq x\big)\right| = o(1).$$
	Therefore, we always have $$\tilde{B}_{\dir,1}+\tilde{B}_{\dir,2} = \frac{\hat{\tau}_{\dir}-\tau_{\dir}}{\Var(\hat{\tau}_{\dir})^{1/2}} \xrightarrow{d} \mathcal{N}(0,1).$$
	
	\ \\
	{\bf Step 2.}  We prove the asymptotic normality of $\hat{\tau}_{\ind}$ and $\hat{\tau}_{\tot}$. We will prove a more general result stated in \Cref{thm:CLT-general} below.
	\begin{theorem}
		\label{thm:CLT-general}
		Suppose that Assumption~\ref{a:bounded-parameter} holds for both $\{\alpha_i,\theta_i,{\gamma}_{ij}\}_{1 \leq i,j\leq N}$ and\\ $\{\alpha_i^\prime,\theta_i^\prime,{\gamma}_{ij}^\prime\}_{1 \leq i,j\leq N}$, Assumptions \ref{a:density-rho_N}--\ref{a:Lindberg-condition-unadj} hold, and for a nonrandom sequence $\Lambdan$, 
		\begin{align*}
			& N^{-1} \max\Big\{\sum_i \Big(\sum_j  E_{ji}Y^\prime_{Z_j=1}\Big)^2, \sum_i \Big(\sum_j  E_{ji}Y^\prime_{Z_j=0}\Big)^2\Big\} = O(\Lambdan ),\\
			&N^{-1} \max\Big\{\maxi\Big(\sum_j  E_{ji}Y^\prime_{Z_j=1}\Big)^2, \maxi\Big(\sum_j  E_{ji}Y^\prime_{Z_j=0}\Big)^2\Big\} = o(\Lambdan ).
		\end{align*}
		Then, we have
		\begin{itemize}
			\item[(i)] if 
			$
			\text{ either } \liminf {\sigma_{\ind,1}^{\prime~2}}/{(N^{-1}\Lambdan +\rhon)} >0 \text{ or }  \liminf {\sigma_{\ind,2}^{\prime~2} }/{(N^{-1}\Lambdan +\rhon)} >0,
			$
			then  $(\hat{\tau}^\prime_{\ind}-\tau^\prime_{\ind})/\Var(\hat{\tau}^\prime_{\ind})^{1/2}\xrightarrow{d}\mathcal{N}(0,1)$;
			\item[(ii)] if $\text{ either } \liminf {\sigma_{\tot,1}^{\prime~2}}/{(N^{-1}\Lambdan +\rhon)} >0 \text{ or } \liminf {\sigma_{\tot,2}^{\prime~2} }/{(N^{-1}\Lambdan +\rhon)} >0,$ then $(\hat{\tau}^\prime_{\tot}-\tau_{\tot}^\prime)/\Var(\hat{\tau}_{\tot}^\prime)^{1/2}\xrightarrow{d}\mathcal{N}(0,1)$.
		\end{itemize}
	\end{theorem}

	Recall that $M_i = \sum_j  E_{ji}$. By Assumptions~\ref{a:bounded-parameter}--\ref{a:Lindberg-condition-unadj}, we have, 
	$$
	N^{-1}\sum_i  \Big(\sum_j   E_{ji}Y^\prime_{Z_j=z} \Big)^2 \lesssim  N^{-1} \bs{1}^\top\bs{E}\bs{E}^\top\bs{1} \lesssim  \|\bs{E}\|_{\oprtnorm}^2 = O(N^2\rhon^2), \quad z=0,1,
	$$
	$$
	N^{-1} \Big\{\maxi\Big(\sum_j  E_{ji}Y^\prime_{Z_j=z}\Big)^2 \Big\} \lesssim  N^{-1} \maxi M_i^2 = o(N^2\rhon^2), \quad z=0,1.
	$$
	Thus, the asymptotic normality of $\hat{\tau}_{\ind}$ and $\hat{\tau}_{\tot}$ follows by taking $\Lambdan \equiv N^2\rhon^2$ and $Y^\prime_i \equiv Y_i$ in \Cref{thm:CLT-general}. In the remainder of this section, we will prove \Cref{thm:CLT-general}. We will only prove \Cref{thm:CLT-general}(i), as the proof of \Cref{thm:CLT-general}(ii) is almost the same.

	We have shown in \eqref{eqn:equality1} that (with $i$ and $j$ exchanged)
	\begin{equation}\label{eqn:equality11}
		\alpha_j  + r_{1} \theta_j + r_{0}\tilde{\gamma}_{ji} +  r_{1}\sum_{k\ne i}\tilde{\gamma}_{jk} = r_1r_0 Y_{Z_j=1}^{Z_i=1} +r_0^2Y_{Z_j=0}^{Z_i=1}+r_1^2Y_{Z_j=1}^{Z_i=0}+r_1r_0Y_{Z_j=0}^{Z_i=0}. \nonumber 
	\end{equation}
	Moreover, we have
	\begin{equation}\label{eqn:equality0}
		r_1Y_{Z_j=1}+r_0Y_{Z_j=0} = \alpha_j+r_1 \theta_j+\sum_k   \tilde{\gamma}_{jk}r_1. \nonumber
	\end{equation}
	
	Similar to the decomposition of $\hat{\tau}_{\ind}$ in \eqref{eq:decompose-indirect-effect}, we have
	\begin{align*}
		\hat{\tau}_{\ind}^\prime - \tau_{\ind}^\prime = B_{\ind,1}^\prime +B_{\ind,2}^\prime,
	\end{align*}
	where 
	\begin{align*}
		B_{\ind,1}^\prime &= \sum_i \frac{\left(Z_i-r_1\right)}{N r_1 r_0}\sum_j E_{j i}\big(r_1r_0 Y_{Z_j=1}^{\prime Z_i=1} +r_0^2 Y_{ Z_j=0}^{\prime Z_i=1}+r_1^2Y_{Z_j=1}^{\prime Z_i=0}+r_1r_0Y_{Z_j=0}^{\prime Z_i=0}\big), \\
		& = \sum_i \frac{\left(Z_i-r_1\right)}{N r_1 r_0} \sum_j E_{ji}\big\{r_1 Y^\prime _{Z_j=1}+ r_0 Y^\prime _{Z_j=0}+ (r_{0}-r_1)\tilde{\gamma}^\prime _{ji} \big\},\\
		B^\prime _{\ind,2} &=   \sum_{i \neq j} \frac{\left(Z_i-r_1\right)\left(Z_j-r_1\right)}{N r_1 r_0} \big(E_{ij}\theta^\prime _{i}+\sum_{k}E_{kj}\tilde{\gamma}^\prime _{ki} \big).
	\end{align*}
	
	Mimicking the proof of \Cref{prop:order-of-adjusted-estimators} with $e_i \equiv Y_i^\prime$ and $\Deltan \equiv \Lambdan$, we have 
	$$
	\begin{aligned}
		&\sigma^{\prime~2}_{\ind,1}=\Var(B^\prime_{\ind,1}) = O(N^{-1}\Lambdan+\rhon),\quad \sigma^{\prime~2}_{\ind,2} =\Var(B^\prime_{\ind,2}) =  O(\rhon).
	\end{aligned}$$

	Let $\tilde{B}_{\ind,1}^\prime = B_{\ind,1}^\prime / \si^\prime$, $\tilde{B}_{\ind,2}^\prime = B_{\ind,2}^\prime / \si^\prime$, $\tilde{f}_{\ind,1}(i) = f_{\ind,1}(i) / \si^\prime$, and $\tilde{f}_{\ind,2}(i,j) = f_{\ind,2}(i,j) / \si^\prime$. By the assumption of $\liminf_{N\rightarrow \infty} \sigma^{\prime~2}_{\ind}/(N^{-1}\Lambdan + \rhon) > 0 $ and the proof of \Cref{prop:marginal-CLTs-for-indirect-effect-components}, we have $\mathcal{M}(f_{\ind,1}) = o(\sigma_{\ind}^{\prime~4})$, $\mathcal{M}(f_{\ind,2}) = o(\rhon)$, $\kappa_4(B_{\ind,1}^\prime) = o(\sigma_{\ind}^{\prime~4})$, and $\kappa_4(B_{\ind,2}^\prime) = o(\rhon^2)$. Since $\tilde{B}_{\ind,i}^\prime= B_{\ind,i}^\prime/\sigma_{\ind}^\prime$, then $\mathcal{M}(\tilde{f}_{\ind,i})  =  o(1)$ and $\kappa_4(\tilde{B}_{\ind,i}^\prime) =o(1)$, $i=1,2$. It is easy to verify that 
	\begin{align*}
		\sum_i \operatorname{Inf}_i\left(\tilde{f}_{\ind,1}\right)^2=o(1),~\left\|\tilde{f}_{\ind,1}\right\|_{\ell_2} = O(1). 
	\end{align*}
	
	By Assumption~\ref{a:opnorm-EE^T}, we have 
	$$
	\sum_i  N_i^2 = \sum_i  (\sum_j  E_{ij})^2  \leq N\|\bs{E}\|_{\oprtnorm}^2 = O(N^3\rho_N^2), $$ 
	$$ \sum_i  M_i^2 = \sum_i  (\sum_j  E_{ji})^2  \leq N\|\bs{E}\|_{\oprtnorm}^2 = O(N^3\rho_N^2).
	$$
	Then,
	\begin{align*}
		\sum_i \operatorname{Inf}_i\left(\tilde{f}_{\ind,2} \right)^2 &= \sum_i\Big\{\sum_j \big(\tilde{f}_{\ind,2} (i,j)\big)^2\Big\}^2 \\
		&\leq N^{-4}(\si^\prime)^{-4}\sum_i\Big(\sum_{j}E_{ij}\theta_{ i}^{\prime 2} +\sum_{j}E_{ji}\theta_{ j}^{\prime 2}+ g_i\Big)^2 \\
		&\lesssim N^{-4}(\si^\prime)^{-4}\sum_i \Big\{ \Big(\sum_{j}E_{ij}\theta_{ i}^{\prime 2}\Big)^2 +\Big(\sum_{j}E_{ji}\theta_{ j}^{\prime 2}\Big)^2+g_i^2 \Big\},\\
		&\leq (N^{2}\Lambdan^{2}+N^{4}\rhon^{2})^{-1} \Big\{O\Big(\sum_i N_i^2\Big) + O\Big(\sum_i M_i^2\Big) + O\Big(\sum_i g_i^2\Big)\Big\}\\
		&\leq (N^{2}\Lambdan^{2}+N^{4}\rhon^{2})^{-1} \big\{O(N^3\rhon^2)+O(N^3\rhon^2)+o(N^4\rhon^2)\big\} = o(1),
	\end{align*}
	where the last inequality is due to \eqref{eq:sum-of-square-of-gi}.
	
	Therefore,
	\begin{align*}
		\delta_{\ind,\scaleobj{0.8}{N}} :=& \left(\delta_0(\bs{\tilde{B}}_{\ind}^\prime)^{\frac{1}{3}}+\delta_1(\bs{\tilde{B}}_{\ind}^\prime)^{\frac{1}{3}}+\max _{k=1,2}\left(\mathcal{M}(\tilde{f}_{\ind,k})\right)^{1 / 2}\right) = o(1).
	\end{align*}
	
	We know that either $\liminf_N \sigma^{\prime~2}_{\ind,1}/(N^{-1}\Lambdan+\rhon) >0$ or $\liminf_N \sigma^{\prime~2}_{\ind,2}/(N^{-1}\Lambdan+\rhon) >0$. Without loss of generality, suppose that $\liminf_N \sigma^{\prime~2}_{\ind,2}/(N^{-1}\Lambdan+\rhon) >0$, and then we split $ \sigma_{\ind,1}^\prime$ into two subsequences.

	The first subsequence satisfies $\sigma_{\ind,1}^{\prime ~ 2}\leq (N^{-1}\Lambdan+\rhon)\delta_{\ind,\scaleobj{0.8}{N}}^{1/3}$. Then, $ \sigma_{\ind,1}^{\prime ~ 2}= o(N^{-1}\Lambdan+\rhon) = o(\sigma_{\ind}^{\prime ~ 2})$. Therefore, $B_{\ind,1}^\prime / \si^\prime \xrightarrow{P} 0$ and $\Var(B_{\ind,2}^\prime) / \si^{\prime~2} \to 1$. By \Cref{prop:marginal-CLTs-for-indirect-effect-components} and Slutsky's throrem, we have $(B_{\ind,1}^\prime+B_{\ind,2}^\prime)/\si^\prime \xrightarrow{d} \mathcal{N}(0,1)$.

	The second subsequence satisfies $\sigma_{\ind,1}^{\prime ~ 2} > (N^{-1}\Lambdan+\rhon)\delta_{\ind,\scaleobj{0.8}{N}}^{1/3}$, then we have
	\begin{eqnarray*}
		&& 1+\frac{1}{\min \left\{\Var(\tilde{B}_{\ind,1}^\prime)^{1/2}, \Var(\tilde{B}_{\ind,2}^\prime)^{1/2}\right\}} \\
		& \lesssim & 1 + \frac{(N^{-1}\Lambdan+\rhon)^{1/2}}{\min \left\{(N^{-1}\Lambdan+\rhon)^{1/2}\delta_{\ind,\scaleobj{0.8}{N}}^{1/6} , (N^{-1}\Lambdan+\rhon)^{1/2}\right\}} = O(\delta_{\ind,\scaleobj{0.8}{N}}^{-1/6}).
	\end{eqnarray*}
	By Theorem 2.1 of \cite{koike2022high}, we have
	\begin{align*}
		&\sup _{\left(x_1, x_2\right) \in \mathbb{R}^2}\left|\textnormal{P}\big(\tilde{B}_{\ind,1}^\prime \leq x_1 ; \tilde{B}_{\ind,2}^\prime \leq x_2\big)-\textnormal{P}\big(G_{\ind,1} \leq x_1 ; G_{\ind,2} \leq x_2\big)\right| = O(\delta_{\ind,\scaleobj{0.8}{N}}^{5/6}),
	\end{align*}
	where $(G_{\ind,1}, G_{\ind,2})$ are normal approximation of $(\tilde{B}_{\ind,1}^\prime, \tilde{B}_{\ind,2}^\prime)$, i.e., $G_{\ind,1}$ and $G_{\ind,2}$ are two independent zero-mean normal random variables with variances $\sigma_{\ind,1}^{\prime ~ 2} / \sigma_{\ind}^{\prime ~ 2}  $ and $\sigma_{\ind,2}^{\prime ~ 2} / \sigma_{\ind}^{\prime ~ 2}$, respectively.
	
	Therefore, we always have $\tilde{B}_{\ind,1}^\prime+\tilde{B}_{\ind,2}^\prime \xrightarrow{d} \mathcal{N}(0,1).$

\end{proof}

\section{Consistency of variance estimator}\label{sec:D}

\subsection{Some useful lemmas}
\Cref{lem:efron-stein} below is from \cite{steele1986efron}.
\begin{lemma}[Efron-Stein inequality]
	\label{lem:efron-stein}
	Let $S_i$, $S_i^\prime$, $i=1,\ldots, N$, be $2N$ i.i.d. random variables. Let $\bs{S} = (S_1,\ldots,S_N)$ and $\bs{S}^{(i)}=(S_1,\ldots,S_{i-1},S_i^{\prime},S_{i+1},\ldots,S_N)$. For any measurable function $f$, we have 
	$$
	\Var\{f(\bs{S})\}\leq\frac{1}{2}\sum_{i} \E\{f(\bs{S})-f(\bs{S}^{(i)})\}^2.
	$$
\end{lemma}

Lemma~\ref{lem:bound-for-variance-efron-stein} below is a direct result of \Cref{lem:efron-stein}.
\begin{lemma}
	\label{lem:bound-for-variance-efron-stein}
	Let $Z_i$, $Z_i^\prime$, $i=1,\ldots, N$, be $2N$ i.i.d. Bernoulli random variables. Let $c_j = \max_{\bs{z}_{(-j)}\in\{0,1\}^{N-1}}|f(\bs{z})-f(\bs{z}^\prime)|$ where $\bs{z} = (z_1,\ldots,z_j=1,\ldots,z_n)$, $\bs{z}^\prime = (z_1,\ldots,z_j=0,\ldots,z_N)$ and $\bs{z}_{(-j)} =(z_1,\ldots,z_{j-1},z_{j+1}\ldots,z_N) $. Then, we have
	\[
	\Var\{f(\bs{Z})\} \leq \frac{1}{2}\sum_j  c_j^2.
	\]
\end{lemma}

\subsection{Proof of Theorem~\ref{thm:variance-estimator-dir}}
\begin{proof}
	{\bf Step 1.} We prove (i). Note that 
	\begin{align*}
		&\E Y_i^2 Z_i = r_1 \E( Y_i^2\mid Z_i=1)= r_1 Y_{Z_i=1}^2 +r_1\Var(Y_i\mid Z_i=1),\\
		&\E Y_i^2 (1-Z_i) = r_0 \E( Y_i^2\mid Z_i=0)= r_0 Y_{Z_i=0}^2 + r_0\Var(Y_i\mid Z_i=0).
	\end{align*}
	Using the independence of $Z_i$'s and that $Y_i = \alpha_i + \theta_i Z_i + \sum_j  \tilde{\gamma}_{ij} Z_j$,  we have
	\[
	\Var(Y_i\mid Z_i=0)=\Var(Y_i\mid Z_i=1) = \sum_j  \tilde{\gamma}_{ij}^2 \Var(Z_j)  = \sum_j  \tilde{\gamma}_{ij}^2 r_0r_1.  
	\]
	Putting together the pieces, we have
	\begin{align}
		\label{eq:E-V_dir}
		\E\hat{V}_{\dir}= &\frac{1}{N^2} \sum_i   \E Y_i^2 Z_i/r_1^2+ \frac{1}{N^2}\sum_i   \E Y_i^2 (1-Z_i)/r_0^2 
		\nonumber \\
		= & \frac{1}{N^2} \sum_i   Y_{Z_i=1}^2/r_1 + \frac{1}{N^2} \sum_i  Y_{Z_i=0}^2/r_0 + \frac{1}{N^2} \sum_j  \tilde{\gamma}_{ij}^2\\
		=& \frac{1}{N^2} \sum_i   \big\{(r_0 Y_{Z_i=1}+r_1Y_{Z_i=0})^2/r_0r_1 + (Y_{Z_i=0}-Y_{Z_i=1})^2\big\} + \frac{1}{N^2} \sum_j  \tilde{\gamma}_{ij}^2\nonumber\\
		=&\frac{1}{N^2} \sum_i   \big\{(r_0 Y_{Z_i=1}+r_1Y_{Z_i=0})^2/r_0r_1 + \theta_i^2\big\} + \frac{1}{N^2} \sum_{i,j} \tilde{\gamma}_{ij}^2\nonumber.
	\end{align}
	Comparing with the expression  of $\sigma^2_{\dir}$, we have
	\[
	\E\hat{V}_{\dir} - \Var(\hat{\tau}_{\ind}) = \frac{1}{N^2} \sum_i   \theta_i^2 -\frac{1}{N^2} \sum_j \sum_i  \tilde{\gamma}_{ij}\tilde{\gamma}_{ji}.
	\]
	Moreover,
	\begin{align*}
		\E (2\hat{V}_{\dir}) - \Var(\hat{\tau}_{\ind}) \geq 0
	\end{align*}
	follows from
	\[
	\E\hat{V}_{\dir} \geq  \frac{1}{N^2} \sum_{i,j} \tilde{\gamma}_{ij}^2 \geq \frac{1}{N^2} \sum_j \sum_i  \tilde{\gamma}_{ij}\tilde{\gamma}_{ji}.
	\]
	
	{\bf Step 2.} We prove (ii). We first show that 
	\begin{equation*}
		N^{-1}\sum_i  Z_i Y_i^2 = \E \Big( N^{-1}\sum_i  Z_i Y_i^2 \Big)  + \op(1). 
	\end{equation*}
	Let $\bs{Z} = (Z_1,\ldots,Z_n)$. Define $f(\bs{Z}) = f(Z_1,\ldots,Z_n): = \sum_i  Z_i Y_i^2 $. Let $\bs{z} = (z_1,\ldots,z_j=1,\ldots,z_n)$, $\bs{z}^\prime = (z_1,\ldots,z_j=0,\ldots,z_n)$, and $\bs{z}_{(-j)} =(z_1,\ldots,z_{j-1},z_{j+1}\ldots,z_n) $. By Assumption~\ref{a:bounded-parameter}, we have 
	\begin{align*}
		\max_{\bs{z}_{(-j)}\in\{0,1\}^{N-1}}|f(\bs{z})-f(\bs{z}^\prime)| \lesssim |Y_j^2(\bs{z})| + \sum_{i:i\ne j}|Y_i^2(\bs{z})-Y_i^2(\bs{z}^\prime)| \lesssim 1 + \sum_i  Q_{ij}.  
	\end{align*}
	As a consequence, by \Cref{lem:bound-for-variance-efron-stein}, we have
	\begin{align*}
		&\Var(f(\bs{Z})) \lesssim \sum_{j}\Big(1 + \sum_i  Q_{ij}\Big)^2 \lesssim \sum_{j} 1 + \sum_{j}\Big(\sum_i  Q_{ij}\Big)^2\\
		&= N + \bs{1}^\top \bs{Q}\bs{Q}^\top \bs{1} \leq N + N \|\bs{Q}\bs{Q}^\top\|_{\oprtnorm} \lesssim N,
	\end{align*}
	which yields 
	\[
	\Var\Big(N^{-1}\sum_i  Z_i Y_i^2\Big) = O(N^{-1}) = o(1).
	\]
	
	Similarly,
	$$
	N^{-1}\sum_i  (1 - Z_i) Y_i^2 = \E \Big\{ N^{-1}\sum_i  ( 1 - Z_i) Y_i^2 \Big\}  + \op(1). 
	$$
	The conclusion follows.
\end{proof}

\subsection{Proof of Theorem \ref{thm:variance-estimator-ind}}
\begin{proof} 
	{\bf Step 1.} We prove (i). We decompose $\E \hat{V}_{\ind}$ as follows:
	\begin{align*}
		\E\hat{V}_{\ind} =  S_1 +  S_2 +  S_3,
	\end{align*}
	where
	\begin{align*}
		S_1 &= \E\sum_{i}  \sum_{j^\prime\ne j} E_{ji}E_{j^\prime i}Y_{j}Y_{j^\prime} \frac{Z_iZ_j Z_{j^\prime}}{N^2r_1^3}  + \E\sum_{i}  \sum_{j^\prime\ne j} E_{ji}E_{j^\prime i}Y_{j}Y_{j^\prime} \frac{Z_i(1-Z_j)(1-Z_{j^\prime})}{N^2r_1^2r_0}\\
		&+ \E\sum_{i}  \sum_{j^\prime\ne j} E_{ji}E_{j^\prime i}Y_{j}Y_{j^\prime} \frac{(1-Z_i)Z_j Z_{j^\prime}}{N^2r_1r_0^2}  + \E\sum_{i}  \sum_{j^\prime\ne j} E_{ji}E_{j^\prime i}Y_{j}Y_{j^\prime} \frac{(1-Z_i)(1-Z_j)(1-Z_{j^\prime})}{N^2r_0^3}\\
		&=: S_{11}+S_{12}+ S_{13}+S_{14},\\
		S_2 & =  \E\sum_{i,j}  E_{ji}Y_{j}^2 \frac{Z_iZ_j}{N^2r_1^2}+ \E\sum_{i,j}  E_{ji}Y_{j}^2 \frac{Z_i(1-Z_j)}{N^2r_1^2}+ \E\sum_{i,j}  E_{ji}Y_{j}^2 \frac{(1-Z_i)Z_j}{N^2r_0^2} \\
		&+ \E\sum_{i,j}  E_{ji}Y_{j}^2 \frac{(1-Z_i)(1-Z_j)}{N^2r_0^2}\\
		&=: S_{21}+S_{22}+ S_{23}+S_{24},\\
		S_3 & =  \E\sum_{i,j}  E_{ji}Y_{j}^2 \frac{r_0Z_iZ_j}{N^2r_1^3}+ \E\sum_{i,j}  E_{ji}Y_{j}^2 \frac{Z_i(1-Z_j)}{N^2r_1r_0}+\E\sum_{i,j}  E_{ji}Y_{j}^2 \frac{(1-Z_i)Z_j}{N^2r_1r_0} \\
		&+ \E\sum_{i,j}  E_{ji}Y_{j}^2 \frac{r_1(1-Z_i)(1-Z_j)}{N^2r_0^3}\\
		&=: S_{31}+S_{32}+ S_{33}+S_{34}.
	\end{align*} 
	Define 
	\begin{align*}
		&\tilde{S}_1 = \frac{1}{N^2}\sum_i  \biggl(\sum_j   E_{ji}Y_{Z_j=1}^{Z_i=1} \biggr)^2,\quad \tilde{S}_2 = \frac{r_0}{N^2r_1}\sum_i  \biggl(\sum_j   E_{ji}Y_{Z_j=0}^{Z_i=1}\biggr)^2\\
		&\tilde{S}_3=\frac{r_1}{N^2r_0}\sum_i  \biggl(\sum_j   E_{ji}Y_{Z_j=1}^{Z_i=0}\biggr)^2,\quad \tilde{S}_4=\frac{1}{N^2}\sum_i  \biggl(\sum_j   E_{ji}Y_{Z_j=0}^{Z_i=0}\biggr)^2,\\
		&\tilde{S}_5 = \frac{1}{N^2}  \sum_{i\ne j} (E_{ij}\theta_{i}+\sum_{k}E_{kj}\tilde{\gamma}_{ki})^2.
	\end{align*}
	
	For $z,z^\prime \in \{0,1\}$ and $i \ne j \ne j^\prime$, simple calculation yields
	\begin{align*}
		&\E Y_{j}Y_{j^\prime} \frac{{I}(Z_i=z,Z_j=Z_{j^\prime} = z^\prime)}{\E {I}(Z_i=z,Z_j=Z_{j^\prime} = z^\prime)}=\E (Y_{j}Y_{j^\prime}\mid Z_i=z,Z_j=Z_{j^\prime} = z^\prime)\\
		=&\E (Y_{j}\mid Z_i=z,Z_j=Z_{j^\prime} = z^\prime)\E (Y_{j^\prime}\mid Z_i=z,Z_j=Z_{j^\prime} = z^\prime) \\
		& + \Cov(Y_{j},Y_{j^\prime}\mid Z_i=z,Z_j=Z_{j^\prime} = z^\prime), \\
		&\E Y_{j}^2 \frac{{I}(Z_i=z,Z_j= z^\prime)}{\E {I}(Z_i=z,Z_j= z^\prime)}  =\E (Y_{j}^2\mid Z_i=z,Z_j= z^\prime) \\
		=&(Y_{Z_j=z^\prime}^{Z_{i}=z})^2+ \Var(Y_{j}\mid Z_i=z,Z_j = z^\prime), \\
		& \Cov(Y_{j},Y_{j^\prime}\mid Z_i=z,Z_j=Z_{j^\prime} = z^\prime) =   \sum_{k:k\ne i}  \tilde{\gamma}_{jk}\tilde{\gamma}_{j^\prime k} r_{1}r_{0}, \\
		& \Var(Y_{j}\mid Z_i=z,Z_j= z^\prime) = \sum_{k:k\ne i}  \tilde{\gamma}_{jk}^2 r_{1}r_{0}.
	\end{align*}

	Moreover, we have
	\begin{align*}
		&\E (Y_{j}\mid Z_i=z,Z_j=Z_{j^\prime} = z^\prime) = \E (Y_{j}\mid Z_i=z,Z_j= z^\prime) + (z^\prime-r_{1}) \tilde{\gamma}_{jj^\prime}.
	\end{align*}
	
	As a consequence, we have
	\begin{align*}
		S_{11} =& \sum_{i}  \sum_{j^\prime\ne j} E_{ji}E_{j^\prime i} \E \Big( Y_{j}Y_{j^\prime} \frac{Z_iZ_j Z_{j^\prime}}{N^2r_1^3} \Big) \\
		=& \frac{1}{N^2}\sum_{i}  \sum_{j^\prime\ne j} E_{ji}E_{j^\prime i} (Y_{Z_j=1}^{Z_i=1}+r_{0}\tilde{\gamma}_{jj^\prime})(Y_{Z_{j^\prime}=1}^{Z_i=1}+r_{0}\tilde{\gamma}_{j^\prime j}) +  \\
		& \frac{1}{N^2}\sum_{i}  \sum_{j^\prime\ne j} \sum_{k:k\ne i} E_{ji}E_{j^\prime i} \tilde{\gamma}_{jk}\tilde{\gamma}_{j^\prime k}r_{1}r_{0}\\
		= & \frac{1}{N^2}\sum_{i}  \sum_{j^\prime\ne j} E_{ji}E_{j^\prime i} (Y_{Z_j=1}^{Z_i=1}Y_{Z_{j^\prime}=1}^{Z_i=1}+r_{0}\tilde{\gamma}_{j^\prime j}Y_{Z_{j}=1}^{Z_i=1}+r_{0}\tilde{\gamma}_{jj^\prime }Y_{Z_{j^\prime}=1}^{Z_i=1}+ r_{0}r_{0} \tilde{\gamma}_{jj^\prime } \tilde{\gamma}_{j^\prime j}) + \\
		&\frac{1}{N^2} \sum_{i}  \sum_{j^\prime\ne j} \sum_{k:k\ne i} E_{ji}E_{j^\prime i} \tilde{\gamma}_{jk}\tilde{\gamma}_{j^\prime k}r_{1}r_{0},
	\end{align*}
	and
	\begin{align*}
		S_{21} = \sum_{i,j}  E_{ji} \E \Big( Y_{j}^2 \frac{Z_iZ_j}{N^2r_1^2} \Big) =  \frac{1}{N^2}\sum_{i,j}  E_{ji}(Y_{Z_j=1}^{Z_i=1})^2 + \frac{1}{N^2} \sum_{i,j} \sum_{k:k\ne i} {E}_{ji}\tilde{\gamma}_{jk}^2r_{1}r_{0}.
	\end{align*}
	Comparing with $\tilde{S}_1$, we have
	\begin{align}
		S_{11}+S_{21} =& \tilde{S}_1 + \frac{1}{N^2}\sum_{i}  \sum_{j^\prime\ne j} E_{ji}E_{j^\prime i} (r_{ 0}\tilde{\gamma}_{j^\prime j}Y_{Z_{j}=1}^{Z_i=1}+r_{ 0}\tilde{\gamma}_{jj^\prime }Y_{Z_{j^\prime}=1}^{Z_i=1} + r_{0}^2\tilde{\gamma}_{jj^\prime } \tilde{\gamma}_{j^\prime j}) + \nonumber\\
		& \frac{1}{N^2}\sum_{i}  \sum_{j^\prime\ne j} \sum_{k:k\ne i} E_{ji}E_{j^\prime i} \tilde{\gamma}_{jk}\tilde{\gamma}_{j^\prime k}r_{1}r_{0} + \frac{1}{N^2}\sum_{i,j} \sum_{k:k\ne i} E_{ji}\tilde{\gamma}_{jk}^2 r_{1}r_{0} \label{eq:S11+S21}.
	\end{align}
	Similarly, we have
	\begin{align}
		S_{12}+S_{22} =& \tilde{S}_2 + \frac{r_0}{N^2r_1}\sum_{i}  \sum_{j^\prime\ne j} E_{ji}E_{j^\prime i} (-r_{1}\tilde{\gamma}_{j^\prime j}Y_{Z_j=0}^{Z_i=1}-r_{ 1}\tilde{\gamma}_{jj^\prime }Y_{Z_{j^\prime}=0}^{Z_i=1} + r_{1}^2 \tilde{\gamma}_{jj^\prime } \tilde{\gamma}_{j^\prime j}) + \nonumber\\
		& \frac{r_0}{N^2r_1}\sum_{i}  \sum_{j^\prime\ne j} \sum_{k:k\ne i} E_{ji}E_{j^\prime i} \tilde{\gamma}_{jk}\tilde{\gamma}_{j^\prime k}r_{1}r_{0} + \frac{r_0}{N^2r_1}\sum_{i,j} \sum_{k:k\ne i} E_{ji}\tilde{\gamma}_{jk}^2 r_{1}r_{0} \label{eq:S12+S22}.
	\end{align}
	\begin{align}
		S_{13}+S_{23} =& \tilde{S}_3 + \frac{r_1}{N^2r_0}\sum_{i}  \sum_{j^\prime\ne j} E_{ji}E_{j^\prime i} (r_{0}\tilde{\gamma}_{j^\prime j}Y_{Z_j=1}^{Z_i=0}+r_{ 0}\tilde{\gamma}_{jj^\prime }Y_{Z_{j^\prime}=1}^{Z_i=0} + r_{0}^2 \tilde{\gamma}_{jj^\prime } \tilde{\gamma}_{j^\prime j}) + \nonumber\\
		& \frac{r_1}{N^2r_0}\sum_{i}  \sum_{j^\prime\ne j} \sum_{k:k\ne i} E_{ji}E_{j^\prime i} \tilde{\gamma}_{jk}\tilde{\gamma}_{j^\prime k}r_{1}r_{0} + \frac{r_1}{N^2r_0}\sum_{i,j} \sum_{k:k\ne i} E_{ji}\tilde{\gamma}_{jk}^2 r_{1}r_{0} \label{eq:S13+S23}.
	\end{align}
	\begin{align}
		S_{14}+S_{24} = & \tilde{S}_4 + \frac{1}{N^2}\sum_{i}  \sum_{j^\prime\ne j} E_{ji}E_{j^\prime i} (-r_{1}\tilde{\gamma}_{j^\prime j}Y_{Z_j=0}^{Z_i=0}-r_{ 1}\tilde{\gamma}_{jj^\prime }Y_{Z_{j^\prime}=0}^{Z_i=0} + r_{1}^2 \tilde{\gamma}_{jj^\prime } \tilde{\gamma}_{j^\prime j}) + \nonumber\\
		& \frac{1}{N^2}\sum_{i}  \sum_{j^\prime\ne j} \sum_{k:k\ne i} E_{ji}E_{j^\prime i} \tilde{\gamma}_{jk}\tilde{\gamma}_{j^\prime k}r_{1}r_{0} + \frac{1}{N^2}\sum_{i,j} \sum_{k:k\ne i} E_{ji}\tilde{\gamma}_{jk}^2 r_{1}r_{0} \label{eq:S14+S24}.
	\end{align}
	Summating \eqref{eq:S11+S21}--\eqref{eq:S14+S24} and using $Y_{Z_j=1}^{Z_i=z}-Y_{Z_j=0}^{Z_i=z} = \theta_j$ for $z\in\{0,1\}$, we have
	\begin{align}
		\label{eq:S1+S2}
		S_1 + S_2 =& \sum_{q=1}^4 \tilde{S}_{q} + \frac{1}{N^2}\sum_{i}  \sum_{j^\prime\ne j} E_{ji}E_{j^\prime i} (\tilde{\gamma}_{j^\prime j}\theta_j + \tilde{\gamma}_{jj^\prime }\theta_{j^\prime} + \tilde{\gamma}_{jj^\prime } \tilde{\gamma}_{j^\prime j}) + \nonumber\\
		& \frac{1}{N^2}\sum_{i}  \sum_{j^\prime\ne j} \sum_{k:k\ne i} E_{ji}E_{j^\prime i} \tilde{\gamma}_{jk}\tilde{\gamma}_{j^\prime k} + \frac{1}{N^2}\sum_{i,j} \sum_{k:k\ne i} E_{ji}\tilde{\gamma}_{jk}^2.
	\end{align}
	On the other hand, we can verify that
	\begin{align}
		\label{eq:second-term-of-var-ind}
		\tilde{S}_5 = & \frac{1}{N^2}\sum_{i,j}  \theta_j^2 E_{ji} + \frac{2}{N^2}\sum_{i,j}  \sum_{k:k\ne i} E_{ji} E_{ki}\tilde{\gamma}_{kj} \theta_j + \nonumber \frac{1}{N^2}\sum_{i\ne j} \sum_{k}  E_{ki}\tilde{\gamma}_{kj}^2 + \\
		&\frac{1}{N^2} \sum_{i\ne j}\sum_{k^\prime\ne k}  E_{ki} {E}_{ k^\prime i} \tilde{\gamma}_{kj} \tilde{\gamma}_{k^\prime j}.
	\end{align}
	Substituting \eqref{eq:second-term-of-var-ind} into \eqref{eq:S1+S2}, we obtain
	\begin{align}
		\label{eq:S1+S2-final}
		S_1 + S_2 &= \sum_{q=1}^5 \tilde{S}_{q} -\frac{1}{N^2}\sum_{i,j}  \theta_j^2 E_{ji}+\frac{1}{N^2}\sum_{i}  \sum_{j^\prime\ne j} E_{ji}E_{j^\prime i} \tilde{\gamma}_{jj^\prime } \tilde{\gamma}_{j^\prime j}.
	\end{align}
	
	Noting that, for $z\in\{0,1\}$,
	\begin{align*}
		r_1(Y_{Z_j=0}^{Z_i=z})^2+r_0(Y_{Z_j=1}^{Z_i=z})^2 &=  (r_1Y_{Z_j=0}^{Z_i=z}+r_0Y_{Z_j=1}^{Z_i=z})^2 + r_1r_0(Y_{Z_j=0}^{Z_i=z}-Y_{Z_j=1}^{Z_i=z})^2\\
		&= (r_1Y_{Z_j=0}^{Z_i=z}+r_0Y_{Z_j=1}^{Z_i=z})^2 + r_1r_0\theta_j^2,
	\end{align*}
	we have
	\begin{align*}
		S_{31} + S_{33} &=\frac{1}{r_1N^2}\sum_{i,j}  E_{ji}  \left\{ r_1(Y_{Z_j=0}^{Z_i=z})^2+r_0(Y_{Z_j=1}^{Z_i=1})^2 \right\} +\frac{1}{r_0N^2}  \sum_{i,j}  E_{ji} \sum_{k:k\ne i} \tilde{\gamma}_{jk}^2 \\
		&=\frac{1}{N^2r_1} \sum_{i,j}  E_{ji}\left\{r_1r_0 \theta_j^2 + (r_0Y_{Z_j=0}^{Z_i=1}+r_1Y_{Z_j=1}^{Z_i=1})^2\right\}+r_0 \sum_{i,j}  E_{ji} \sum_{k:k\ne i} \tilde{\gamma}_{jk}^2,\\
		S_{32} + S_{34} &=\frac{1}{r_0N^2}\sum_{i,j}  E_{ji}  \left\{ r_1(Y_{Z_j=0}^{Z_i=0})^2+r_0(Y_{Z_j=1}^{0})^2 \right\} +\frac{1}{r_1N^2} \sum_{i,j}  E_{ji} \sum_{k:k\ne i} \tilde{\gamma}_{jk}^2 \\
		&=\frac{1}{r_0N^2} \sum_{i,j}  E_{ji}\left\{r_1r_0 \theta_j^2 + (r_0Y_{Z_j=0}^{Z_i=0}+r_1Y_{Z_j=1}^{Z_i=0})^2\right\}+\frac{1}{r_1N^2} \sum_{i,j}  E_{ji} \sum_{k:k\ne i} \tilde{\gamma}_{jk}^2.
	\end{align*}
	Simple calculation yields
	\begin{align*}
		&r_0(r_1Y_{Z_j=0}^{Z_i=1}+r_0Y_{Z_j=1}^{Z_i=1})^2 + r_1(r_1Y_{Z_j=0}^{Z_i=0}+r_0Y_{Z_j=1}^{Z_i=0})^2 \\
		=&  (r_0r_1Y_{Z_j=0}^{Z_i=1}+r_0^2Y_{Z_j=1}^{Z_i=1} + r_1^2Y_{Z_j=0}^{Z_i=0}+r_0r_1Y_{Z_j=1}^{Z_i=0})^2  +r_1r_0 (r_1Y_{Z_j=0}^{Z_i=1}+r_0Y_{Z_j=1}^{Z_i=1}- \\
		&r_1Y_{Z_j=0}^{Z_i=0}-r_0Y_{Z_j=1}^{Z_i=0})^2\\
		=& (r_0r_1Y_{Z_j=0}^{Z_i=1}+r_0^2Y_{Z_j=1}^{Z_i=1} + r_1^2Y_{Z_j=0}^{Z_i=0}+r_0r_1Y_{Z_j=1}^{Z_i=0})^2 +r_1r_0\gamma_{ji}^2.
	\end{align*}
	Combing these, we have
	\begin{align}
		\label{eq:S3-final}
		S_{3} =& \frac{1}{N^2}\sum_{i,j}  E_{ji}\left\{\theta_j^2 + \gamma_{ji}^2 + (r_0r_1Y_{Z_j=0}^{Z_i=1}+r_0^2Y_{Z_j=1}^{Z_i=1} + r_1^2Y_{Z_j=0}^{Z_i=0}+r_0r_1Y_{Z_j=1}^{Z_i=0})^2/(r_1r_0) \right\}\nonumber\\
		&+\frac{1}{N^2}\sum_{i,j}  E_{ji} \sum_{k:k\ne i} \tilde{\gamma}_{jk}^2\nonumber\\
		=&\frac{1}{r_1r_0N^2}\sum_{i,j}  E_{ji}(r_0r_1Y_{Z_j=0}^{Z_i=1}+r_0^2Y_{Z_j=1}^{Z_i=1}  r_1^2Y_{Z_j=0}^{Z_i=0}+r_0r_1Y_{Z_j=1}^{Z_i=0})^2+\nonumber\\
		&+\frac{1}{N^2}\sum_{i,j}  E_{ji} \big(\theta_j^2+\sum_{k} \tilde{\gamma}_{jk}^2\big).
	\end{align}
	Combining \eqref{eq:S1+S2-final} and \eqref{eq:S3-final}, we have
	\begin{align}
		\label{eq:E-V_ind}
		\E\hat{V}_{\ind} &= \sum_{q=1}^5\tilde{S}_q + \mathcal{B}_{\ind},
	\end{align}
	where 
	\begin{align*}
		\mathcal{B}_{\ind} =& \frac{1}{r_1r_0N^2}\sum_{i,j}  E_{ji}(r_0r_1Y_{Z_j=0}^{Z_i=1}+r_0^2Y_{Z_j=1}^{Z_i=1} + r_1^2Y_{Z_j=0}^{Z_i=0}+r_0r_1Y_{Z_j=1}^{Z_i=0})^2 +\\
		&\frac{1}{N^2}\sum_{i,j}  E_{ji} \sum_{k} \tilde{\gamma}_{jk}^2 +\frac{1}{N^2}\sum_{i}  \sum_{j^\prime\ne j} E_{ji}E_{j^\prime i} \tilde{\gamma}_{jj^\prime } \tilde{\gamma}_{j^\prime j}.
	\end{align*}
	
	Note that
	\begin{align*}
		\Big|\sum_{i}  \sum_{j^\prime\ne j} (E_{ji}\tilde{\gamma}_{jj^\prime })(E_{j^\prime i}\tilde{\gamma}_{j^\prime j})\Big| &\leq \Big\{\sum_{i}  \sum_{j^\prime\ne j} (E_{ji}\tilde{\gamma}_{jj^\prime })^2 \Big\}^{1/2}\Big\{ \sum_{i}  \sum_{j^\prime\ne j} (E_{j^\prime i}\tilde{\gamma}_{j^\prime j})^2\Big\}^{1/2}\\
		&=\sum_{i}  \sum_{j^\prime\ne j} E_{ji} \tilde{\gamma}_{jj^\prime }^2 = \sum_{i,j}  E_{ji} \sum_{k} \tilde{\gamma}_{jk}^2.
	\end{align*}
	As a consequence, we have 
	\begin{equation}\label{eqn:EVlarger}
		\E\hat{V}_{\ind} \geq \sum_{q=1}^5\tilde{S}_q.
	\end{equation}
	Using the equality $r_1a_i^2 + r_0 b_i^2 = (r_1 a_i + r_0 b_i)^2 + r_1r_0(a_i-b_i)^2$, we have
	\begin{align*}
		&\sum_{q=1}^4 \tilde{S}_q  \\
		=&\frac{1}{r_1N^2}\sum_i  \Bigl\{\sum_j   E_{ji}(r_1 Y_{Z_j=1}^{Z_i=1} +r_0Y_{Z_j=0}^{Z_i=1})\Bigr\}^2 + \frac{1}{N^2r_0}\sum_i  \Bigl\{\sum_j   E_{ji}(r_1Y_{Z_j=1}^{Z_i=0}+r_0Y_{Z_j=0}^{Z_i=0})\Bigr\}^2  \\
		& +\frac{r_1r_0}{r_1N^2} \sum_i  \Bigl\{\sum_j   E_{ji}( Y_{Z_j=1}^{Z_i=1} -Y_{Z_j=0}^{Z_i=1})\Bigr\}^2 + \frac{r_1r_0}{r_0N^2} \sum_i  \Bigl\{\sum_j   E_{ji}( Y_{Z_j=1}^{Z_i=0} -Y_{Z_j=0}^{Z_i=0})\Bigr\}^2\\
		= &\frac{1}{r_1N^2}\sum_i  \Bigl\{\sum_j   E_{ji}(r_1 Y_{Z_j=1}^{Z_i=1} +r_0Y_{Z_j=0}^{Z_i=1})\Bigr\}^2 + \frac{1}{N^2r_0}\sum_i  \Bigl\{\sum_j   E_{ji}(r_1Y_{Z_j=1}^{Z_i=0}+r_0Y_{Z_j=0}^{Z_i=0})\Bigr\}^2 \\
		& + \frac{1}{N^2} \sum_i  \Bigl\{\sum_j   E_{ji}\theta_j\Bigr\}^2\\
		= &\frac{1}{N^2r_1r_0}\sum_j  \Bigl\{\sum_i   E_{ij}(r_1r_0 Y_{Z_i=1}^{Z_j=1} +r_0^2Y_{Z_i=0}^{Z_j=1}+r_1^2Y_{Z_i=1}^{Z_j=0}+r_1r_0Y_{Z_i=0}^{Z_j=0})\Bigr\}^2 \\
		& +\frac{1}{N^2} \sum_i  \Bigl(\sum_j   E_{ji}\theta_j\Bigr)^2 + \frac{1}{N^2} \sum_i  \Bigl(\sum_j   E_{ji}\gamma_{ji}\Bigr)^2.  
	\end{align*}
	As a consequence, we have
	\begin{align}
		\label{eq:two-cauchy-schwarz-inequality}
		\sum_{q=1}^4 \tilde{S}_q = \sigma^2_{\ind,1} + \frac{1}{N^2} \sum_i  \Bigl(\sum_j   E_{ji}\theta_j\Bigr)^2 + \frac{1}{N^2} \sum_i  \Bigl(\sum_j   E_{ji}\gamma_{ji}\Bigr)^2.
	\end{align}
	
	Comparing with the formula of $\Var(\hat{\tau}_{\ind})$ and by \eqref{eqn:EVlarger}, we have
	\begin{align}
		\E \hat{V}_{\ind}-\Var(\hat{\tau}_{\ind})  &\geq \sum_{q=1}^5 \tilde{S}_q -\Var(\hat{\tau}_{\ind}) \geq  \Big(\sum_{q=1}^4 \tilde{S}_q-\sigma_{\ind,1}^2\Big) + (\tilde{S}_5-\sigma_{\ind,2}^2) \nonumber\\
		& \geq -\sum_{i\ne j}\frac{1}{N^2} \Big(\theta_{j}E_{ji}+\sum_{k}E_{ki}\tilde{\gamma}_{kj}\Big)\Big(\theta_{i}E_{ij}+\sum_{k}E_{kj}\tilde{\gamma}_{ki}\Big).\label{eq:hat-V-ind-bias}
	\end{align}
	Therefore, 
	\begin{align*}
		\E \hat{V}_{\ind} &\geq \Var(\hat{V}_{\ind}) - \sum_{i\ne j}\frac{1}{N^2} \Big(\theta_{j}E_{ji}+\sum_{k}E_{ki}\tilde{\gamma}_{kj}\Big)\Big(\theta_{i}E_{ij}+\sum_{k}E_{kj}\tilde{\gamma}_{ki}\Big)\\
		& \geq \sum_{i\ne j}\frac{1}{N^2} \Big(\theta_{j}E_{ji}+\sum_{k}E_{ki}\tilde{\gamma}_{kj}\Big)^2\\
		& \geq \sum_{i\ne j}\frac{1}{N^2} \Big(\theta_{j}E_{ji}+\sum_{k}E_{ki}\tilde{\gamma}_{kj}\Big)\Big(\theta_{i}E_{ij}+\sum_{k}E_{kj}\tilde{\gamma}_{ki}\Big).
	\end{align*}
	It follows that $\E (2\hat{V}_{\ind}) -\Var(\hat{\tau}_{\ind}) \geq 0$.

	Note that, by \Cref{a:bounded-parameter} and \Cref{a:opnorm-EE^T}, we have
	\begin{align*}
		&\frac{1}{r_1r_0N^2}\sum_{i,j}  E_{ji}(r_0r_1Y_{Z_j=0}^{Z_i=1}+r_0^2Y_{Z_j=1}^{Z_i=1} + r_1^2Y_{Z_j=0}^{Z_i=0}+r_0r_1Y_{Z_j=1}^{Z_i=0})^2 \\
		\lesssim & \frac{1}{N^2}\sum_{i,j}  E_{ji} = O(\rhon), \\
		&\frac{1}{N^2}\sum_{i,j}  E_{ji} \sum_{k} \tilde{\gamma}_{jk}^2+\frac{1}{N^2}\sum_{i}  \sum_{j^\prime\ne j} E_{ji}E_{j^\prime i} \tilde{\gamma}_{jj^\prime } \tilde{\gamma}_{j^\prime j} \lesssim \frac{1}{N^2}\sum_{i,j}  E_{ji} \sum_{k} \tilde{\gamma}_{jk}^2\\
		\lesssim & \frac{1}{N^2}\sum_{i,j}  E_{ji} \sum_{k} Q_{jk}^2 \lesssim \frac{1}{N^2}\sum_{i,j} \sum_{k}  E_{ji}  Q_{jk} \\
		=& \frac{1}{N^2} \bs{1}^\top \bs{E}^\top \bs{Q}\bs{1} = O(\rhon).
	\end{align*}
	If follows that
	\[
	\mathcal{B}_{\ind} = O(\rhon),
	\]
	and
	$$\E\hat{V}_{\ind} = \sum_{q=1}^5\tilde{S}_q + O(\rhon).$$
	
	Since
	\begin{align*}
		&\frac{1}{N^2} \sum_i  \Bigl(\sum_j   E_{ji}\gamma_{ji}\Bigr)^2  \lesssim \frac{1}{N^2} \sum_i  \Bigl(\sum_j   Q_{ji}\Bigr)^2 = O(N^{-1}),\quad \tilde{S}_5 \lesssim \sigma^2_{\ind,2} = O(\rhon).
	\end{align*}
	then, we have
	\begin{align*}
		\E \hat{V}_{\ind} =& \sum_{q=1}^5 \tilde{S}_q + O(\rhon) = \sigma^2_{\ind,1} +\frac{1}{N^2} \sum_i  \Bigl(\sum_j   E_{ji}\theta_j\Bigr)^2 + O(\rhon)\\
		=& \Var(\hat{\tau}_{\ind}) + \frac{1}{N^2} \sum_i  \Bigl(\sum_j   E_{ji}\theta_j\Bigr)^2+O(\rhon). 
	\end{align*}

	When $N\rhon \rightarrow \infty$, we have $\rhon = o(N\rhon^2)$. Therefore, 
	$$
	\E \hat{V}_{\ind} - \Var(\hat{\tau}_{\ind}) = \frac{1}{N^2} \sum_i  \Bigl(\sum_j   E_{ji}\theta_j\Bigr)^2 + o(N\rhon).
	$$

	{\bf Step 2.} We prove (ii). Recall that
	\[
	\hat{\tau}^\prime_{\ind}  =  \frac{1}{N}\sum_{i,j} E_{ij} \Bigl\{\frac{Y_i^\prime Z_j}{r_1} - \frac{Y_i^\prime (1-Z_j)}{r_0}\Bigr\},
	\]
	where $Y_i^\prime = \alpha_i^\prime + \theta_i^\prime Z_i +\sum_j  \tilde{E}_{ij} {\gamma}_{ij}^\prime Z_j$ with $\tilde{\gamma}^\prime_{ij} = \tilde{E}_{ij} {\gamma}_{ij}^\prime$. We will prove a general result that encompasses \Cref{thm:variance-estimator-ind}(ii) as a special case. Define $\hat{V}_{\ind}^\prime$ similarly as $\hat{V}_{\ind}$ with $Y_i$ replaced by $Y_i^\prime$.
	
	\begin{theorem}
		\label{thm:consistency-of-oracle-variance-estimator-general-ind}
		Suppose that \Cref{a:bounded-parameter} holds for both $\{\alpha_i,\theta_i,{\gamma}_{ij}\}_{1 \leq i,j\leq N}$ and\\ $\{\alpha_i^\prime,\theta_i^\prime,{\gamma}_{ij}^\prime\}_{1 \leq i,j\leq N}$,  Assumptions~\ref{a:density-rho_N}--\ref{a:Lindberg-condition-unadj} hold, and for a nonrandom sequence $\Lambdan$,
		\begin{align*}
			& N^{-1} \max\Big\{\sum_i \Big(\sum_j  E_{ji}Y_{Z_j=1}^\prime\Big)^2, \sum_i \Big(\sum_j  E_{ji}Y_{Z_j=0}^\prime\Big)^2\Big\} = O(\Lambdan ),\\
			&N^{-1} \max\Big\{\maxi\Big(\sum_j  E_{ji}Y_{Z_j=1}^\prime\Big)^2, \maxi\Big(\sum_j  E_{ji}Y_{Z_j=0}^\prime\Big)^2\Big\} = o(\Lambdan ).
		\end{align*}
		Then, we have $\hat{V}_{\ind}^\prime-\operatorname{E} (\hat{V}_{\ind}^\prime)  = \op(N^{-1}\Lambdan + \rhon)$.
	\end{theorem}
	\Cref{thm:variance-estimator-ind}(ii) follows from \Cref{thm:consistency-of-oracle-variance-estimator-general-ind} with  $\Lambdan \equiv N^2\rhon^2$ and $Y_i^\prime \equiv Y_i$. Thus, in the remainder of this section, we will prove \Cref{thm:consistency-of-oracle-variance-estimator-general-ind}. 
	
\end{proof}

\begin{proof}[Proof of \Cref{thm:consistency-of-oracle-variance-estimator-general-ind}]
	We will only show that
	\[
	N^{-1}\sum_i \biggl(\sum_j  E_{ji}Y_{j}^\prime Z_j\biggr)^2 Z_i = \E N^{-1}\sum_i \biggl(\sum_j  E_{ji}Y_{j}^\prime Z_j\biggr)^2 Z_i + \op(\Lambdan + N\rhon).
	\]
	The proofs for the convergence of other components in $\hat{V}_{\ind}^\prime$ follow similar steps.

	Simple calculation yields
	\begin{align*}
		\sum_i \biggl(\sum_j  E_{ji}Y_{j}^\prime Z_j\biggr)^2 Z_i = \sum_{i\ne j} E_{ji}Y_{j}^{\prime 2} Z_j Z_i+ \sum_{i\ne j \ne j^\prime} E_{ji}E_{j^\prime i}Y_{j}^\prime Y_{j^\prime}^{\prime}Z_jZ_{j^\prime} Z_i.
	\end{align*}
	Define $f(\bs{Z}): = f(Z_1,\ldots,Z_n) : =  \sum_{i\ne j} E_{ji}Y_{j}^{\prime 2} Z_j Z_i$. Let $\bs{z} = (z_1,\ldots,z_k=1,\ldots,z_n)= (z_1,\ldots,z_n)$, $\bs{z}^\prime = (z_1,\ldots,z_k=0,\ldots,z_n) = (z_1^\prime,\ldots,z_n^\prime)$, we have
	\begin{align*}
		|f(\bs{z})-f(\bs{z}^\prime)| = & \Big|\sum_{i\ne j} E_{ji}Y_{j}^\prime (\bs{z})^2 z_j z_i-\sum_{i\ne j} E_{ji}Y_{j}^\prime (\bs{z}^\prime)^2 z_j^\prime z_i^\prime\Big|\\
		\leq & \Big|\sum_{i\ne j} E_{ji}\{Y_{j}^\prime (\bs{z})^2-Y_{j}^\prime (\bs{z}^\prime)^2\} z_j z_iI(i\ne j \ne k)\Big| +\\
		&\Big|\sum_i  E_{ki}\{Y_{k}(\bs{z})^2-Y_{k}(\bs{z}^\prime)^2\} z_i\Big| +  \Big|\sum_j  E_{jk} \{Y_j(\bs{z})^2-Y_j(\bs{z}^\prime)^2\}z_j\Big|\\
		\lesssim & \sum_i   \sum_{j} E_{ji}Q_{jk} + \sum_i  E_{ki} + \sum_j  E_{jk}.
	\end{align*}

	By Lemma~\ref{lem:efron-stein} and \Cref{a:opnorm-EE^T}, we have
	\begin{align*}
		\Var(f(\bs{Z}))\leq & \sum_k \Big(\sum_i   \sum_{j} E_{ji}Q_{jk} + \sum_i  E_{ki} + \sum_j  E_{jk}\Big)^2\\
		\lesssim&\sum_k \Big(\sum_i   \sum_{j} E_{ji}Q_{jk}\Big)^2 + \sum_k \Big(\sum_i  E_{ki}\Big)^2 + \sum_k \Big(\sum_j  E_{jk}\Big)^2\\
		=&\bs{1}^\top\bs{E}^\top\bs{Q}\bs{Q}^\top\bs{E}\bs{1} + \bs{1}^\top\bs{E}^\top\bs{E}\bs{1} + \bs{1}^\top\bs{E}\bs{E}^\top\bs{1} \lesssim N^3\rhon^2.
	\end{align*}
	
	As a consequence, we have
	\[
	N^{-1}\sum_{i,j} E_{ji}Y_{j}^{\prime 2} Z_j Z_i-N^{-1}\E\sum_{i,j} E_{ji}Y_{j}^{\prime 2} Z_j Z_i =  \Op((N\rhon^2)^{1/2}) = \op(N\rhon).
	\]
	
	Let $\tilde{Z}_i = Z_i-r_1$. Then, we have
	\begin{align*}
		&N^{-1}\sum_{i\ne j \ne j^\prime} E_{ji}E_{j^\prime i}Y_{j}^\prime Y_{j^\prime}^\prime Z_jZ_{j^\prime} Z_i\\
		=& N^{-1}\sum_{i\ne j \ne j^\prime} E_{ji}E_{j^\prime i}\Big(Y_{Z_j=1}^\prime + \sum_{k}\tilde{Z}_k\tilde{\gamma}^\prime_{jk}\Big)\Big(Y_{Z_{j^\prime}=1}^\prime  + \sum_{k}\tilde{Z}_k\tilde{\gamma}^\prime_{j^\prime k}\Big)Z_jZ_{j^\prime} Z_i\\
		=&N^{-1}\sum_{i\ne j \ne j^\prime} E_{ji}E_{j^\prime i}\Big(Y_{Z_j=1}^\prime + r_0 \tilde{\gamma}^\prime_{ji} + r_0 \tilde{\gamma}^\prime_{jj^\prime} + \sum_{k:k\notin \{i,j^\prime\}}\tilde{Z}_k\tilde{\gamma}^\prime_{jk}\Big)\\
		& \qquad \qquad \qquad \Big(Y_{Z_{j^\prime}=1}^\prime  + r_0 \tilde{\gamma}^\prime_{j^\prime i} +  r_0 \tilde{\gamma}^\prime_{j^\prime j}  + \sum_{k:k\notin \{i,j\}}\tilde{Z}_k\tilde{\gamma}^\prime_{j^\prime k}\Big)Z_jZ_{j^\prime} Z_i.
	\end{align*}
	Let $F_{ji} = E_{ji}(Y_{Z_j=1}^\prime + r_0 \tilde{\gamma}^\prime_{ji})$ and we decompose 
	\begin{align*}
		N^{-1}\sum_{i\ne j \ne j^\prime} E_{ji}E_{j^\prime i}Y_{j}^\prime Y_{j^\prime}^\prime Z_jZ_{j^\prime} Z_i = \sum_{q=1}^6 S_q,
	\end{align*} 
	where
	\begin{align*}
		&S_1 = N^{-1}\sum_{i\ne j \ne j^\prime} F_{ji}F_{j^\prime i} Z_jZ_{j^\prime} Z_i, \\
		&S_2 = N^{-1}\sum_{i\ne j \ne j^\prime} \Big\{E_{ji}E_{j^\prime i}(Y_{Z_j=1}^\prime  + r_0 \tilde{\gamma}^\prime_{ji} + r_0 \tilde{\gamma}^\prime_{jj^\prime})(Y_{Z_{j^\prime}=1}^\prime  + r_0 \tilde{\gamma}^\prime_{j^\prime i} +  r_0 \tilde{\gamma}^\prime_{j^\prime j})\\
		& \qquad \qquad \qquad \qquad -F_{ji}F_{j^\prime i}\Big\}Z_jZ_{j^\prime} Z_i,\\
		&S_3 = 2N^{-1}\sum_{i\ne j \ne j^\prime} F_{ji}E_{j^\prime i}(\sum_{k:k\ne i}\tilde{Z}_k\tilde{\gamma}^\prime_{j^\prime k}) Z_jZ_{j^\prime} Z_i,\\
		&S_4 = 2N^{-1}\sum_{i\ne j \ne j^\prime} E_{ji}E_{j^\prime i}r_0 \tilde{\gamma}^\prime_{jj^\prime}(\sum_{k:k\ne i}\tilde{Z}_k\tilde{\gamma}^\prime_{j^\prime k}) Z_jZ_{j^\prime} Z_i,\\
		&S_5 = N^{-1}\sum_{i\ne j\ne j^\prime \ne k } E_{ji}E_{j^\prime i}\tilde{\gamma}^\prime_{jk}\tilde{\gamma}^\prime_{j^\prime k}Z_jZ_{j^\prime} Z_i \tilde{Z}_k^2, \\
		&S_6 = N^{-1}\sum_{i\ne j\ne j^\prime \ne k \ne k^\prime} E_{ji}E_{j^\prime i} \tilde{\gamma}^\prime_{jk}\tilde{\gamma}^\prime_{j^\prime k^\prime} \tilde{Z}_k\tilde{Z}_{k^\prime}Z_jZ_{j^\prime} Z_i .
	\end{align*}
	
	Next, we will show that, for $1\leq q \leq 6$,
	\[
	\Var(S_q) = o(\Lambdan^2 + N^2\rhon^2).
	\]
	
	\textbf{Step 1}: For $S_{1}$, we have
	\begin{align*}
		\Var(S_{1}) &= N^{-2} \sum_{i_1\ne j_1 \ne j^\prime_1} \sum_{i_2\ne j_2 \ne j^\prime_2}  F_{j_1i_1}F_{j^\prime_1 i_1} F_{j_2i_2}F_{j^\prime_2 i_2}\Cov( Z_{j_1}Z_{j^\prime_1} Z_{i_1},  Z_{j_2}Z_{j^\prime_2} Z_{i_2}) \\
		&=:C_1 S_{11} + C_2 S_{12} +C_3 S_{13} + C_4 S_{14} + C_5 S_{15} + C_6 S_{16} + C_7 S_{17} +  C_8 S_{18} +  S_{19},
	\end{align*}
	where 
	\begin{align*}
		& S_{11} = N^{-2} \sum_{i\ne j\ne k} F_{ji}^2F_{ki}^2,\quad S_{12} = N^{-2} \sum_{i\ne j\ne k}F_{ji}F_{ki}F_{ij}F_{kj}, \\
		& S_{13} = N^{-2}\sum_{i\ne j\ne k\ne m} F_{ji}F_{ki}F_{jm}F_{km},\quad S_{14} = N^{-2}\sum_{i\ne j\ne k\ne m}F_{ji}F_{ki}F_{jm}F_{im},\\
		&S_{15} = N^{-2}\sum_{i\ne j\ne k\ne m} F_{ji}^2F_{ki}F_{mi},\quad S_{16} = N^{-2}\sum_{i\ne j\ne k\ne m} F_{ji}F_{ki}F_{ij}F_{mj} ,\\
		& S_{17} = N^{-2}\sum_{i\ne j\ne k\ne m\ne l} F_{ji}F_{ki}F_{mi}F_{li},\quad S_{18} =   N^{-2}\sum_{i\ne j\ne k\ne m\ne l} N^{-2} F_{ji}F_{ki}F_{mj}F_{lj},\\
		&S_{19} = N^{-2}\sum_{i\ne j\ne k\ne m\ne l} F_{ji}F_{ki} F_{kl}F_{ml}.
	\end{align*}
	We ignore the values of $C_q$, $1 \leq q \leq 9$, which are all $O(1)$. 
	
	We handle these terms $S_{1q}$, $1 \leq q \leq 9$, one by one. Note that
	\[
	F_{ij} \leq CE_{ij},
	\]
	for some large enough constant $C$. It follows that
	\begin{align*}
		\max\{S_{11},S_{12}\} & \lesssim N^{-2} \sum_{i\ne j\ne k} E_{ji}E_{ki} = N^{-2}\bs{1}^\top\bs{E}\bs{E}^\top\bs{1} = O(N\rhon^2),\\
		S_{13} & \lesssim N^{-2}\sum_{i\ne j\ne k\ne m} E_{ji}E_{ki}E_{jm} \lesssim N^{-2}\bs{1}^\top\bs{E}^\top\bs{E}\bs{E}^\top\bs{1}\lesssim O(N^2\rhon^3).
	\end{align*}
	
	To deal with $S_{14}$, let $\xi_i  = \sum_j   F_{ji}$. Recall that $\delta_i = \sum_j Q_{ji}$. Then,
	\begin{align*}
		\max_i \xi_i^2 \lesssim \max_i \Big\{\big(\delta_i\big)^2,\big(\sum_j   E_{ji}Y_{Z_j=1}^\prime \big)^2\Big\} = o(N + N\Lambdan),
	\end{align*}
	which implies that $\max_i |\xi_i| =  o(N^{1/2} + N^{1/2}\Lambdan^{1/2})$.
	
	We make the following decomposition:
	\begin{align*}
		S_{14} &= N^{-2}\sum_{i\ne j\ne m}F_{ji}\xi_i F_{jm}F_{im}- \sum_{i\ne j\ne m}F_{ji}F_{ji}F_{jm}F_{im}-\sum_{i\ne j\ne m}F_{ji}F_{mi}F_{jm}F_{im} \\
		&=:  S_{141} + S_{142} + S_{143}.
	\end{align*}
	
	We handle these three terms separately. For $S_{141}$, we have
	\begin{align*}
		S_{141} &\lesssim N^{-2} \max_i|\xi_i| \sum_{i\ne j\ne m}E_{ji} E_{jm} = N^{-2}\max_i|\xi_i|\bs{1}^\top\bs{E}^\top\bs{E}\bs{1} \\
		&= N\rhon^{2}o(N^{1/2} + N^{1/2}\Lambdan^{1/2}). 
	\end{align*}
	By Young's inequality, for any $0< \delta < 2 $, 
	$$\Lambdan^\delta (N\rhon)^{2-\delta} = O\Big(\Lambdan^2+N^2\rhon^2\Big),$$
	then, we have
	$$N^{3/2}\Lambdan^{1/2}\rhon^{2} = \Lambdan^{1/2}\cdot N^{3/2}\rhon^{3/2} \cdot \rhon^{1/2} = O\Big(\Lambdan^2+N^2\rhon^2\Big)\cdot o(1) = o\Big(\Lambdan^2+N^2\rhon^2\Big).$$
	If follows that $$S_{141} = o\Big(\Lambdan^2+N^2\rhon^2\Big).$$
	
	For the other two terms, we have
	\begin{align*}
		\max\bigg\{|S_{142}|,|S_{143}|\bigg\}\lesssim N^{-2}\sum_{i\ne j\ne m}E_{ji}E_{jm}E_{im} = N^{-2}\sum_{i\ne j\ne m}E_{ji}E_{jm} = O(N\rhon^2).
	\end{align*}
	As a consequence, we have $$S_{14} = o\Big(\Lambdan^2+N^2\rhon^2\Big).$$ 
	
	Similarly, we have
	\begin{align*}
		S_{15} &=N^{-2}\sum_{i\ne j\ne m} F_{ji}^2\xi_i F_{mi}-N^{-2}\sum_{i\ne j\ne m} F_{ji}^2F_{ji}F_{mi}-N^{-2}\sum_{i\ne j\ne m} F_{ji}^2F_{mi}F_{mi} \\
		&\lesssim N^{-2}(\max_i |\xi_i|)\sum_{i\ne j\ne m} E_{ji}^2 E_{mi}+ N^{-2}\sum_{i\ne j\ne m} E_{ji}^2 E_{mi} =o\Big(\Lambdan^2+N^2\rhon^2\Big),\\
		S_{16}&=N^{-2}\sum_{i\ne j\ne k} F_{ji}F_{ki}F_{ij}\xi_j-N^{-2}\sum_{i\ne j\ne k} F_{ji}F_{ki}F_{ij}F_{ij}-N^{-2}\sum_{i\ne j\ne k} F_{ji}F_{ki}F_{ij}F_{kj} \\
		& \lesssim \max_i |\xi_i| N^{-2}\sum_{i\ne j\ne k} E_{ji}E_{ki}E_{ij} +  N^{-2}\sum_{i\ne j\ne k} E_{ji}E_{ki}E_{ij} = o\Big(\Lambdan^2+N^2\rhon^2\Big). 
	\end{align*}
	
	For $S_{17}$, we make the following decomposition:
	\begin{align*}
		S_{17} & =N^{-2}\sum_{i\ne j\ne k\ne m} F_{ji}F_{ki} F_{mi}\xi_i-3N^{-2}\sum_{i\ne j\ne k\ne m} F_{ji}F_{ki} F_{mi}F_{ji}\\
		& =:  S_{171} - 3S_{172}  + \text{remainder terms},
	\end{align*}
	where 
	\[
	S_{171} = N^{-2}\sum_{i\ne j\ne k\ne m} F_{ji}F_{ki} F_{mi}\xi_i,\quad S_{172} = N^{-2}\sum_{i\ne j\ne m} F_{ji}\xi_i F_{mi}F_{ji},
	\]
	and the remainder terms are bounded by $N^{-2}\sum_{i\ne j\ne m} E_{ji} E_{mi}\lesssim N\rhon^2 = o(N^2\rhon^2)$. 
	
	For $S_{172}$, we have
	\begin{align*}
		S_{172} \lesssim N^{-2}\max_i|\xi_i|\sum_{i\ne j\ne m} E_{mi}E_{ji} =o\Big(\Lambdan^2+N^2\rhon^2\Big).
	\end{align*}
	
	For $S_{171}$, we have
	\begin{align*}
		&S_{171} = N^{-2}\sum_{i\ne j\ne k} F_{ji}F_{ki} \xi_i^2-2N^{-2}\sum_{i\ne j\ne k} F_{ji}F_{ki} F_{ji}\xi_i\\
		&= N^{-2}\sum_{i\ne j\ne k} F_{ji}F_{ki} \xi_i^2 - 2S_{172}.
	\end{align*}
	
	To deal with the term $N^{-2}\sum_{i\ne j\ne k} F_{ji}F_{ki} \xi_i^2$, we have
	\begin{align*}
		&N^{-2}\sum_{i\ne j \ne k} F_{ji}F_{ki} \xi_i^2 = N^{-2}\sum_{i\ne j} F_{ji}\xi_i^3-N^{-2}\sum_{i\ne j} F_{ji}F_{ji} \xi_i^2 = N^{-2}\sum_i  \xi_i^4-N^{-2}\sum_{i\ne j} F_{ji}F_{ji} \xi_i^2.
	\end{align*}
	Note that
	\begin{align*}
		N^{-2}\sum_{i\ne j} F_{ji}F_{ji} \xi_i^2  \lesssim & N^{-2}\sum_{i\ne j} E_{ji} \xi_i^2 \lesssim (\max_i M_i) N^{-2}\sum_i    \xi_i^2  \\
		= & o(N^{1/2}\rhon+N^{1/2} \Lambdan \rhon)=   o\Big(\Lambdan^2+N^2\rhon^2\Big),\\
		N^{-2}\sum_i  \xi_i^4 = & o\Big(\Lambdan^2+N^2\rhon^2\Big).
	\end{align*}
	It follows that $S_{171} = o(\Lambdan^2+N^2\rhon^2)$. Therefore, $$S_{17} =  o\Big(\Lambdan^2+N^2\rhon^2\Big).$$
	
	For $S_{18}$, we have
	\begin{align*}
		S_{18} =& \sum_{i\ne j\ne k\ne m} N^{-2} F_{ji}F_{ki}F_{mj}\xi_{j}- \sum_{i\ne j\ne k\ne m} N^{-2} F_{ji}F_{ki}F_{mj}F_{mj} - \\
		& 
		\sum_{i\ne j\ne k\ne m} N^{-2} F_{ji}F_{ki}F_{mj}F_{kj}-\sum_{i\ne j\ne k\ne m} N^{-2} F_{ji}F_{ki}F_{mj}F_{ij},
	\end{align*}
	where the last three terms are bounded, up to some constant, by 
	\begin{align*}
		\sum_{i\ne j\ne k\ne m} N^{-2} E_{ji}E_{ki} E_{mj}= N^{-2} \bs{1}^\top\bs{E}\bs{E}\bs{E}^\top \bs{1} \lesssim N^2\rhon^3 = o\Big(\Lambdan^2+N^2\rhon^2\Big).
	\end{align*}
	For the first term, we have
	\begin{align*}
		\sum_{i\ne j\ne k\ne m} N^{-2} F_{ji}F_{ki}F_{mj}\xi_{j} =& \sum_{i\ne j\ne k} N^{-2} F_{ji}F_{ki}\xi_{j}^2-\sum_{i\ne j\ne k} N^{-2} F_{ji}F_{ki}F_{kj}\xi_{j} - \\
		&\sum_{i\ne j\ne k} N^{-2} F_{ji}F_{ki}F_{ij}\xi_{j},
	\end{align*}
	where the last two terms are bounded, up to some constant, by 
	\[
	\max_i |\xi_i| \sum_{i\ne j\ne k} N^{-2} E_{ji}E_{ki} =o\Big(\Lambdan^2+N^2\rhon^2\Big).
	\]
	Furthermore, let $(|\xi_i|)_{i} = (|\xi_1|,\ldots, |\xi_N|)$, and we have 
	\begin{align*}
		\sum_{i\ne j\ne k} N^{-2} F_{ji}F_{ki}\xi_{j}^2 = \sum_{i\ne j} N^{-2} F_{ji}\xi_i\xi_{j}^2-\sum_{i\ne j} N^{-2} F_{ji}F_{ji}\xi_{j}^2,
	\end{align*}
	where  
	\begin{align*}
		\sum_{i\ne j} N^{-2} F_{ji}F_{ji}\xi_{j}^2 &\lesssim \sum_{i\ne j} N^{-2} E_{ji}\xi_{j}^2 \leq \max_i \xi_{i}^2 \sum_{i\ne j} N^{-2} E_{ji} =\rho_N \max_i \xi_{i}^2 \\
		& =  \rho_N o(N + N\Lambdan) =o\Big(\Lambdan^2+N^2\rhon^2\Big),\\
		\sum_{i\ne j} N^{-2} F_{ji}\xi_i\xi_{j}^2 & \lesssim \max_j |\xi_j| \sum_{i\ne j} N^{-2} E_{ji}|\xi_{j}||\xi_i| \lesssim  N^{-2}(\max_i |\xi_i|) (|\xi_i|)_i^\top \bs{E} (|\xi_i|)_i \\
		& \lesssim \max_i |\xi_i|  \rhon N^{-1} \sum_i   \xi_{i}^2 \lesssim o(N^{1/2}\rhon+N^{1/2}\rhon\Lambdan^{3/2}) \\
		&= o\Big(\Lambdan^2+N^2\rhon^2\Big).
	\end{align*}
	As a consequence, we have
	\[
	S_{18} = o\Big(\Lambdan^2+N^2\rhon^2\Big).
	\]
	
	Using the same technique, we have
	\begin{align*}
		S_{19} = N^{-2}\sum_{i\ne k\ne l} \xi_i F_{ki} F_{kl}\xi_l + \text{remainder terms},
	\end{align*}
	where the remainder terms are $o(\Lambdan^2+N^2\rhon^2)$. 
	
	As
	\begin{align*}
		N^{-2}\sum_{i\ne k\ne l} \xi_i F_{ki} F_{kl}\xi_l& \lesssim N^{-2}\sum_{i\ne k\ne l} |\xi_i| E_{ki} E_{kl}|\xi_l| \lesssim N^{-2} (|\xi_i|)_i^\top \bs{E}^\top \bs{E} (|\xi_i|)_i \\
		& \lesssim  \rhon^2 \sum_i   \xi_i^2 = o\Big(\Lambdan^2+N^2\rhon^2\Big),
	\end{align*}
	we have $$S_{19} = o\Big(\Lambdan^2+N^2\rhon^2\Big).$$
	It follows that $$\Var(S_1) = o\Big(\Lambdan^2+N^2\rhon^2\Big).$$

	\textbf{Step 2}: For $S_2$, we write
	\begin{align*}
		S_2 =& 2N^{-1}\sum_{i\ne j \ne j^\prime} r_0 F_{ji}E_{j^\prime i}  \tilde{\gamma}^\prime_{j^\prime j}Z_jZ_{j^\prime} Z_i + N^{-1}\sum_{i\ne j \ne j^\prime}r_0^2 E_{ji}E_{j^\prime i}\tilde{\gamma}^\prime_{jj^\prime} \tilde{\gamma}^\prime_{j^\prime j} Z_jZ_{j^\prime} Z_i\\
		=:&2S_{21} + S_{22}.
	\end{align*}
	Define $ g(\bs{Z}) = g(Z_1,\ldots,Z_n) =: N^{-1} \sum_{i,j,j^\prime} F_{ji}E_{j^\prime i}  \tilde{\gamma}^\prime_{j^\prime j}Z_jZ_{j^\prime} Z_i $. Let $\bs{z} = (z_1,\ldots,z_k=1,\ldots,z_n)$, $\bs{z}^\prime = (z_1,\ldots,z_k=0,\ldots,z_n)$ and $\bs{z}_{(-k)} =(z_1,\ldots,z_{k-1},z_{k+1}\ldots,z_n) $. Then, we have
	\begin{align*}
		&\max_{\bs{z}_{(-k)}\in\{0,1\}^{N-1}} |g(\bs{z})-g(\bs{z}^\prime)| \\
		\lesssim &  N^{-1} \sum_{j, j^\prime} |F_{jk}E_{j^\prime k}\tilde{\gamma}^\prime_{j^\prime j}| + N^{-1} \sum_{i, j^\prime} |F_{ki}E_{j^\prime i}\tilde{\gamma}^\prime_{j^\prime k}| + N^{-1} \sum_{i, j} |F_{ji}E_{ki}  \tilde{\gamma}^\prime_{k j}| \\
		\lesssim & N^{-1}\sum_{j, j^\prime} E_{j^\prime k}Q_{j^\prime j} + N^{-1} \sum_{i, j^\prime} E_{j^\prime i}Q_{j^\prime k} + N^{-1} \sum_{i, j} E_{ji}  Q_{kj}.
	\end{align*}
	
	By Assumption~\ref{a:opnorm-EE^T}, we have
	\begin{align*}
		&N^{-2} \sum_k \Big(\sum_{j, j^\prime} E_{j^\prime k}Q_{j^\prime j} \Big)^2 =  N^{-2}\bs{1}^\top\bs{Q}^\top\bs{E}\bs{E}^\top\bs{Q}\bs{1} = O(N\rhon^2),\\
		&N^{-2} \sum_k \Big(\sum_{i, j^\prime} E_{j^\prime i}Q_{j^\prime k} \Big)^2 =  N^{-2}\bs{1}^\top\bs{E}^\top\bs{Q}\bs{Q}^\top\bs{E}\bs{1} = O(N\rhon^2),\\
		&N^{-2} \sum_k \Big(\sum_{i, j^\prime} E_{ji}  Q_{kj} \Big)^2 =  N^{-2}\bs{1}^\top\bs{E}^\top\bs{Q}^\top\bs{Q}\bs{E}\bs{1} = O(N\rhon^2).
	\end{align*}
	Then, by Lemma~\ref{lem:bound-for-variance-efron-stein}, we have
	\[
	\Var\big(S_{21} \big) = O(N\rhon^2) = o((N\rhon)^2).
	\]
	Using $|E_{ji}E_{j^\prime i}\tilde{\gamma}^\prime_{jj^\prime} \tilde{\gamma}^\prime_{j^\prime j}| \lesssim E_{j^\prime i}E_{ji}Q_{j^\prime j}$, we can prove in a similar way that
	\begin{align*}
		\Var\big(S_{22}\big) = o((N\rhon)^2).
	\end{align*}
	It follows that
	\[
	\Var(S_2) = o((N\rhon)^2) = o(\Lambdan^2 + N^2\rhon^2).
	\]

	\textbf{Step 3}: For $S_{3}$, we have
	\begin{align*}
		\Var(S_3) = N^{-2}\sum_{i_1\ne j_1\ne k_1 \ne l_1}\sum_{i_2\ne j_2\ne k_2 \ne l_2} & F_{j_1i_1}E_{k_1i_1}\tilde{\gamma}^\prime_{k_1l_1}F_{j_2i_2}E_{k_2i_2}\tilde{\gamma}^\prime_{k_2l_2}\\
		&\Cov(Z_{i_1} Z_{j_1}Z_{k_1} \tilde{Z}_{l_1},Z_{i_2} Z_{j_2}Z_{k_2} \tilde{Z}_{l_2}).
	\end{align*}
	
	We consider only the nonzero terms in the above summation, i.e., the terms correspond to $\Cov(Z_{i_1} Z_{j_1}Z_{k_1} \tilde{Z}_{l_1},Z_{i_2} Z_{j_2}Z_{k_2} \tilde{Z}_{l_2}) \ne 0.$ For these terms, we must have $l_1\in\{i_2,j_2,k_2,l_2\}$ and $l_2\in\{i_1,j_1,k_1,l_1\}$. We group these terms into $4$ categories according to $|\{i_1, j_1, k_1, l_1\}\cup \{i_2, j_2, k_2, l_2\}| = q$, $4 \leq q\leq 7$. We handle these $4$ categories separately.

	\textbf{Case 1}: $|\{i_1, j_1, k_1, l_1\}\cup \{i_2, j_2, k_2, l_2\}| = 7$.
	There is only one term under this category:
	\[
	N^{-2}\sum_{i_1\ne j_1\ne k_1 \ne l \ne i_2\ne j_2\ne k_2 }F_{j_1i_1}E_{k_1i_1}\tilde{\gamma}^\prime_{k_1l}F_{j_2i_2}E_{k_2i_2}\tilde{\gamma}^\prime_{k_2l} = S_{31} + \text{remainder terms},
	\]
	where
	\[
	S_{31} = N^{-2}\sum_{i_1\ne k_1 \ne l \ne i_2\ne j_2\ne k_2 }\xi_{i_1}E_{k_1i_1}\tilde{\gamma}^\prime_{k_1l}F_{j_2i_2}E_{k_2i_2}\tilde{\gamma}^\prime_{k_2l},
	\]
	and the remainder terms are bounded up to a constant by
	\begin{align*}
		& N^{-2}\sum_{i_1\ne k_1 \ne l \ne i_2\ne j_2\ne k_2 }E_{k_1i_1}Q_{k_1l}E_{j_2i_2}E_{k_2i_2}Q_{k_2l} \\
		\leq & N^{-2}\bs{1}^\top \bs{E}^\top\bs{Q}\bs{Q}^\top\bs{E}\bs{E}^\top \bs{1} = O(N^2\rhon^3) = o(N^2\rhon^2).
	\end{align*}
	Moreover, we have
	\begin{align*}
		S_{31} =  S_{311} +\text{remainder terms}
	\end{align*}
	where 
	\[
	S_{311} = N^{-2}\sum_{i_1\ne k_1 \ne l \ne i_2\ne k_2 }\xi_{i_1}E_{k_1i_1}\tilde{\gamma}^\prime_{k_1l}\xi_{i_2}E_{k_2i_2}\tilde{\gamma}^\prime_{k_2l},
	\]
	and the remainder terms are bounded up to a constant by
	\begin{align*}
		&\max_i |\xi_{i}| N^{-2}\sum_{i_1\ne k_1 \ne l \ne i_2\ne k_2 }E_{k_1i_1}Q_{k_1l}E_{k_2i_2}Q_{k_2l}\\
		\leq &   \max_i |\xi_{i}| N^{-2} \bs{1}^\top \bs{E}^\top \bs{Q}\bs{Q}^\top \bs{E} \bs{1} 
		= \max_i |\xi_{i}| N\rhon^2 = o\Big(\Lambdan^2+N^2\rhon^2\Big).
	\end{align*}
	
	For $S_{311}$, we have
	\begin{align*}
		S_{311} \lesssim & N^{-2}\sum_{i_1\ne k_1 \ne l \ne i_2\ne k_2 }|\xi_{i_1}|E_{k_1i_1}Q_{k_1l}|\xi_{i_2}|E_{k_2i_2}Q_{k_2l} \\
		\lesssim & N^{-2} (|\xi_i|)_i^\top \bs{E}^\top \bs{Q}\bs{Q}^\top \bs{E} (|\xi_i|)_i \lesssim N^{-2}\sum_i   \xi_i^2 (N\rhon)^2 = o\Big(\Lambdan^2+N^2\rhon^2\Big).
	\end{align*}
	As a consequence, we have
	\[
	N^{-2}\sum_{i_1\ne j_1\ne k_1 \ne l \ne i_2\ne j_2\ne k_2 }F_{j_1i_1}E_{k_1i_1}\tilde{\gamma}^\prime_{k_1l}F_{j_2i_2}E_{k_2i_2}\tilde{\gamma}^\prime_{k_2l} = o\Big(\Lambdan^2+N^2\rhon^2\Big).
	\]

	\textbf{Case 2:} $|\{i_1, j_1, k_1, l_1\}\cup \{i_2, j_2, k_2, l_2\}| = 6$. 
	
	\textbf{Case 2.1:} $\{k_1,l_1\}=\{k_2,l_2\}$.  This includes two cases: $(k_1,l_1)=(k_2,l_2)$ and $(k_1,l_1)=(l_2,k_2)$. 
	
	For $(k_1,l_1)=(k_2,l_2)$, the summation of the terms under this case is equal to
	\[
	S_{321} = N^{-2}\sum_{i_1\ne j_1\ne k_1 \ne l_1 \ne i_2\ne j_2}F_{j_1i_1}E_{k_1i_1}\tilde{\gamma}^\prime_{k_1l_1}F_{j_2i_2}E_{k_1i_2}\tilde{\gamma}^\prime_{k_1l_1}  = S_{3211}  + \text{remainder terms},
	\]
		where 
		\[
		S_{3211} = N^{-2}\sum_{i_1\ne k_1 \ne l_1 \ne i_2\ne j_2} \xi_{i_1} E_{k_1i_1}\tilde{\gamma}^\prime_{k_1l_1}F_{j_2i_2}E_{k_1i_2}\tilde{\gamma}^\prime_{k_1l_1},
		\]
		and the remainder terms are considered in the case of $|\{i_1, j_1, k_1, l_1\}\cup \{i_2, j_2, k_2, l_2\}| = 5$ in the following proof and have a magnitude of $o(\Lambdan^2 + N^2\rhon^2)$.
		
		For $S_{3211}$, we have
		\begin{align*}
			& S_{3211}-N^{-2}\sum_{i_1\ne k_1 \ne l_1 \ne i_2}\xi_{i_1} E_{k_1i_1}\tilde{\gamma}^\prime_{k_1l_1}\xi_{i_2}E_{k_1i_2}\tilde{\gamma}^\prime_{k_1l_1}  \\
			\lesssim &N^{-2}\sum_{i_1\ne k_1 \ne l_1 \ne i_2}|\xi_{i_1}| E_{k_1i_1}E_{k_1i_2}Q_{k_1l_1} + N^{-2}\sum_{i_1\ne k_1 \ne l_1 \ne i_2\ne j_2}E_{k_1i_1}Q_{k_1l_1}E_{j_2i_2}E_{k_1i_2}.
		\end{align*}
		As a consequence, we have
		\begin{align*}
			S_{3211} \lesssim & N^{-2}\sum_{i_1\ne k_1 \ne l_1 \ne i_2}|\xi_{i_1}| E_{k_1i_1}Q_{k_1l_1}|\xi_{i_2}|E_{k_1i_2}+ N^{-2}\sum_{i_1\ne k_1 \ne i_2}|\xi_{i_1}| E_{k_1i_1}E_{k_1i_2} + \\
			& N^{-2}\sum_{i_1\ne k_1 \ne l_1 \ne i_2\ne j_2}E_{k_1i_1}E_{j_2i_2}E_{k_1i_2}\\
			\lesssim & N^{-2}\sum_{i_1\ne k_1\ne i_2}|\xi_{i_1}| E_{k_1i_1}|\xi_{i_2}|E_{k_1i_2}+ o\Big(\Lambdan^2+N^2\rhon^2\Big)\\
			= &N^{-2} (|\xi_i|)_i^\top \bs{E}^\top \bs{E} (|\xi_i|)_i + o\Big(\Lambdan^2+N^2\rhon^2\Big) = o\Big(\Lambdan^2+N^2\rhon^2\Big).
		\end{align*}
		
		If $(k_1,l_1)=(l_2,k_2)$, we can verify in a similar way that
		\begin{align*}
			&N^{-2} \sum_{i_1\ne j_1\ne k_1 \ne l_1 \ne i_2\ne j_2} F_{j_1i_1}E_{k_1i_1}\tilde{\gamma}^\prime_{k_1l_1}F_{j_2i_2}E_{l_1i_2}\tilde{\gamma}^\prime_{l_1k_1}\\
			\lesssim & N^{-2} \sum_{i_1\ne k_1 \ne l_1 \ne i_2\ne j_2} |\xi_{i_1}|E_{k_1i_1}Q_{k_1l_1}|\xi_{i_2}|E_{l_1i_2}Q_{l_1k_1} + o\Big(\Lambdan^2+N^2\rhon^2\Big) \\
			= & o\Big(\Lambdan^2+N^2\rhon^2\Big),
		\end{align*}
		where the last equality is due to
		\begin{align*}
			& N^{-2} \sum_{i_1\ne k_1 \ne l_1 \ne i_2\ne j_2} |\xi_{i_1}|E_{k_1i_1}Q_{k_1l_1}|\xi_{i_2}|E_{l_1i_2}Q_{l_1k_1}
			\\ \lesssim &   N^{-2}(|\xi_i|)_i^\top \bs{E}^\top\bs{Q}\bs{E} (|\xi_i|)_i^\top  \lesssim  \rhon^2 \sum_i   \xi_i^2 =o\Big(\Lambdan^2+N^2\rhon^2\Big).
		\end{align*}
		
		\textbf{Case 2.2}: $\{l_1,j_1\}= \{l_2,j_2\}$, which includes two cases: $(l_1,j_1)=(l_2,j_2)$ and $(l_1,j_1)=(j_2,l_2)$. 
		
		For $(l_1,j_1)=(l_2,j_2)$, the summation under this case is equal to
		\begin{align*}
			&N^{-2}\sum_{i_1 \ne i_2 \ne j \ne l \ne k_1  \ne k_2 }F_{ji_1}E_{k_1 i_1}\tilde{\gamma}^\prime_{k_1l}F_{ji_2}E_{k_2 i_2}\tilde{\gamma}^\prime_{k_2l} \\
			\lesssim & N^{-2}\sum_{i_1 \ne i_2 \ne j \ne l \ne k_1  \ne k_2 }E_{ji_1}E_{k_1 i_1}Q_{k_1l}E_{ji_2}E_{k_2 i_2}Q_{k_2l}\\
			\leq & N^{-2}\sum_{i_1 \ne i_2 \ne j \ne l \ne k_1  \ne k_2 }E_{k_1 i_1}Q_{k_1l}E_{ji_2}E_{k_2 i_2}Q_{k_2l} \\
			\leq &N^{-2}\bs{1}^\top\bs{E}\bs{E}^\top\bs{Q}\bs{Q}^\top\bs{E} \bs{1}  =O(N^2\rhon^3) = o(N^2\rhon^2).
		\end{align*}
		
		The proof for the case of $(l_1,j_1)=(j_2,l_2)$ is similar, so we omit it.
		
		\textbf{Case 2.3} : $\{l_1,i_1\}=\{l_2,i_2\}$. 
		
		For $(l_1,i_1) = (l_2,i_2)$, the conclusion follows from 
		\begin{align*}
			& N^{-2}\sum_{i \ne j_1 \ne k_1 \ne l \ne j_2 \ne k_2 } F_{j_1i}E_{k_1i}\tilde{\gamma}^\prime_{k_1l}F_{j_2i}E_{k_2i}\tilde{\gamma}^\prime_{k_2l} \\
			=& N^{-2}\sum_{i\ne k_1 \ne l\ne k_2 } \xi_{i}E_{k_1i}\tilde{\gamma}^\prime_{k_1l}\xi_i E_{k_2i}\tilde{\gamma}^\prime_{k_2l} +  o\Big(\Lambdan^2+N^2\rhon^2\Big),
		\end{align*}
		and
		\begin{align*}
			&N^{-2}\sum_{i  \ne k_1 \ne l  \ne k_2 } \xi_{i}E_{k_1i}\tilde{\gamma}^\prime_{k_1l}\xi_i E_{k_2i}\tilde{\gamma}^\prime_{k_2l} \\
			\lesssim & \max_{i} \xi_i^2 N^{-2}\sum_{i  \ne k_1 \ne l \ne k_2 } Q_{k_1l} E_{k_2i}Q_{k_2l}\\
			\lesssim& \max_{i} \xi_i^2 N^{-2} \bs{1}\bs{Q}\bs{Q}^\top\bs{E}\bs{1} = o((N+N\Lambdan) N^{-2} N^2\rhon) \\
			= & o((N+N\Lambdan)\rhon) = o\Big(\Lambdan^2+N^2\rhon^2\Big).
		\end{align*}
		
		Similarly, we can verify the case of $(l_1,i_1) = (i_2,l_2)$.
		
		\textbf{Case 2.4} $\{l_1,k_1\}=\{l_2,i_2\}$.  
		
		For $(l_1,k_1)=(l_2,i_2)$,  the conclusion follows from
		\begin{align*}
			& N^{-2}\sum_{i_1 \ne j_1 \ne j_2 \ne l \ne k  \ne k_2 } F_{j_1i_1}E_{ki_1}\tilde{\gamma}^\prime_{kl}F_{j_2k}E_{k_2k}\tilde{\gamma}^\prime_{k_2l}\\ 
			\lesssim & N^{-2}\sum_{i_1\ne l\ne k\ne k_2 } |\xi_{i_1}|E_{ki_1}Q_{kl}|\xi_k| E_{k_2k}Q_{k_2l} +o\Big(\Lambdan^2+N^2\rhon^2\Big),
		\end{align*}
		and
		\begin{align*}
			&N^{-2}\sum_{i_1\ne l\ne k\ne k_2 } |\xi_{i_1}|E_{ki_1}Q_{kl}|\xi_k| E_{k_2k}Q_{k_2l}\\
			\lesssim & \max_i|\xi_{i}|^2 N^{-2}\sum_{i_1\ne l\ne k\ne k_2 } E_{ki_1}Q_{kl}Q_{k_2l} \\ 
			\lesssim & N^{-2}\max_i|\xi_{i}|^2 \bs{1}^\top\bs{E}^\top\bs{Q}\bs{Q}^\top \bs{1}
			=  \max_i|\xi_{i}|^2 (N\rhon) N^{-1} \\
			= & o\Big(\Lambdan^2+N^2\rhon^2\Big).
		\end{align*}

		For $(l_1,k_1)=(i_2,l_2)$, recall that $\delta_i = \sum_j Q_{ji}$ and by \Cref{a:Lindberg-condition-unadj},
		\[
		\max_i|\delta_i| = o(N^{1/2}).
		\]
		Then,
		\begin{align*}
			&N^{-2}\sum_{i_1\ne j_1 \ne k_1\ne l_1 \ne j_2 \ne k_2} F_{j_1i_1}E_{k_1i_1}\tilde{\gamma}^\prime_{k_1l_1}F_{j_2l_1}E_{k_2l_1}\tilde{\gamma}^\prime_{k_2k_1}\\
			\lesssim & N^{-2}\sum_{i_1 \ne k_1\ne l_1 \ne k_2} |\xi_{i_1}|E_{k_1i_1}Q_{k_1l_1}|\xi_{l_1} | E_{k_2l_1}Q_{k_2k_1} + \\
			& N^{-2}\sum_{i_1 \ne k_1\ne l_1  \ne k_2} |\xi_{i_1}|E_{k_1i_1}Q_{k_1l_1}E_{k_2l_1}Q_{k_2k_1} + o\Big(\Lambdan^2+N^2\rhon^2\Big) \\
			= & o\Big(\Lambdan^2+N^2\rhon^2\Big),
		\end{align*}
		where the last equality is due to 
		\begin{align*}
			&N^{-2}\sum_{i_1 \ne k_1\ne l_1 \ne k_2} |\xi_{i_1}|E_{k_1i_1}Q_{k_1l_1}|\xi_{l_1} | E_{k_2l_1}Q_{k_2k_1} \\
			\lesssim & N^{-2}\sum_{i_1 \ne k_1\ne l_1 \ne k_2} |\xi_{i_1}|E_{k_1i_1}Q_{k_1l_1}|\xi_{l_1} | Q_{k_2k_1} \\
			\lesssim & N^{-2} \max_i \delta_i \sum_{i_1 \ne k_1\ne l_1} |\xi_{i_1}|E_{k_1i_1}Q_{k_1l_1}|\xi_{l_1} | \lesssim  o(N^{-3/2}) (|\xi_i|)_i^\top \bs{E}\bs{Q} (|\xi_i|)_i \\
			= & o\Big(\Lambdan^2+N^2\rhon^2\Big),
		\end{align*}   
		and
		\begin{align*}
			&N^{-2}\sum_{i_1 \ne k_1\ne l_1  \ne k_2} |\xi_{i_1}|E_{k_1i_1}Q_{k_1l_1}E_{k_2l_1}Q_{k_2k_1} \\
			\lesssim & N^{-2}\max_i |\xi_{i}|\sum_{i_1 \ne k_1\ne l_1  \ne k_2} E_{k_1i_1}Q_{k_1l_1}E_{k_2l_1}Q_{k_2k_1} \\
			\lesssim & N^{-2}(\max_i N_i) \max_i |\xi_{i}|\sum_{ k_1\ne l_1  \ne k_2} Q_{k_1l_1}E_{k_2l_1}Q_{k_2k_1} \\
			\leq & N^{-2}(\max_i N_i) \max_i |\xi_{i}|\sum_{ k_1\ne l_1  \ne k_2} Q_{k_1l_1}Q_{k_2k_1}\\
			= & N^{-2}o(N^{3/2}\rhon)o(N^{1/2}\Lambdan^{1/2}+N^{1/2})O(N) 
			=o\Big(\Lambdan^2+N^2\rhon^2\Big),
		\end{align*}
		where the last equality is due to Young's inequality and  $N\rhon \geq C_{+}$.
		
		\textbf{Case 2.5}: $\{l_1,k_1\}=\{l_2,j_2\}$. 
		
		For $(l_1,k_1)=(j_2,l_2)$ We have
		\begin{align*}
			&N^{-2}\sum_{i_1\ne j_1 \ne k_1\ne l_1 \ne i_2 \ne k_2} F_{j_1i_1}E_{k_1i_1}\tilde{\gamma}^\prime_{k_1l_1}F_{l_1i_2}E_{k_2i_2}\tilde{\gamma}^\prime_{k_2k_1}  \\
			\lesssim  & N^{-2}\sum_{i_1\ne j_1 \ne k_1\ne l_1 \ne i_2 \ne k_2} |\xi_{i_1}|E_{k_1i_1}Q_{k_1l_1}E_{l_1i_2}E_{k_2i_2}Q_{k_2k_1} +o\Big(\Lambdan^2+N^2\rhon^2\Big).
		\end{align*}
		
		The conclusion follows from
		\begin{align*}
			&  N^{-2}\sum_{i_1\ne j_1 \ne k_1\ne l_1 \ne i_2 \ne k_2} |\xi_{i_1}|E_{k_1i_1}Q_{k_1l_1}E_{l_1i_2}E_{k_2i_2}Q_{k_2k_1} \\
			\lesssim & N^{-2}\sum_{i_1\ne j_1 \ne k_1\ne l_1 \ne i_2 \ne k_2} |\xi_{i_1}|E_{k_1i_1}Q_{k_1l_1}E_{l_1i_2}Q_{k_2k_1}\\
			\lesssim &   \max_i|\delta_i| N^{-2}\sum_{i_1\ne j_1 \ne k_1\ne l_1 \ne i_2 \ne k_2} |\xi_{i_1}|E_{k_1i_1}Q_{k_1l_1}E_{l_1i_2} 
			=   \max_i|\delta_i| N^{-2}(|\xi_i|)_i \bs{E}^\top\bs{Q}\bs{E}\bs{1} \\
			\lesssim &  \max_i|\delta_i| N^{-2}  N^{1/2}\|\bs{E}^\top\bs{Q}\bs{E}\|_{\oprtnorm}\| (|\xi_i|)_i \|_2 \\
			= & o(N\rhon^2\cdot(N^{1/2}+N^{1/2}\Lambdan^{1/2}))=o(N^2\rhon^2 + \Lambdan^2).
		\end{align*}
		
		Similarly, we can prove the case for $(l_1,k_1)=(l_2,j_2)$.
		
		\textbf{Case 2.6}: $\{l_1,i_1\}=\{l_2,j_2\}$. 
		
		For $(l_1,i_1)=(l_2,j_2)$, we have
		\begin{align*}
			&N^{-2}\sum_{i_1\ne j_1 \ne k_1\ne l_1 \ne i_2 \ne k_2} F_{j_1i_1}E_{k_1i_1}\tilde{\gamma}^\prime_{k_1l_1}F_{i_1i_2}E_{k_2i_2}\tilde{\gamma}^\prime_{k_2l_1} \\
			=&N^{-2}\sum_{i_1 \ne k_1\ne l_1 \ne i_2 \ne k_2} \xi_{i_1} E_{k_1i_1}\tilde{\gamma}^\prime_{k_1l_1}F_{i_1i_2}E_{k_2i_2}\tilde{\gamma}^\prime_{k_2l_1} + o(N^2\rhon^2 + \Lambdan^2).
		\end{align*}
		
		The conclusion follows from
		\begin{align*}
			&N^{-2}\sum_{i_1 \ne k_1\ne l_1 \ne i_2 \ne k_2} \xi_{i_1} E_{k_1i_1}\tilde{\gamma}^\prime_{k_1l_1}F_{i_1i_2}E_{k_2i_2}\tilde{\gamma}^\prime_{k_2l_1} \\
			\lesssim & N^{-2}\sum_{i_1 \ne k_1\ne l_1 \ne i_2 \ne k_2} |\xi_{i_1}| E_{k_1i_1}Q_{k_1l_1}E_{k_2i_2}Q_{k_2l_1} \\
			\leq & N^{-2}(|\xi_i|)_i^\top \bs{E}^\top \bs{Q}^\top\bs{Q} \bs{E}\bs{1} =O((N+N\Lambdan^{1/2})\rhon^2) = o(N^2\rhon^2 + \Lambdan^2).
		\end{align*}
		
		Similarly, we can prove the case for $(l_1,i_1)=(j_2,l_2)$.
		
		\textbf{Case 3:} $|\{i_1, j_1, k_1, l_1\}\cup \{i_2, j_2, k_2, l_2\}| = 5$. Let $k=\{i_2, j_2, k_2, l_2\}\backslash\{i_1, j_1, k_1, l_1\}$. We further divide Case $3$ into two cases. 
		
		\textbf{Case 3.1:} $\{k_1,l_1\}\ne \{k_2,l_2\}$. If $k = \{k_2,l_2\}\backslash \{k_1,l_1\}$, then we must have $k=k_2$, and therefore,
		\[ F_{j_1i_1}E_{k_1i_1}\tilde{\gamma}^\prime_{k_1l_1}F_{j_2i_2}E_{k_2i_2}\tilde{\gamma}^\prime_{k_2l_2} \lesssim E_{j_1i_1}E_{k_1i_1}Q_{k_1l_1}Q_{k_2l_2}.
		\]
		As $l_2\in \{i_1,j_1,k_1,l_1\}$, we have
		\begin{align*}
			&N^{-2}\sum_{i_1\ne j_1 \ne k_1 \ne l_1 \ne k_2} E_{j_1i_1}E_{k_1i_1}Q_{k_1l_1}Q_{k_2l_2} \\
			\lesssim & N^{-2}\max_i |\delta_i|\sum_{i_1\ne j_1 \ne k_1 \ne l_1} E_{j_1i_1}E_{k_1i_1}Q_{k_1l_1} \\
			\lesssim & N^{-2}\max_i |\delta_i| \bs{1}^\top\bs{E}\bs{E}^\top\bs{Q}\bs{1} = N^{-2}o(N^{1/2})O(N^3\rhon^2) = o(N^{3/2}\rhon^2).
		\end{align*}
		
		If $k\ne \{k_2,l_2\}\backslash \{k_1,l_1\}$, then $\{k_2,l_2\}\in \{i_1,j_1,k_1,l_1\}$. Recall that $\bar{\bs{E}} = \bs{E}+\bs{E}^\top$ and $\bar{\bs{Q}}= \bs{Q} + \bs{Q}^\top$. Since $\{k_2,l_2\}\ne \{k_1,l_1\}$, we have
		\begin{align*}
			F_{j_1i_1}E_{k_1i_1}\tilde{\gamma}^\prime_{k_1l_1}F_{j_2i_2}E_{k_2i_2}\tilde{\gamma}^\prime_{k_2l_2} \lesssim E_{j_1i_1}E_{k_1i_1}Q_{k_1l_1}Q_{k_2l_2}\bar{E}_{qk},~\text{for}~ q\in\{i_1,j_1,k_1,l_1\}.
		\end{align*}
		Therefore,
		\begin{align*}
			&N^{-2}\sum_{i_1\ne j_1 \ne k_1 \ne l_1 \ne k}E_{j_1i_1}E_{k_1i_1}Q_{k_1l_1}Q_{k_2l_2}\bar{E}_{qk} \\
			\lesssim & N^{-2}\max_i\{N_i,M_i\} \sum_{i_1\ne j_1 \ne k_1 \ne l_1}E_{j_1i_1}E_{k_1i_1}Q_{k_1l_1}Q_{k_2l_2}   \\
			\lesssim & N^{-2}o(N^{3/2}\rhon) \sum_{i_1\ne j_1 \ne k_1 \ne l_1} \bar{E}_{j_1i_1}\bar{E}_{k_1i_1}\bar{Q}_{k_1l_1}\bar{Q}_{k_2l_2}.
		\end{align*}
		
		We now bound $\sum_{i_1\ne j_1 \ne k_1 \ne l_1} \bar{E}_{j_1i_1}\bar{E}_{k_1i_1}\bar{Q}_{k_1l_1}\bar{Q}_{k_2l_2}$.  We consider five cases: $\{k_2,l_2\} = \{i_1,j_1\}$, $\{k_2,l_2\} = \{i_1,k_1\}$, $\{k_2,l_2\} = \{i_1,l_1\}$, $\{k_2,l_2\} = \{j_1,k_1\}$, $\{k_2,l_2\} = \{j_1,l_1\}$. Their corresponding magnitudes are as follows:
		\begin{align*}
			&\sum_{i_1\ne j_1 \ne k_1 \ne l_1} \bar{E}_{j_1i_1}\bar{E}_{k_1i_1}\bar{Q}_{k_1l_1}\bar{Q}_{i_1j_1} \lesssim \sum_{i_1\ne j_1 \ne k_1 \ne l_1} \bar{E}_{k_1i_1}\bar{Q}_{k_1l_1}\bar{Q}_{i_1j_1} \lesssim \bs{1}^\top \bar{\bs{Q}}\bar{\bs{E}}\bar{\bs{Q}} \bs{1} \lesssim N^2\rhon ,\\
			&\sum_{i_1\ne j_1 \ne k_1 \ne l_1} \bar{E}_{j_1i_1}\bar{E}_{k_1i_1}\bar{Q}_{k_1l_1}\bar{Q}_{i_1k_1} \lesssim \sum_{i_1\ne j_1 \ne k_1 \ne l_1} \bar{E}_{j_1i_1}\bar{Q}_{k_1l_1}\bar{Q}_{i_1k_1} \lesssim \bs{1}^\top \bar{\bs{E}}\bar{\bs{Q}}\bar{\bs{Q}} \bs{1} \lesssim N^2\rhon ,\\
			& \sum_{i_1\ne j_1 \ne k_1 \ne l_1} \bar{E}_{j_1i_1}\bar{E}_{k_1i_1}\bar{Q}_{k_1l_1}\bar{Q}_{i_1l_1} \lesssim \sum_{i_1\ne j_1 \ne k_1 \ne l_1} \bar{E}_{j_1i_1}\bar{Q}_{k_1l_1}\bar{Q}_{i_1l_1} \lesssim \bs{1}^\top \bar{\bs{E}}\bar{\bs{Q}}\bar{\bs{Q}} \bs{1} \lesssim N^2\rhon ,\\
			&\sum_{i_1\ne j_1 \ne k_1 \ne l_1} \bar{E}_{j_1i_1}\bar{E}_{k_1i_1}\bar{Q}_{k_1l_1}\bar{Q}_{j_1k_1} \lesssim \sum_{i_1\ne j_1 \ne k_1 \ne l_1} \bar{E}_{j_1i_1}\bar{Q}_{k_1l_1}\bar{Q}_{j_1k_1} \lesssim \bs{1}^\top \bar{\bs{E}}\bar{\bs{Q}}\bar{\bs{Q}} \bs{1} \lesssim N^2\rhon ,\\
			& \sum_{i_1\ne j_1 \ne k_1 \ne l_1} \bar{E}_{j_1i_1}\bar{E}_{k_1i_1}\bar{Q}_{k_1l_1}\bar{Q}_{j_1l_1} \lesssim \sum_{i_1\ne j_1 \ne k_1 \ne l_1} \bar{E}_{j_1i_1}\bar{Q}_{k_1l_1}\bar{Q}_{j_1l_1} \lesssim \bs{1}^\top \bar{\bs{E}}\bar{\bs{Q}}\bar{\bs{Q}} \bs{1} \lesssim N^2\rhon.
		\end{align*}
		Putting together the pieces, we have for  $k\ne \{k_2,l_2\}\backslash \{k_1,l_1\}$ and $\{k_1,l_1\}\ne \{k_2,l_2\}$ those summands are of magnitude $o(N^{3/2}\rhon^2)$.
		
		\textbf{Case 3.2:} $\{k_1,l_1\}=\{k_2,l_2\}$. There are two possibilities for $k$: $k=i_2$ and $k=j_2$.
		
		\textbf{Case 3.2.1:} $k=i_2$. In this case, $j_2 = i_1$ or $j_2 = j_1$. For $j_2=i_1$, we consider $(k_1,l_1)=(k_2,l_2)$ and $(k_1,l_1)=(l_2,k_2)$ separately.
		
		When $(k_1,l_1)=(k_2,l_2)$, we have
		\begin{align*}
			&N^{-2}\sum_{i_1\ne j_1 \ne k_1 \ne l_1 \ne i_2} F_{j_1i_1}E_{k_1i_1}\tilde{\gamma}^\prime_{k_1l_1}F_{i_1i_2}E_{k_2i_2}\tilde{\gamma}^\prime_{k_2l_2} \\
			\lesssim & N^{-2}\sum_{i_1\ne j_1 \ne k_1 \ne l_1 \ne i_2} E_{j_1i_1}Q_{k_1l_1}E_{i_1i_2}E_{k_1i_2} \lesssim N^{-2} \bs{1}^\top\bs{E}\bs{E}\bs{E}^\top\bs{Q}\bs{1} \lesssim N^2\rhon^3.
		\end{align*}
		
		When $(k_1,l_1)=(l_2,k_2)$, we have
		\begin{align*}
			&N^{-2}\sum_{i_1\ne j_1 \ne k_1 \ne l_1 \ne i_2} F_{j_1i_1}E_{k_1i_1}\tilde{\gamma}^\prime_{k_1l_1}F_{i_1i_2}E_{k_2i_2}\tilde{\gamma}^\prime_{k_2l_2} \\
			\lesssim & N^{-2}\sum_{i_1\ne j_1 \ne k_1 \ne l_1 \ne i_2} E_{j_1i_1}Q_{k_1l_1}E_{i_1i_2}E_{l_1i_2} \lesssim N^{-2} \bs{1}^\top\bs{E}\bs{E}\bs{E}^\top\bs{Q}^\top\bs{1} \lesssim N^2\rhon^3.
		\end{align*}

		For $j_2=j_1$, we have
		\begin{align*}
			&N^{-2}\sum_{i_1\ne j_1 \ne k_1 \ne l_1 \ne i_2} F_{j_1i_1}E_{k_1i_1}\tilde{\gamma}^\prime_{k_1l_1}F_{j_2i_2}E_{k_2i_2}\tilde{\gamma}^\prime_{k_2l_2} \\
			\lesssim & N^{-2}\sum_{i_1\ne j_1 \ne k_1 \ne l_1 \ne i_2} E_{j_1i_1}E_{k_1i_1}Q_{k_1l_1}E_{j_1i_2} \lesssim N^{-2}\bs{1}^\top\bs{E}^\top\bs{E}\bs{E}^\top\bs{Q}\bs{1} \lesssim N^2 \rhon^3.
		\end{align*}
		
		\textbf{Case 3.2.2:} $k=j_2$. In this case, $i_2=i_1$ or $i_2 = j_1$. For $i_2=i_1$, we have
		\begin{align*}
			&N^{-2}\sum_{i_1\ne j_1 \ne k_1 \ne l_1 \ne j_2} F_{j_1i_1}E_{k_1i_1}\tilde{\gamma}^\prime_{k_1l_1}F_{j_2i_2}E_{k_2i_2}\tilde{\gamma}^\prime_{k_2l_2} \\
			=& N^{-2}\sum_{i_1\ne j_1 \ne k_1 \ne l_1} F_{j_1i_1}E_{k_1i_1}\tilde{\gamma}^\prime_{k_1l_1}\xi_{i_1}E_{k_2i_1}\tilde{\gamma}^\prime_{k_2l_2} + \text{remainder terms},
		\end{align*}
		where the remainder terms are bounded by
		\begin{align*}
			N^{-2}\sum_{i_1\ne j_1 \ne k_1 \ne l_1} E_{j_1i_1}E_{k_1i_1}Q_{k_1l_1}E_{k_2i_1}Q_{k_2l_2} \lesssim N^{-2}\sum_{i_1\ne j_1 \ne k_1 \ne l_1} E_{j_1i_1}E_{k_1i_1}Q_{k_1l_1} \lesssim N\rhon^2.
		\end{align*}
		
		The conclusion follows from
		\begin{align*}
			& N^{-2}\sum_{i_1\ne j_1 \ne k_1 \ne l_1} F_{j_1i_1}E_{k_1i_1}\tilde{\gamma}^\prime_{k_1l_1}\xi_{i_1}E_{k_2i_1}\tilde{\gamma}^\prime_{k_2l_2} \\
			\lesssim & (\max_i|\xi_i|)N^{-2}\sum_{i_1\ne j_1 \ne k_1 \ne l_1} E_{j_1i_1}E_{k_1i_1}Q_{k_1l_1} = o(N^{3/2}\rhon^2 (\Lambdan^{1/2}+1)) = o\Big(\Lambdan^2+N^2\rhon^2\Big).
		\end{align*}
		
		The case of $k=j_2, \ i_2 = j_1$ is equivalent to $k=i_2$, $j_2=i_1$ by symmetry.
		
		\textbf{Case 4:} $|\{i_1, j_1, k_1, l_1\}\cup \{i_2, j_2, k_2, l_2\}| = 4$. The summation under this category is equal to 
		\begin{align*}
			& N^{-2}\sum_{i_1\ne j_1\ne k_1 \ne l_1} F_{j_1i_1}E_{k_1i_1}\tilde{\gamma}^\prime_{k_1l_1}F_{j_2i_2}E_{k_2i_2}\tilde{\gamma}^\prime_{k_2l_2} \\
			\lesssim & N^{-2}\sum_{i_1\ne j_1\ne k_1 \ne l_1} E_{j_1i_1}E_{k_1i_1} E_{k_1l_1}
			\lesssim  N^{-2} \bs{1}^\top \bs{E}\bs{E}^\top \bs{E} \bs{1} = O(N^2 \rhon^3) = o\Big(\Lambdan^2+N^2\rhon^2\Big).
		\end{align*}
		
		Combing Case 1 to Case 4, we have $$\Var(S_3) = o\Big(\Lambdan^2+N^2\rhon^2\Big).$$
		
		\textbf{Step 4}. For $S_4$, by \Cref{lem:bound-for-variance-efron-stein}, we have
		\begin{align*}
			\Var(S_4) \lesssim & N^{-2}\sum_{l} \Big(\sum_{ j\ne j^\prime \ne k} E_{jl}E_{j^\prime l}Q_{jj^\prime}Q_{j^\prime k} \Big)^2 + N^{-2}\sum_{l}\Big(\sum_{i\ne j^\prime \ne k} E_{li}E_{j^\prime i}Q_{lj^\prime}Q_{j^\prime k}\Big)^2+\\
			&N^{-2}\sum_{l}\Big(\sum_{i\ne j \ne k} E_{ji}E_{li}Q_{jl}Q_{l k}\Big)^2+ N^{-2}\sum_{l}\Big(\sum_{i\ne j\ne j^\prime }E_{ji}E_{j^\prime i}Q_{jj^\prime}Q_{j^\prime l}\Big)^2.
		\end{align*}
		Note that 
		\begin{align*}
			& N^{-2}\sum_{l}\Big(\sum_{ j\ne j^\prime \ne k} E_{jl}E_{j^\prime l}Q_{jj^\prime}Q_{j^\prime k}\Big)^2 \lesssim N^{-2}\sum_{l}\Big(\sum_{ j\ne j^\prime \ne k} E_{jl}Q_{jj^\prime}Q_{j^\prime k}\Big)^2  \\
			\lesssim  & N^{-2}\bs{1}^\top \bs{Q}^\top\bs{Q}^\top\bs{E}\bs{E}^\top\bs{Q}\bs{Q}\bs{1} \lesssim N\rhon^2,\\
			& N^{-2}\sum_{l}\Big(\sum_{i\ne j^\prime \ne k} E_{li}E_{j^\prime i}Q_{lj^\prime}Q_{j^\prime k}\Big)^2 \lesssim N^{-2}\sum_{l}\Big(\sum_{i\ne j^\prime \ne k} E_{li}Q_{lj^\prime}Q_{j^\prime k}\Big)^2\\
			\lesssim &\max_i N_i^2 N^{-2}\sum_{l}\Big(\sum_{j^\prime \ne k} Q_{lj^\prime}Q_{j^\prime k}\Big)^2 \lesssim o(N^3\rhon^2)N^{-2}\bs{1}^\top \bs{Q}^\top\bs{Q}^\top\bs{Q}\bs{Q}\bs{1} = o(N^2\rhon^2),\\
			& N^{-2}\sum_{l}\Big(\sum_{i\ne j\ne k} E_{ji}E_{li}Q_{jl}Q_{l k}\Big)^2 \lesssim N^{-2}\sum_{l}\Big(\sum_{i\ne j\ne k}E_{ji}Q_{jl}Q_{l k}\Big)^2 \\
			\lesssim & N^{-2}\sum_{l}\Big(\sum_{i\ne j}E_{ji}Q_{jl}\Big)^2 \lesssim N^{-2}\bs{1}^\top \bs{E}^\top\bs{Q}\bs{Q}^\top\bs{E}\bs{1}  = o(N\rhon^2)\\
			&N^{-2}\sum_{l}\Big(\sum_{i\ne j\ne j^\prime }E_{ji}E_{j^\prime i}Q_{jj^\prime}Q_{j^\prime l}\Big)^2 \lesssim N^{-2}\sum_{l}\Big(\sum_{i\ne j\ne j^\prime }E_{ji}Q_{jj^\prime}Q_{j^\prime l}\Big)^2 \\
			\lesssim & N^{-2}\bs{1}^\top \bs{E}^\top\bs{Q}\bs{Q}\bs{Q}^\top\bs{Q}^\top\bs{E}\bs{1} \lesssim N\rhon^2.
		\end{align*}
		Therefore, $$\Var(S_4) = o\Big(\Lambdan^2+N^2\rhon^2\Big).$$
		
		\textbf{Step 5}. For $S_5$, by \Cref{lem:bound-for-variance-efron-stein}, we have 
		\begin{align*}
			\Var(S_5) = & N^{-2}\sum_{l}\Big(\sum_{i\ne j\ne k} E_{ji}E_{l i}Q_{jk}Q_{l k}\Big)^2 +
			N^{-2}\sum_{l}\Big(\sum_{ j\ne j^\prime \ne k} E_{jl}E_{j^\prime l}Q_{jk}Q_{j^\prime k}\Big)^2+\\
			&N^{-2}\sum_{l}\Big(\sum_{i\ne j\ne j^\prime } E_{ji}E_{j^\prime i}Q_{jl}Q_{j^\prime l}\Big)^2.
		\end{align*}
		
		Since $\max_{i}\delta_i^2 = o(N)$, then we have
		\begin{align*}
			&N^{-2}\sum_{l}\Big(\sum_{i\ne j\ne k} E_{ji}E_{l i}Q_{jk}Q_{l k}\Big)^2 \\
			\lesssim & N^{-2}\sum_{l}\Big(\sum_{i\ne j\ne k} E_{ji}Q_{jk}Q_{l k}\Big)^2 
			\lesssim  N^{-2}\bs{1}^\top\bs{E}^\top\bs{Q}\bs{Q}^\top\bs{Q}\bs{Q}^\top\bs{E}\bs{1}\lesssim N \rhon^2,\\
			& N^{-2}\sum_{l}\Big(\sum_{ j\ne j^\prime \ne k} E_{jl}E_{j^\prime l}Q_{jk}Q_{j^\prime k}\Big)^2 \\
			\lesssim & N^{-2}\sum_{l}\Big(\sum_{ j\ne j^\prime \ne k}E_{j^\prime l}Q_{jk}Q_{j^\prime k}\Big)^2 \lesssim N^{-2}\bs{1}^\top\bs{Q}\bs{Q}^\top\bs{E}\bs{E}^\top\bs{Q}\bs{Q}^\top\bs{1} \lesssim N \rhon^2,\\
			& N^{-2}\sum_{l}\Big(\sum_{i\ne j\ne j^\prime } E_{ji}E_{j^\prime i}Q_{jl}Q_{j^\prime l}\Big)^2 \\
			\lesssim & N^{-2}\sum_{l}\Big(\sum_{i\ne j\ne j^\prime } E_{ji}Q_{jl}Q_{j^\prime l}\Big)^2 
			\lesssim  \max_{i}\delta_i^2 N^{-2}\sum_{l}\Big(\sum_{i\ne j } E_{ji}Q_{jl}\Big)^2\\
			\lesssim & \max_{i}\delta_i^2 N^{-2}\bs{1}^\top\bs{E}^\top\bs{Q}\bs{Q}^\top\bs{E}\bs{1} \lesssim N \rhon^2 \max_{i}\delta_i^2 = o(N^2\rhon^2).
		\end{align*}
		
		As a consequence, we have $$\Var(S_5) = o\Big(\Lambdan^2+N^2\rhon^2\Big).$$
		
		\textbf{Step 6}. For $S_6$, we have
		\begin{align*}
			\Var(S_6) =& N^{-2}\sum_{i_1\ne j_1\ne k_1 \ne m_1 \ne l_1}\sum_{i_2\ne j_2\ne k_2 \ne m_2 \ne l_2} E_{j_1i_1}E_{k_1 i_1} \tilde{\gamma}^\prime_{j_1m_1}\tilde{\gamma}^\prime_{k_1 l_1} E_{j_2i_2}E_{k_2 i_2} \tilde{\gamma}^\prime_{j_2m_2}\tilde{\gamma}^\prime_{k_2 l_2} \\
			&\Cov(\tilde{Z}_{m_1}\tilde{Z}_{l_1}Z_{j_1}Z_{k_1} Z_{i_1}, \tilde{Z}_{m_2}\tilde{Z}_{l_2}Z_{j_2}Z_{k_2} Z_{i_2}) \\
			\lesssim& \sum_{i_1\ne j_1\ne k_1 \ne m_1 \ne l_1}\sum_{i_2\ne j_2\ne k_2 \ne m_2 \ne l_2} E_{j_1i_1}E_{k_1 i_1} Q_{j_1m_1}Q_{k_1 l_1} E_{j_2i_2}E_{k_2 i_2} Q_{j_2m_2}Q_{k_2 l_2} \\
			&I(\Cov(\tilde{Z}_{m_1}\tilde{Z}_{l_1}Z_{j_1}Z_{k_1} Z_{i_1}, \tilde{Z}_{m_2}\tilde{Z}_{l_2}Z_{j_2}Z_{k_2} Z_{i_2})\ne 0) .
		\end{align*}
		We only need to consider the terms with $\Cov(\tilde{Z}_{m_1}\tilde{Z}_{l_1}Z_{j_1}Z_{k_1} Z_{i_1}, \tilde{Z}_{m_2}\tilde{Z}_{l_2}Z_{j_2}Z_{k_2} Z_{i_2})\ne 0$, i.e., $m_1,l_1\in\{i_2, j_2, k_2 , m_2 , l_2\}$, $m_2,l_2\in\{i_1, j_1, k_1 , m_1, l_1\}$. We group these terms into $4$ categories according to $|\{i_1, j_1, k_1 , m_1, l_1\}\cup \{i_2, j_2, k_2 , m_2, l_2\}|=q$, $5 \leq q \leq 8$, and deal with these $4$ categories one by one.
		
		\textbf{Case 1:} $|\{i_1, j_1, k_1 , m_1, l_1\}\cup \{i_2, j_2, k_2 , m_2, l_2\}|=8$. In this case, we have $m_1=m_2=m$ and $l_1=l_2=l$, (the case that $m_1=l_2,\ m_2=l_1$ is equivalent, due to the exchangeable role of $m_1$ and $l_1$). Thus,
		\begin{align*}
			&N^{-2}\sum_{i_1\ne j_1\ne k_1 \ne i_2\ne j_2\ne k_2\ne l\ne m } E_{j_1i_1}E_{k_1 i_1} Q_{j_1m}Q_{k_1 l} E_{j_2i_2}E_{k_2 i_2} Q_{j_2m}Q_{k_2 l}  \\
			\lesssim & N^{-2}\sum_{i_1\ne j_1\ne k_1 \ne i_2\ne j_2\ne k_2\ne l\ne m } E_{k_1 i_1} Q_{j_1m}Q_{k_1 l} E_{j_2i_2}E_{k_2 i_2} Q_{j_2m}Q_{k_2 l}\\
			\lesssim &N^{-2}\bs{1}\bs{E}^\top\bs{Q}\bs{Q}^\top\bs{E}\bs{E}^\top\bs{Q}\bs{Q}^\top\bs{1} \lesssim N^{-1}\|\bs{Q}\|^4_{\oprtnorm}\|\bs{E}\|^3_{\oprtnorm} \lesssim N^2\rhon^3.
		\end{align*}
		
		\textbf{Case 2:} $|\{i_1, j_1, k_1 , m_1, l_1\}\cup \{i_2, j_2, k_2 , m_2, l_2\}|=7$. Let $\{l_1,m_1,q\} = \{i_1, j_1, k_1 , m_1, l_1\}\cap \{i_2, j_2, k_2 , m_2, l_2\}$. For this category, we consider two cases.
		
		\textbf{Case 2.1:} $\{l_1,m_1,j_1\} = \{l_2,m_2,j_2\}$ or $\{l_1,m_1,j_1\} = \{l_2,m_2,k_2\}$, or $\{l_1,m_1,k_1\} = \{l_2,m_2,j_2\}$, or $\{l_1,m_1,k_1\} = \{l_2,m_2,j_2\}$. By the exchangeable role of $l_q$ and $m_q$, and $j_q$ and $k_q$, $q=1,2$, without loss of generality, we consider the case $\{l_1,m_1,j_1\} = \{l_2,m_2,j_2\}$.

		\textbf{Case 2.1.1}: $\{j_1,m_1\} = \{j_2,m_2\}$.
		
		For $(l_1,m_1,j_1)=(l_2,m_2,j_2)$, we have
		\begin{align*}
			&N^{-2}\sum_{i_1\ne i_2\ne j\ne k_1\ne k_2 \ne m \ne l}E_{ji_1}E_{k_1 i_1} Q_{j m}Q_{k_1 l} E_{ji_2}E_{k_2 i_2} Q_{jm}Q_{k_2 l}\\
			\lesssim&  N^{-2}\sum_{i_1\ne i_2 \ne j\ne k_1\ne k_2  \ne l}E_{k_1 i_1}Q_{k_1 l} E_{ji_2}E_{k_2 i_2} Q_{k_2 l}
			\lesssim  N^{-1}\|\bs{Q}\|^2_{\oprtnorm}\|\bs{E}\|^3_{\oprtnorm} 
			\lesssim N^2\rhon^3.
		\end{align*}
		
		For $(l_1,m_1,j_1)=(l_2,j_2,m_2)$, we have
		\begin{align*}
			&N^{-2}\sum_{i_1\ne i_2\ne j\ne m \ne k_1 \ne k_2 \ne l}E_{ji_1}E_{k_1 i_1} Q_{jm}Q_{k_1 l} E_{mi_2}E_{k_2 i_2} Q_{mj}Q_{k_2 l}\\
			\lesssim & N^{-2}\sum_{i_1\ne i_2\ne j\ne m \ne k_1 \ne k_2 \ne l} E_{k_1 i_1} Q_{k_1 l} E_{mi_2}E_{k_2 i_2} Q_{mj}Q_{k_2 l} \lesssim N^{-1}\|\bs{Q}\|^3_{\oprtnorm}\|\bs{E}\|^3_{\oprtnorm} \lesssim N^2\rhon^3.
		\end{align*}
		
		\textbf{Case 2.1.2}: $\{j_2,m_2\} = \{j_1,l_1\}$ or $\{j_2,m_2\} = \{m_1,l_1\}$. 
		
		In this case, we have
		\begin{align*}
			&N^{-2}\sum_{i_1 \ne j_1 \ne k_1 \ne m_1 \ne l_1 \ne i_2 \ne k_2}  E_{j_1i_1}E_{k_1 i_1} Q_{j_1m_1}Q_{k_1 l_1} E_{j_2i_2}E_{k_2 i_2} Q_{j_2m_2}Q_{k_2 l_2} \\
			\lesssim &N^{-2}\sum_{i_1 \ne j_1 \ne k_1 \ne m_1 \ne l_1 \ne i_2 \ne k_2}  E_{j_1i_1}E_{k_1 i_1} Q_{j_1m_1}Q_{k_1 l_1} E_{j_2i_2}Q_{j_2m_2}Q_{k_2 l_2}  \\
			\lesssim & (\max_i \delta_i) (\max_i N_i) N^{-2}\sum_{i_1 \ne j_1 \ne k_1 \ne m_1 \ne l_1}  E_{j_1i_1}E_{k_1 i_1} Q_{j_1m_1}Q_{k_1 l_1} Q_{j_2m_2} = o(N^2\rhon)\rhon,
		\end{align*}
		where the last equality holds because $N^{-2}\sum_{i_1 \ne j_1 \ne k_1 \ne m_1 \ne l_1}  E_{j_1i_1}E_{k_1 i_1} Q_{j_1m_1}Q_{k_1 l_1} Q_{j_2m_2}$, under $\{j_2,m_2\} = \{j_1,l_1\}$ and $\{j_2,m_2\} = \{m_1,l_1\}$, are bounded, respectively, by
		\begin{align*}
			&N^{-2}\sum_{i_1 \ne j_1 \ne k_1 \ne m_1 \ne l_1}  E_{j_1i_1}E_{k_1 i_1} Q_{j_1m_1}Q_{k_1 l_1} \bar{Q}_{j_1l_1} \\ \lesssim & N^{-2}\sum_{i_1 \ne j_1 \ne k_1 \ne m_1 \ne l_1}  E_{k_1 i_1} Q_{j_1m_1}Q_{k_1 l_1} \bar{Q}_{j_1l_1} \lesssim N^{-1}\|\bs{Q}\|^3_{\oprtnorm}\|\bs{E}\|_{\oprtnorm} \lesssim \rhon,
		\end{align*}
		and
		\begin{align*}
			&N^{-2}\sum_{i_1 \ne j_1 \ne k_1 \ne m_1 \ne l_1}  E_{j_1i_1}E_{k_1 i_1} Q_{j_1m_1}Q_{k_1 l_1} \bar{Q}_{m_1l_1} \\ \lesssim &  N^{-2}\sum_{i_1 \ne j_1 \ne k_1 \ne m_1 \ne l_1}  E_{k_1 i_1} Q_{j_1m_1}Q_{k_1 l_1} 
			\bar{Q}_{m_1l_1} \lesssim N^{-1}\|\bs{Q}\|^3_{\oprtnorm}\|\bs{E}\|_{\oprtnorm}  \lesssim \rhon.
		\end{align*}

		\textbf{Case 2.2}: $\{l_1,m_1,j_1\} \ne \{l_2,m_2,j_2\}$, $\{l_1,m_1,j_1\} \ne \{l_2,m_2,k_2\}$, $\{l_1,m_1,k_1\} \ne \{l_2,m_2,j_2\}$, and $\{l_1,m_1,k_1\} \ne \{l_2,m_2,j_2\}$. 
		
		In this case, we have $\{j_1,k_1\} = \{i_1,j_1,k_1,m_1,l_1\}\backslash \{i_2,j_2,k_2,m_2,l_2\}$ or $\{j_2,k_2\} = \{i_2,j_2,k_2,m_2,l_2\}\backslash $ $\{i_1,j_1,k_1,m_1,l_1\}$. By symmetry, without loss of generality, suppose that $\{j_2,k_2\} = \{i_2,j_2,k_2,m_2,l_2\}\backslash$ $\{i_1,j_1,k_1,m_1,l_1\}$. Then,

		\begin{align*}
			&N^{-2} \sum_{j_1\ne m_1\ne l_1 \ne i_1 \ne k_1 \ne j_2 \ne k_2} E_{j_1i_1}E_{k_1 i_1} Q_{j_1m_1}Q_{k_1 l_1} E_{j_2i_2}E_{k_2 i_2} Q_{j_2 m_2}Q_{k_2 l_2}  \\
			\lesssim &N^{-2} \sum_{j_1\ne m_1\ne l_1 \ne i_1 \ne k_1 \ne j_2 \ne k_2} E_{j_1i_1}E_{k_1 i_1} Q_{j_1m_1}Q_{k_1 l_1} Q_{j_2 m_2}Q_{k_2 l_2} \\
			\lesssim  & N^{-2} \max_i \delta_i^2\sum_{j_1\ne m_1\ne l_1 \ne i_1 \ne k_1 \ne j_2 \ne k_2} E_{j_1i_1}E_{k_1 i_1} Q_{j_1m_1}Q_{k_1 l_1} E_{j_2i_2}E_{k_2 i_2}   \\
			\lesssim & \max_i \delta_i^2 N^{-1}\|\bs{Q}\|^2_{\oprtnorm}\|\bs{E}\|^2_{\oprtnorm}  = o(N^2\rhon^2).
		\end{align*}

		\textbf{Case 3}: $|\{i_1, j_1, k_1 , m_1, l_1\}\cup \{i_2, j_2, k_2 , m_2, l_2\}|=6$. Let $q = \{i_2, j_2, k_2 , m_2, l_2\}\backslash$ $ \{i_1, j_1, k_1 , m_1, l_1\} $. 
		
		\textbf{Case 3.1}:  $\{j_2,m_2\}\ne \{j_1,m_1\}$ and   $\{j_2,m_2\}\ne \{k_1,l_1\}$; or $\{k_2,l_2\}\ne \{j_1,m_1\}$ and   $\{k_2,l_2\}\ne \{k_1,l_1\}$. Without loss of generality, let's consider $\{j_2,m_2\}\ne \{j_1,m_1\}$ and   $\{j_2,m_2\}\ne \{k_1,l_1\}$. 
		
		\textbf{Case 3.1.1}: $q\in \{j_2,m_2\}$.
		
		In this case, we have
		\begin{align*}
			&N^{-2}\sum_{i_1\ne j_1 \ne k_1\ne l_1 \ne m_1 \ne q}E_{j_1i_1}E_{k_1 i_1} \tilde{\gamma}^\prime_{j_1m_1}\tilde{\gamma}^\prime_{k_1 l_1} E_{j_2i_2}E_{k_2 i_2} \tilde{\gamma}^\prime_{j_2m_2}\tilde{\gamma}^\prime_{k_2 l_2}  \\
			\lesssim &\max_i |\delta_i| N^{-2}\sum_{i_1\ne j_1 \ne k_1\ne l_1 \ne m_1} \bar{E}_{j_1i_1}\bar{E}_{k_1 i_1} \bar{Q}_{j_1m_1}\bar{Q}_{k_1 l_1} \\
			= &
			\max_i |\delta_i| 
			N^{-1}\|\bs{Q}\|^2_{\oprtnorm}\|\bs{E}\|^2_{\oprtnorm} = o(N^{3/2}\rhon^2).
		\end{align*}
		
		\textbf{Case 3.1.2}: $q\not \in \{j_2,m_2\}$.
		
		In this case, we have
		\begin{align*}
			&N^{-2}\sum_{i_1\ne j_1 \ne k_1\ne l_1 \ne m_1 \ne q}E_{j_1i_1}E_{k_1 i_1} \tilde{\gamma}^\prime_{j_1m_1}\tilde{\gamma}^\prime_{k_1 l_1} E_{j_2i_2}E_{k_2 i_2} \tilde{\gamma}^\prime_{j_2m_2}\tilde{\gamma}^\prime_{k_2 l_2}  \\
			\lesssim &\max_{i}\{M_i,N_i\}N^{-2}\sum_{i_1\ne j_1 \ne k_1\ne l_1 \ne m_1 }\bar{E}_{j_1i_1}\bar{E}_{k_1 i_1} \bar{Q}_{j_1m_1}\bar{Q}_{k_1 l_1}\bar{Q}_{j_2m_2}.
		\end{align*}
		
		We consider five cases: $\{j_2,m_2\} = \{m_1,i_1\}$, $\{j_2,m_2\} = \{m_1,k_1\}$, $\{j_2,m_2\} = \{m_1,l_1\}$, $\{j_2,m_2\} = \{j_1,i_1\}$, $\{j_2,m_2\} = \{j_1,k_1\}$. We omit other cases due to symmetry. Under those cases, we can bound $N^{-2}\sum_{i_1\ne j_1 \ne k_1\ne l_1 \ne m_1 }\bar{E}_{j_1i_1}\bar{E}_{k_1 i_1} \bar{Q}_{j_1m_1}\bar{Q}_{k_1 l_1}\bar{Q}_{j_2m_2}$ as follows:
		\begin{align*}
			&N^{-2}\sum_{i_1\ne j_1 \ne k_1\ne l_1 \ne m_1 }\bar{E}_{j_1i_1}\bar{E}_{k_1 i_1} \bar{Q}_{j_1m_1}\bar{Q}_{k_1 l_1}\bar{Q}_{m_1 i_1}
			\\
			\lesssim & N^{-2}\sum_{i_1\ne j_1 \ne k_1\ne l_1 \ne m_1 }\bar{E}_{k_1 i_1} \bar{Q}_{j_1m_1}\bar{Q}_{k_1 l_1}\bar{Q}_{m_1 i_1}
			\lesssim  N^{-1}\|\bs{Q}\|^3_{\oprtnorm}\|\bs{E}\|_{\oprtnorm}\lesssim \rhon, \\
			&N^{-2}\sum_{i_1\ne j_1 \ne k_1\ne l_1 \ne m_1 }\bar{E}_{j_1i_1}\bar{E}_{k_1 i_1} \bar{Q}_{j_1m_1}\bar{Q}_{k_1 l_1}\bar{Q}_{m_1 k_1} \\
			\lesssim & N^{-2}\sum_{i_1\ne j_1 \ne k_1\ne l_1 \ne m_1 }\bar{E}_{j_1i_1} \bar{Q}_{j_1m_1}\bar{Q}_{k_1 l_1}\bar{Q}_{m_1 k_1}
			\lesssim  N^{-1}\|\bs{Q}\|^3_{\oprtnorm}\|\bs{E}\|_{\oprtnorm}\lesssim \rhon,\\
			&N^{-2}\sum_{i_1\ne j_1 \ne k_1\ne l_1 \ne m_1 }\bar{E}_{j_1i_1}\bar{E}_{k_1 i_1} \bar{Q}_{j_1m_1}\bar{Q}_{k_1 l_1}\bar{Q}_{m_1 l_1} \\
			\lesssim & N^{-2}\sum_{i_1\ne j_1 \ne k_1\ne l_1 \ne m_1 }\bar{E}_{k_1 i_1} \bar{Q}_{j_1m_1}\bar{Q}_{k_1 l_1}\bar{Q}_{m_1 l_1} 
			\lesssim  N^{-1}\|\bs{Q}\|^3_{\oprtnorm}\|\bs{E}\|_{\oprtnorm}\lesssim \rhon,\\
			& N^{-2}\sum_{i_1\ne j_1 \ne k_1\ne l_1 \ne m_1 }\bar{E}_{j_1i_1}\bar{E}_{k_1 i_1} \bar{Q}_{j_1m_1}\bar{Q}_{k_1 l_1}\bar{Q}_{j_1 i_1} \\
			\lesssim & N^{-2}\sum_{i_1\ne j_1 \ne k_1\ne l_1 \ne m_1 }\bar{E}_{k_1 i_1} \bar{Q}_{j_1m_1}\bar{Q}_{k_1 l_1}\bar{Q}_{j_1 i_1}
			\lesssim  N^{-1}\|\bs{Q}\|^3_{\oprtnorm}\|\bs{E}\|_{\oprtnorm}\lesssim \rhon,\\
			& N^{-2}\sum_{i_1\ne j_1 \ne k_1\ne l_1 \ne m_1 }\bar{E}_{j_1i_1}\bar{E}_{k_1 i_1} \bar{Q}_{j_1m_1}\bar{Q}_{k_1 l_1}\bar{Q}_{j_1 k_1} \\
			\lesssim & N^{-2}\sum_{i_1\ne j_1 \ne k_1\ne l_1 \ne m_1 }\bar{E}_{k_1 i_1} \bar{Q}_{j_1m_1}\bar{Q}_{k_1 l_1}\bar{Q}_{j_1 k_1}
			\lesssim   \max_i \{N_i,M_i\} N^{-1}\|\bs{Q}\|^3_{\oprtnorm}= o(\rhon N^{1/2}).
		\end{align*}
		Therefore, the order of the summation under this case is $o(N^{1/2}\rhon)o(N^{3/2}\rhon) = o(N^2\rhon^2)$.
		
		\textbf{Case 3.2}:  $\{j_2,m_2\}= \{j_1,m_1\}$ or   $\{j_2,m_2\} = \{k_1,l_1\}$; and $\{k_2,l_2\} = \{j_1,m_1\}$ or   $\{k_2,l_2\} = \{k_1,l_1\}$. Without loss of generality, let's consider $\{j_2,m_2\} = \{j_1,m_1\}$ and   $\{k_2,l_2\} = \{k_1,l_1\}$. 
		
		In this case, we have
		
		\begin{align*}
			&N^{-2}\sum_{i_1\ne j_1 \ne k_1\ne l_1 \ne m_1 \ne q}E_{j_1i_1}E_{k_1 i_1} \tilde{\gamma}^\prime_{j_1m_1}\tilde{\gamma}^\prime_{k_1 l_1} E_{j_2i_2}E_{k_2 i_2} \tilde{\gamma}^\prime_{j_2m_2}\tilde{\gamma}^\prime_{k_2 l_2}  \\
			\lesssim &N^{-2}\sum_{i_1\ne j_1 \ne k_1\ne l_1 \ne m_1 \ne q}E_{j_1i_1}E_{k_1 i_1} Q_{j_1m_1}^2 Q_{k_1 l_1}^2 E_{j_1i_2}E_{k_1 i_2}\\
			\lesssim & \big(\max_{j_1}\sum_{m_1} Q_{j_1m_1}^2\big) \big(\max_{k_1}\sum_{l_1} Q_{k_1 l_1}^2\big) N^{-2}\sum_{i_1\ne j_1 \ne k_1\ne l_1 \ne m_1 \ne q}E_{j_1i_1}E_{k_1 i_1} E_{j_1i_2}E_{k_1 i_2}\\
			\leq & \|\bs{Q}\|_{\oprtnorm}^4 N^{-1} \|\bs{E}\|_{\oprtnorm}^3 = O(N^2\rhon^3) = o(N^2\rhon^2).
		\end{align*}
		
		\textbf{Case 4}: $|\{i_1, j_1, k_1 , m_1, l_1\}\cup \{i_2, j_2, k_2 , m_2, l_2\}|=5$.
		
		In this case, we have
		\[
		E_{j_1i_1}E_{k_1 i_1} Q_{j_1m_1}Q_{k_1 l_1} E_{j_2i_2}E_{k_2 i_2} Q_{j_2m_2}Q_{k_2 l_2} \lesssim E_{j_1i_1}E_{k_1 i_1} Q_{j_1m_1}Q_{k_1 l_1},
		\] 
		and therefore, the summation of these terms is bounded by
		\[
		N^{-2}\sum_{i_1\ne j_1 \ne k_1 \ne m_1 \ne l_1} E_{j_1i_1}E_{k_1 i_1} Q_{j_1m_1}Q_{k_1 l_1} \lesssim N^{-1}\|\bs{Q}\|^2_{\oprtnorm}\|\bs{E}\|^2_{\oprtnorm} = O(N\rhon^2).
		\]
		
		As a consequence, we have $$\Var(S_6) = o\Big(\Lambdan^2+N^2\rhon^2\Big)$$.
		
		It follows that
		\[
		\hat{V}_{\ind} = \E \hat{V}_{\ind} + \op\Big(N^{-1}\Lambdan +\rhon\Big).
		\]

	\end{proof}
	
	\subsection{Proof of Theorem \ref{thm:variance-estimator-tot}}

	Recall that $\hat{\tau}^\prime_{\tot} = \hat{\tau}_{\dir} + \hat{\tau}^\prime_{\ind}$ with 
	\[
	\hat{\tau}^\prime_{\ind}  =  \frac{1}{N}\sum_{i,j} E_{ij} \Bigl\{\frac{Y_i^\prime Z_j}{r_1} - \frac{Y_i^\prime (1-Z_j)}{r_0}\Bigr\},
	\]
	where $Y_i^\prime = \alpha_i^\prime + \theta_i^\prime Z_i +\sum_j  \tilde{E}_{ij} {\gamma}_{ij}^\prime Z_j$ with $\tilde{\gamma}^\prime_{ij} = \tilde{E}_{ij} {\gamma}_{ij}^\prime$. We will prove a general result that encompasses \Cref{thm:variance-estimator-tot}(ii) as a special case. Denote $\sigma_{\tot}^{\prime~2} = \Var(\hat{\tau}^\prime_{\tot})$ with
	\begin{align*}
		\hat{V}^\prime_{\tot} = &\sum_i \biggl(Y_i+\sum_j   E_{ji}Y_{j}^\prime\frac{Z_j}{r_1}\biggr)^2 \frac{Z_i}{N^2r_1}+\sum_i \biggl(Y_i+\sum_j   E_{ji}Y_{j}^\prime \frac{1-Z_j}{r_0}\biggr)^2 \frac{r_0Z_i}{N^2r_1^2}  + \\
		& \sum_i \biggl(Y_i +\sum_j   E_{ji}Y_{j}^\prime 
		\frac{Z_j}{r_1}\biggr)^2 \frac{r_1(1-Z_i)}{N^2r_0^2} + \sum_i \biggl(Y_i +\sum_j   E_{ji}Y_{j}^\prime\frac{1-Z_j}{r_0}\biggr)^2 \frac{1-Z_i}{N^2r_0}.
	\end{align*}

	\begin{theorem}
		\label{thm:variance-tot-general}
		(i) $\E (2\hat{V}^\prime_{\tot}) \geq \Var(\hat{\tau}_{\tot}^\prime)$ and (ii) suppose that \Cref{a:bounded-parameter} holds for both $\{\alpha_i,\theta_i,{\gamma}_{ij}\}_{1 \leq i,j\leq N}$ and $\{\alpha_i^\prime,\theta_i^\prime,{\gamma}_{ij}^\prime\}_{1 \leq i,j\leq N}$, Assumptions~\ref{a:density-rho_N}--\ref{a:opnorm-EE^T} hold, then, we have
		\[\E \hat{V}^\prime_{\tot} = \Var(\hat{\tau}^\prime_{\tot}) + \frac{1}{N^2} \sum_i \Big(\theta_i +\sum_j  E_{ji}\theta_{j}^\prime\Big)^2 + O(\rhon).
		\]
	\end{theorem}

	\begin{proof}[Proof of \Cref{thm:variance-tot-general}]
		Recall the definition of $\hat{V}_{\ind}^\prime$ as $\hat{V}_{\ind}$ with $Y_i$ replaced by $Y_i^\prime$.
		We make the following decomposition:
		\begin{align}
			\label{eq:V-tot-V-ind-V-dir}
			\E\big(\hat{V}_{\tot}^\prime-\hat{V}_{\ind}^\prime -\hat{V}_{\dir} \big) = S_4 + S_5 + S_6 + S_7,
		\end{align}
		where
		\begin{align*}
			&S_4 = \frac{2}{N^2} \E \sum_i   Y_i Z_i\sum_j   E_{ji}Y_j^\prime \frac{Z_j}{r_1^2},\quad S_5  = \frac{2}{N^2} \E \sum_i   Y_i (1-Z_i)\sum_j   E_{ji}Y_j^\prime \frac{Z_j}{r_1^2},\\
			&S_6 = \frac{2}{N^2} \E \sum_i   Y_i Z_i\sum_j   E_{ji}Y_j^\prime \frac{1-Z_j}{r_0^2},\quad S_7 =  \frac{2}{N^2} \E \sum_i   Y_i (1-Z_i)\sum_j   E_{ji}Y_j^\prime \frac{1-Z_j}{r_0^2}.
		\end{align*}
		
		Since
		\begin{align*}
			&\E Y_i Y_{j}^\prime \frac{I(Z_i=z,Z_j=z^\prime)}{\E I(Z_i=z,Z_j=z^\prime)} \\
			=& \E(Y_i \mid Z_i=z,Z_j=z^\prime) Y_{Z_i=z}^{\prime Z_j=z^\prime} + \Cov(Y_i,Y_j^\prime\mid Z_i=z,Z_j=z^\prime) \\
			=&\E(Y_i\mid Z_i=z,Z_j=z^\prime) Y_{Z_j=z^\prime}^{\prime Z_i=z} + r_1 r_0\sum_{k}\tilde{\gamma}_{ik}\tilde{\gamma}_{jk}^\prime \\
			=& \{Y_{Z_i=z} +(z^\prime-r_1)\tilde{\gamma}_{ij}\} Y_{ Z_j=z^\prime}^{\prime Z_i=z} + r_1 r_0\sum_{k}\tilde{\gamma}_{ik}\tilde{\gamma}_{jk}^\prime,
		\end{align*}
		then, we have
		\begin{align*}
			\sum_{q=4}^7 S_q =& 
			\frac{2}{N^2}\sum_{i,j}  E_{ji} Y_{Z_i=1} Y_{Z_j=1}^{\prime Z_i=1} + \frac{2r_0}{r_1N^2}\sum_{i,j}  E_{ji}Y_{Z_i=1} Y_{ Z_j=0}^{\prime Z_i=1}+\\
			&\frac{2r_1}{r_0 N^2}\sum_{i,j}  E_{ji} Y_{Z_i=0} Y_{Z_j=1}^{\prime Z_i=0} + \frac{2}{N^2}\sum_{i,j}  E_{ji}Y_{Z_i=0} Y_{Z_j=0}^{\prime Z_i=0} + \frac{2}{N^2} \sum_{i,j}  \sum_{k}E_{ji}\tilde{\gamma}_{ik}\tilde{\gamma}_{jk}^\prime +\\
			&\frac{2r_0}{N^2}\sum_{i,j}  E_{ji} (Y_{Z_j=1}^{\prime Z_i=1}-Y_{ Z_j=0}^{\prime Z_i=1})\tilde{\gamma}_{ij} +\frac{2r_1}{N^2}\sum_{i,j}  E_{ji} (Y_{Z_j=1}^{\prime Z_i=0}-Y_{Z_j=0}^{\prime Z_i=0})\tilde{\gamma}_{ij}.
		\end{align*}
		
		Since
		\begin{align*}
			&\sum_{i,j}  \sum_{k}E_{ji}\tilde{\gamma}_{ik}\tilde{\gamma}_{jk}^\prime+ r_0\sum_{i,j}  E_{ji} (Y_{Z_j=1}^{\prime Z_i=1}-Y_{ Z_j=0}^{\prime Z_i=1})\tilde{\gamma}_{ij} +r_1\sum_{i,j}  E_{ji} (Y_{Z_j=1}^{\prime Z_i=0}-Y_{Z_j=0}^{\prime Z_i=0})\tilde{\gamma}_{ij}\\
			= &\sum_{i,j}  E_{ji}\big(\sum_{k}\tilde{\gamma}_{ik}\tilde{\gamma}_{jk}^\prime+\tilde{\gamma}_{ij} \theta_j^\prime \big) = \sum_{i,j}  E_{ji}\theta_j^\prime \tilde{\gamma}_{ij} +\sum_{i,j} \sum_{k} E_{ji}\tilde{\gamma}_{ik}\tilde{\gamma}_{jk}^\prime \\
			= &\sum_{i,j}  E_{ji}\theta_j^\prime \tilde{\gamma}_{ij} +\sum_{i,j,k}  E_{ki}\tilde{\gamma}_{ij} \tilde{\gamma}_{kj}^\prime = \sum_{i,j}  \tilde{\gamma}_{ij} ( E_{ji}\theta_j^\prime + E_{ki}\sum_{k}\tilde{\gamma}_{kj}^\prime),
		\end{align*}
		then, we have
		\begin{align}
			\nonumber
			\sum_{q=4}^7 S_q = &\frac{2}{N^2}\sum_{i,j}  E_{ji} Y_{Z_i=1} Y_{Z_j=1}^{\prime Z_i=1} + \frac{2r_0}{r_1N^2}\sum_{i,j}  E_{ji}Y_{Z_i=1} Y_{ Z_j=0}^{\prime Z_i=1} + \frac{2r_1}{r_0 N^2}\sum_{i,j}  E_{ji} Y_{Z_i=0} Y_{Z_j=1}^{\prime Z_i=0}+\\
			& \frac{2}{N^2}\sum_{i,j}  E_{ji}Y_{Z_i=0} Y_{Z_j=0}^{\prime Z_i=0}+\frac{2}{N^2} \sum_{i,j}  E_{ji}\sum_{i,j}  \tilde{\gamma}_{ij} ( E_{ji}\theta_j^\prime + E_{ki}\sum_{k}\tilde{\gamma}_{kj}^\prime)
			\label{eq:expression-of-S4-to-S7}.
		\end{align}
		Denote 
		\begin{align*}
			S_{\tot} = S_{\tot,1} + S_{\tot,2},
		\end{align*}
		where
		\begin{align*}
			S_{\tot,1}  :=& \frac{1}{N^2}\sum_i  \biggl(Y_{Z_i=1} +\sum_j   E_{ji}Y_{Z_j=1}^{\prime Z_i=1} \biggr)^2+\frac{r_0}{N^2r_1}\sum_i  \biggl(Y_{Z_i=1} +\sum_j   E_{ji}Y_{ Z_j=0}^{\prime Z_i=1}\biggr)^2  + \\
			& \frac{r_1}{N^2r_0}\sum_i  \biggl(Y_{Z_i=0} +\sum_j   E_{ji}Y_{Z_j=1}^{\prime Z_i=0}\biggr)^2  + \frac{1}{N^2}\sum_i  \biggl(Y_{Z_i=0} +\sum_j   E_{ji}Y_{Z_j=0}^{\prime Z_i=0}\biggr)^2,\\
			S_{\tot,2}  := & \frac{1}{N^2}  \sum_{i\ne j} (\tilde{\gamma}_{ji} +E_{ij}\theta_{i}^\prime+\sum_k  E_{kj}\tilde{\gamma}_{ki}^\prime)^2.
		\end{align*}
		Similar to \eqref{eq:E-V_ind}, we have
		\begin{align}
			\label{eq:E-V_ind-prime}
			\E\hat{V}_{\ind}^\prime &= \sum_{q=1}^5\tilde{S}_q^\prime +\mathcal{B}_{\ind}^\prime,
		\end{align}
		where $\tilde{S}_q^\prime$ and $\mathcal{B}_{\ind}^\prime$ are defined similarly to $\tilde{S}_q$ and $\mathcal{B}_{\ind}$ with $Y_i$ replaced by $Y_i^\prime$. Note that
		\begin{align*}
			\mathcal{B}_{\ind}^\prime =& \frac{1}{r_1r_0N^2}\sum_{i,j}  E_{ji}(r_0r_1Y_{ Z_j=0}^{\prime Z_i=1}+r_0^2Y_{Z_j=1}^{\prime Z_i=1} + r_1^2Y_{Z_j=0}^{\prime Z_i=0}+r_0r_1Y_{Z_j=1}^{\prime Z_i=0})^2+\\
			&+\frac{1}{N^2}\sum_{i,j}  E_{ji} \sum_{k} \tilde{\gamma}_{jk}^{\prime 2}+\frac{1}{N^2}\sum_{i}  \sum_{j^\prime\ne j} E_{ji}E_{j^\prime i} \tilde{\gamma}^\prime_{jj^\prime } \tilde{\gamma}^\prime_{j^\prime j} \geq 0.
		\end{align*}
		
		By \eqref{eq:V-tot-V-ind-V-dir} and summating \eqref{eq:V-tot-V-ind-V-dir}, \eqref{eq:E-V_dir} and \eqref{eq:E-V_ind-prime}, we have
		\begin{align}
			\E \hat{V}_{\tot}^\prime &= S_{\tot}  + \mathcal{B}_{\ind}^\prime \geq  S_{\tot} . \label{eq:hat-V-tot-expectation}
		\end{align}

		Since  $\Var(\hat{\tau}^\prime_{\tot}) = \sigma_{\tot,1}^{\prime~2}+ \sigma_{\tot,2}^{\prime~2}$ with
		\begin{align*}
			\sigma_{\tot,2}^{\prime~2} = &  \frac{1}{N^2}  \sum_{ i\ne j }  \Big(\tilde{\gamma}_{ij} +E_{ji}\theta_{j}^\prime+\sum_k  E_{ki}\tilde{\gamma}_{kj}^\prime\Big)\Big(\tilde{\gamma}_{ji}+E_{ij}\theta_{i}^\prime+\sum_k  E_{kj}\tilde{\gamma}_{ki}^\prime\Big) \\
			&+ \frac{1}{N^2}  \sum_{ i\ne j } \Big(\tilde{\gamma}_{ji}+E_{ij}\theta_{i}^\prime +\sum_k  E_{kj}\tilde{\gamma}_{ki}^\prime\Big)^2,
		\end{align*}
		then mimicking the way we derive \eqref{eq:two-cauchy-schwarz-inequality}, we have
		\[
		S_{\tot,1} = \sigma_{\tot,1}^{\prime~2} + \frac{1}{N^2} \sum_i \left(\theta_i  +\sum_j  E_{ji}\theta_{j}^\prime\right)^2+\frac{1}{N^2} \sum_i  \Bigl(\sum_j   E_{ji}\gamma_{ji}^\prime\Bigr)^2.
		\]
		It follows that
		\begin{align*}
			S_{\tot}  =& \Var(\hat{\tau}_{\tot}^\prime) + \frac{1}{N^2} \sum_i \left(\theta_i +\sum_j  E_{ji}\theta_{j}^\prime\right)^2+\frac{1}{N^2} \sum_i  \Bigl(\sum_j   E_{ji}\gamma_{ji}^\prime\Bigr)^2-\\
			&\sum_{i\ne j}\frac{1}{N^2} (\tilde{\gamma}_{ij} +\theta_{j}^\prime E_{ji}+\sum_k  E_{ki}\tilde{\gamma}_{kj}^\prime)(\tilde{\gamma}_{ji}+\theta_{i}^\prime E_{ij}+\sum_k  E_{kj}\tilde{\gamma}_{ki}^\prime).
		\end{align*}
		Therefore, 
		\begin{align*}
			S_{\tot}  \geq & \Var(\hat{\tau}_{\tot}^\prime) -\sum_{i\ne j}\frac{1}{N^2} (\tilde{\gamma}_{ij}^\prime +\theta_{j}E_{ji}+\sum_k  E_{ki}\tilde{\gamma}_{kj})(\tilde{\gamma}_{ji}^\prime+\theta_{i}E_{ij}+\sum_k  E_{kj}\tilde{\gamma}_{ki})\\
			\geq& \Var(\hat{\tau}_{\tot}^\prime) - S_{\tot,2} \geq \Var(\hat{\tau}_{\tot}^\prime) - S_{\tot}.
		\end{align*}
		Combining with \eqref{eq:hat-V-tot-expectation}, we have $$\E (2\hat{V}^\prime_{\tot}) \geq \Var(\hat{\tau}_{\tot}^\prime). $$
		
		Moreover, under Assumptions~\ref{a:bounded-parameter}--\ref{a:opnorm-EE^T} for both $\{\alpha_i,\theta_i,{\gamma}_{ij}\}_{1 \leq i,j\leq N}$ and $\{\alpha_i^\prime,\theta_i^\prime,{\gamma}_{ij}^\prime\}_{1 \leq i,j\leq N}$, we have 
		\begin{align*}
			S_{\tot}  =& \Var(\hat{\tau}_{\tot}^\prime) + \frac{1}{N^2} \sum_i \left(\theta_i +\sum_j  E_{ji}\theta_{j}^\prime\right)^2 + O(\rhon).
		\end{align*}

		Similar to the way we derive $\mathcal{B}_{\ind} = O(\rhon)$ in the proof of Theorem \ref{thm:variance-estimator-ind}(i), we have $\mathcal{B}_{\ind}^\prime = O(\rhon)$ and then,
		\begin{align*}
			\E \hat{V}_{\tot}^\prime =& S_{\tot}  + O(\rhon)= \Var(\hat{\tau}^\prime_{\tot}) + \frac{1}{N^2} \sum_i \left(\theta_i +\sum_j  E_{ji}\theta_{j}^\prime\right)^2 + O(\rhon).  
		\end{align*}

	\end{proof}
	
	Next, we prove \Cref{thm:consistency-of-oracle-variance-estimator-general-tot} below. \Cref{thm:variance-estimator-tot}(ii) is a special case of \Cref{thm:consistency-of-oracle-variance-estimator-general-tot}.
	
	\begin{theorem}
		\label{thm:consistency-of-oracle-variance-estimator-general-tot}
		Suppose that both $\{\alpha_i,\theta_i,{\gamma}_{ij}\}_{1 \leq i,j\leq N}$ and $\{\alpha_i^\prime,\theta_i^\prime,{\gamma}_{ij}^\prime\}_{1 \leq i,j\leq N}$ satisfy \Cref{a:bounded-parameter}. Under Assumptions~\ref{a:density-rho_N}--\ref{a:Lindberg-condition-unadj}, and for a nonrandom sequence $\Lambdan$,
		\begin{align*}
			& N^{-1} \max\Big\{\sum_i \Big(\sum_j  E_{ji}Y_{Z_j=1}^\prime\Big)^2, \sum_i \Big(\sum_j  E_{ji}Y_{Z_j=0}^\prime\Big)^2\Big\} = O(\Lambdan ),\\
			&N^{-1} \max\Big\{\maxi\Big(\sum_j  E_{ji}Y_{Z_j=1}^\prime\Big)^2, \maxi\Big(\sum_j  E_{ji}Y_{Z_j=0}^\prime\Big)^2\Big\} = o(\Lambdan ),
		\end{align*}
		we have $\hat{V}_{\tot}^\prime-\operatorname{E} (\hat{V}_{\tot}^\prime)  = \op(N^{-1}\Lambdan + \rhon)$.
	\end{theorem}
	
	\begin{proof}[Proof of \Cref{thm:consistency-of-oracle-variance-estimator-general-tot}]
		Recall that
		\begin{align*}
			\hat{V}_{\tot}^\prime-\hat{V}_{\ind}^\prime-\hat{V}_{\dir} = S_4 + S_5 + S_6 + S_7.
		\end{align*}
		where
		\begin{align*}
			&S_4 = \frac{2}{N^2}\sum_i   Y_i Z_i\sum_j   E_{ji}Y_j^\prime \frac{Z_j}{r_1^2},\quad S_5  = \frac{2}{N^2}\sum_i   Y_i (1-Z_i)\sum_j   E_{ji}Y_j^\prime \frac{Z_j}{r_1^2},\\
			&S_6 = \frac{2}{N^2}\sum_i   Y_i Z_i\sum_j   E_{ji}Y_j^\prime \frac{1-Z_j}{r_0^2},\quad S_7 =  \frac{2}{N^2}\sum_i   Y_i(1-Z_i)\sum_j   E_{ji}Y_j^\prime  \frac{1-Z_j}{r_0^2}.
		\end{align*}
		
		Define $f(\bs{Z}) = f(Z_1,\ldots,Z_n) : = S_4$. Let $\bs{z} = (z_1,\ldots,z_k=1,\ldots,z_n)$, $\bs{z}^\prime = (z_1,\ldots,z_k=0,\ldots,z_n)$ and $\bs{z}_{(-k)} =(z_1,\ldots,z_{k-1},z_{k+1}\ldots,z_n) $, then, we have
		\begin{align*}
			&\max_{\bs{z}_{(-k)}\in\{0,1\}^{N-1}}|f(\bs{z})-f(\bs{z}^\prime)| \\
			\lesssim & \frac{2}{N^2}\sum_{i,j: i \ne k, j \ne k}  \Big|Y_i(\bs{z}^\prime)   E_{ji}Y_j^\prime(\bs{z}^\prime)-Y_i(\bs{z})   E_{ji}Y_j^\prime(\bs{z}) \Big| \\
			&+ \frac{2}{N^2}\sum_j   \Big|Y_k(\bs{z}^\prime)   E_{jk}Y_j^\prime(\bs{z}^\prime) -Y_k(\bs{z})   E_{jk}Y_j^\prime(\bs{z}) \Big| \\
			&+ \frac{2}{N^2}\sum_i   \Big|Y_i(\bs{z}^\prime)   E_{ki}Y_k^\prime(\bs{z}^\prime) -Y_i(\bs{z})   E_{ki}Y_k^\prime(\bs{z}) \Big|.
		\end{align*}
		
		Since both $\{\alpha_i,\theta_i,{\gamma}_{ij}\}_{1 \leq i,j\leq N}$ and $\{\alpha_i^\prime,\theta_i^\prime,{\gamma}_{ij}^\prime\}_{1 \leq i,j\leq N}$ satisfy \Cref{a:bounded-parameter},  we have  
		$\max_{\bs{z}\in \{0,1\}^N}|Y_i(\bs{z})|=O(1)$, $\max_{\bs{z}\in \{0,1\}^N}|Y_i^\prime(\bs{z})|=O(1)$, and for $k\not\in\{i,j\}$,
		\[
		|Y_i(\bs{z}^\prime)   E_{ji}Y_j^\prime (\bs{z}^\prime)-Y_i(\bs{z})   E_{ji}Y_j^\prime(\bs{z})| \lesssim \bar{E}_{ji}(Q_{ik}+Q_{jk}).
		\]
		It follows that
		\begin{align*}
			\max_{\bs{z}_{(-k)}\in\{0,1\}^{N-1}}|f(\bs{z})-f(\bs{z}^\prime)|\lesssim \sum_{i,j} \bar{E}_{ji}(Q_{ik}+Q_{jk}) + \sum_i  E_{ki} + \sum_j  E_{jk},
		\end{align*}
		and
		\begin{align*}
			\Var(S_4) \lesssim & N^{-4} \sum_{k}\Big(\sum_i   E_{ki}\Big)^2 + N^{-4}\sum_{k}\Big(\sum_j   E_{jk} \Big)^2 +N^{-4} \sum_{k}\Big(\sum_i  \sum_j  E_{ji}Q_{ik}\Big)^2 \\
			&+ N^{-4} \sum_{k}\Big(\sum_i  \sum_j  E_{ji}Q_{jk}\Big)^2\\
			\lesssim & N^{-4} \bs{1}\bs{E}^\top\bs{E}\bs{1} + N^{-4} \bs{1}\bs{E}\bs{E}^\top\bs{1} + N^{-4}\bs{1}\bs{E}\bs{Q}\bs{Q}^\top\bs{E}^\top\bs{1} + N^{-4}\bs{1}\bs{E}^\top\bs{Q}\bs{Q}^\top\bs{E}\bs{1} \\
			\lesssim & N^{-1}\rhon^2 = o(\rhon^2).
		\end{align*} 
		Therefore, $S_4 = \E S_4 + \op(\rhon)$. Similarly, we can prove that, for all $6\leq j\leq 7$, $S_j = \E S_j + \op(\rhon)$. Together with Theorem~\ref{thm:variance-estimator-dir}(ii) and Theorem~\ref{thm:consistency-of-oracle-variance-estimator-general-ind}, we have
		\[
		\hat{V}_{\tot}^\prime = \E \hat{V}_{\tot}^\prime  +  \op(N^{-1}\Lambdan +\rhon).
		\]
	\end{proof}

	\begin{proof}[Proof of \Cref{thm:variance-estimator-tot}]

		Theorem \ref{thm:variance-estimator-tot}(i) follows from  \Cref{thm:variance-tot-general} with $Y_i^\prime \equiv Y_i$. Theorem \ref{thm:variance-estimator-tot}(ii) follows from \Cref{thm:consistency-of-oracle-variance-estimator-general-tot} with $\Lambdan \equiv N^2\rhon^2$ and $Y_i^\prime \equiv Y_i$.

	\end{proof}

	\section{Proof for the asymptotic results of the eigenvector-adjusted estimator}\label{sec:E}
	
	\subsection{Preliminary results}
	
	\begin{lemma}
		\label{lem:consistency-of-beta-hat}
		Under Assumptions~\ref{a:bounded-parameter}, \ref{a:opnorm-EE^T}, and \ref{a:bounded-eigenvectors}, we have $\hat{\bs{\beta}}_{z}-\bs{\beta}_{z,\ora} = \Op(N^{-1/2})$, for $z=0,1$.
	\end{lemma}
	\begin{proof}
		By definition, the adjustment coefficient is
		\[
		\hat{\bs{\beta}}_z= \argmin_{\bs{\beta}} \sum_i  (Y_i-\bs{\beta}^\top \bs{W}_i)^2I(Z_i=z) = \Big(\sum_{i:Z_i=1} \bs{W}_i\bs{W}_i^\top \Big)^{-1}\sum_{i:Z_i=1} \bs{W}_iY_i.
		\]
		Simple calculation yields 
		\begin{align*}
			& \frac{1}{Nr_1}\sum_{i:Z_i=1} \bs{W}_i\bs{W}_i^\top - \frac{1}{N}\sum_i  \bs{W}_i\bs{W}_i^\top = \frac{1}{Nr_1}\sum_i  (Z_i-r_1) \bs{W}_i\bs{W}_i^\top 
		\end{align*}
		and
		\begin{align*}
			& \frac{1}{Nr_1}\sum_{i:Z_i=1} \bs{W}_iY_i- \frac{1}{N}\sum_i  \bs{W}_iY_{Z_i=1} \\
			=& \frac{1}{Nr_1}\sum_i  Z_i\bs{W}_iY_i- \frac{1}{N}\sum_i  \bs{W}_iY_{Z_i=1}\\
			=&\frac{1}{Nr_1} \sum_i  (Z_i-r_1)\bs{W}_i(\alpha_i+\theta_i+r_1h_i) + \frac{1}{N} \sum_i  (Z_i-r_1)\bs{W}_j \Big(\sum_j  \tilde{\gamma}_{ji}\Big) + \\
			& \frac{1}{Nr_1} \sum_{i,j} (Z_i-r_1)(Z_j-r_1)\bs{W}_i\tilde{\gamma}_{ij}.  
		\end{align*}
		By \Cref{a:bounded-eigenvectors} and $N^{-1}\sum_i  \bs{W}_i\bs{W}_i^\top = \bs{I}_{K+1}$, we have
		\[
		(Nr_1)^{-1}\sum_{i:Z_i=1} \bs{W}_i\bs{W}_i^\top = \bs{I}_{K+1} + \Op(N^{-1/2}).
		\]
		Let ${W}_{ik}$ be the $k$th element of $\bs{W}_i$. By \Cref{a:bounded-parameter} and \Cref{a:bounded-eigenvectors}, we can verify that
		\begin{align*}
			&\frac{1}{Nr_1} \sum_i  (Z_i-r_1){W}_{ik}(\alpha_i+\theta_i+r_1h_i) = \op(N^{-1/2}), \\
			&\frac{1}{Nr_1} \sum_{i,j} (Z_i-r_1)(Z_j-r_1){W}_{ik}\tilde{\gamma}_{ij} = \op(N^{-1/2}),
		\end{align*}
		\begin{align*}
			\Var\Big\{  \frac{1}{N} \sum_i  (Z_i-r_1){W}_{jk} \Big(\sum_j  \tilde{\gamma}_{ji}\Big) \Big\} \lesssim \frac{1}{N^2} \sum_i  \Big(\sum_j  Q_{ji}\Big)^2 = \frac{1}{N^2} \bs{1}^\top \bs{Q}\bs{Q}^\top \bs{1} = O(N^{-1}).
		\end{align*}
		As a consequence, we have
		\[
		\frac{1}{Nr_1}\sum_{i:Z_i=1} \bs{W}_iY_i- \frac{1}{N}\sum_i  \bs{W}_iY_{Z_i=1} =\Op(N^{-1/2}).
		\]
		Putting together the pieces, we have $(\hat{\bs{\beta}}_{1}-\bs{\beta}_{1,\ora}) = \Op(N^{-1/2})$. Similarly, we have $(\hat{\bs{\beta}}_{0}-\bs{\beta}_{0,\ora}) = \Op(N^{-1/2})$.
	\end{proof}

	\begin{theorem}
		\label{thm:CLT-EV-adj-tdl}
		Under Assumptions~\ref{a:bounded-parameter}--\ref{a:Lindberg-condition-unadj} and \ref{a:bounded-parameter-adjusted}--\ref{a:lind-berg-condition-for-adjusted-estimator}, we have $(\tilde{\tau}_{\ind}^{\ev}-\tau_{\ind})/\Var(\tilde{\tau}_{\ind}^{\ev})^{1/2}\xrightarrow{d}\mathcal{N}(0,1)$ and $(\tilde{\tau}_{\tot}^{\ev}-\tau_{\tot})/\Var(\tilde{\tau}_{\tot}^{\ev})^{1/2}\xrightarrow{d}\mathcal{N}(0,1)$.
	\end{theorem}
	
	\begin{proof}
		The asymptotic normality of $\tilde{\tau}^{\ev}_{\tot}$ and $\tilde{\tau}^{\ev}_{\ind}$ follows from \Cref{thm:CLT-general} with $\Lambdan \equiv \Deltan$ and $Y_i^\prime  \equiv e_i$.
	\end{proof}

	\subsection{Proof of \Cref{thm:CLT-EV-adj}}    
	
	\begin{proof}
		Simple calculation yields
		\begin{align*}
			e_i-\hat{e}_i =  Z_i(\hat{\bs{\beta}}_{1}-\bs{\beta}_{1,\ora})^\top \bs{W}_i + (1-Z_i)(\hat{\bs{\beta}}_{0}-\bs{\beta}_{0,\ora})^\top \bs{W}_i.
		\end{align*}
		Then, we have
		\begin{align*}
			\tilde{\tau}_{\ind}^{\ev}-\hat{\tau}_{\ind}^{\ev} =& \frac{1}{N}\sum_{i,j} E_{ij} \Bigl\{\frac{(\hat{\bs{\beta}}_{1}-\bs{\beta}_{1,\ora})^\top \bs{W}_j Z_j}{r_1} - \frac{(\hat{\bs{\beta}}_{1}-\bs{\beta}_{0,\ora})^\top \bs{W}_j(1-Z_j)}{r_0}\Bigr\}\\
			=& \Op(N^{-1/2})\Op(N^{1/2}\rhon) = \Op(\rhon) = \op(\rhon^{1/2}),
		\end{align*}
		where the second equality is due to \Cref{lem:consistency-of-beta-hat} and the last equality is due to $\rhon \rightarrow 0$ by \Cref{a:density-rho_N}. 
		
		By \Cref{a:control-the-order-of-two-component}(i),  \Cref{thm:CLT-EV-adj-tdl}, and Slutsky's Theorem, we have
		$(\hat{\tau}_{\ind}^{\ev}-\tau_{\ind})/\Var(\hat{\tau}_{\ind}^{\ev})^{1/2}$ $\xrightarrow{d}\mathcal{N}(0,1)$. The proof for the asymptotic normality of $\hat{\tau}_{\tot}^{\ev}$ is similar, so we omit it.
		
	\end{proof}

	\subsection{Proof of \Cref{thm:asymptotic-variance-estimator-adjusted}}
	\begin{proof}
		{\bf Step 1.} We prove (i). Simple calculation yields
		\begin{align*}
			&N^{-2}\sum_i \biggl(\sum_j  E_{ji}\hat{e}_{j}\frac{Z_j}{r_1}\biggr)^2 \frac{Z_i}{r_1} \\
			=&  N^{-2}\sum_i \biggl\{\sum_j  E_{ji}({e}_{j}+\hat{e}_{j}-e_j)\frac{Z_j}{r_1}\biggr\}^2 \frac{Z_i}{r_1}\\
			=& N^{-2}\sum_i \biggl(\sum_j  E_{ji}{e}_{j}\frac{Z_j}{r_1}\biggr)^2 \frac{Z_i}{r_1} + N^{-2}\sum_i \biggl\{\sum_j  E_{ji}(\hat{e}_{j}-e_j)\frac{Z_j}{r_1}\biggr\}^2 \frac{Z_i}{r_1} +\\
			& N^{-2} \sum_i \biggl\{\sum_j  E_{ji}(\hat{e}_{j}-e_j)\frac{Z_j}{r_1}\biggr\}\biggl(\sum_j  E_{ji}{e}_{j}\frac{Z_j}{r_1}\biggr) \frac{Z_i}{r_1} \\
			=: & S_1+S_2+S_3.
		\end{align*}
		By \Cref{a:bounded-eigenvectors}, we have
		\begin{align*}
			S_2 =& \sum_i \biggl\{\sum_j  E_{ji}(\hat{\bs{\beta}}_{1}-\bs{\beta}_{1,\ora})^\top \bs{W}_j\frac{Z_j}{r_1}\biggr\}^2 \frac{Z_i}{N^2r_1}\\
			=&\sum_i  (\hat{\bs{\beta}}_{1}-\bs{\beta}_{1,\ora})^\top \biggl\{\sum_j  E_{ji}\bs{W}_j\frac{Z_j}{r_1}\biggr\}\biggl\{\sum_j  E_{ji}\bs{W}_j^\top\frac{Z_j}{r_1}\biggr\}(\hat{\bs{\beta}}_{1}-\bs{\beta}_{1,\ora}) \frac{Z_i}{N^2r_1}\\
			=& \Op(N^{-1})\Op(N\rhon^2) = \Op(\rhon^2) = \op(\rhon).
		\end{align*}
		Therefore, $S_2 = \op(N^{-1}\Deltan +\rhon)$. By the proof of \Cref{thm:consistency-of-oracle-variance-estimator-general-ind} with $Y_i^\prime \equiv e_i$ and  $\Lambdan\equiv \Deltan$, we have
		$$
		S_1 = \E S_1 + \op(N^{-1}\Deltan +\rhon).
		$$
		By Cauchy--Schwarz inequality, we have $S_3 \lesssim (S_1S_2)^{1/2} = \op\Big(N^{-1}\Deltan +\rhon\Big)$. Therefore,
		\[
		N^{-2}\sum_i \biggl(\sum_j  E_{ji}\hat{e}_{j}\frac{Z_j}{r_1}\biggr)^2 \frac{Z_i}{N^2r_1} = \E S_1 + \op(N^{-1}\Deltan +\rhon).
		\]

		Define $\tilde{V}_{\ind}^{\ev}$ by replacing $Y_i$ in $\hat{V}_{\ind}$ with the oracle residual $e_i$.     Mimicking the above proof for the other terms of $\hat{V}_{\ind}^{\ev}$, we can obtain that 
		\begin{equation}\label{eqn:vindev}
			\hat{V}_{\ind}^{\ev} = \tilde{V}_{\ind}^{\ev} + \op\Big(N^{-1}\Deltan +\rhon\Big) = \E\tilde{V}_{\ind}^{\ev}+ \op\Big(N^{-1}\Deltan +\rhon\Big),
		\end{equation}
		where the last equality is due to \Cref{thm:consistency-of-oracle-variance-estimator-general-ind} with $\Lambdan \equiv \Deltan$ and $Y_i^\prime  \equiv e_i$.
		Similarly, 
		$$\hat{V}_{\tot}^{\ev} = \tilde{V}_{\tot}^{\ev} + \op\Big(N^{-1}\Deltan +\rhon\Big) = \E\tilde{V}_{\tot}^{\ev}+ \op\Big(N^{-1}\Deltan +\rhon\Big).$$
		
		By the proof of \Cref{thm:variance-estimator-ind} with $Y_i$ replaced by $e_i$, we have $\E (2\tilde{V}^{\ev}_{\ind}) -\Var(\tilde{\tau}^{\ev}_{\ind}) \geq 0$ and 
		\begin{align}\label{eqn:vindev2}
			\E \tilde{V}^{\ev}_{\ind} &=  \Var(\tilde{\tau}^{\ev}_{\ind}) + \frac{1}{N^2} \sum_i  \Bigl(\sum_j   E_{ji}\theta_{\bsw,j}\Bigr)^2+O(\rhon).
		\end{align}
		Combined with \Cref{prop:order-of-adjusted-estimators}, we have
		$$
		\hat{V}_{\ind}^{\ev} = \Op\Big(N^{-1}\Deltan +\rhon\Big).
		$$
		Similarly, we have $\E (2\tilde{V}^{\ev}_{\tot}) -\Var(\tilde{\tau}^{\ev}_{\tot}) \geq 0$, 
		$$
		\E \tilde{V}^{\ev}_{\tot} =  \Var(\tilde{\tau}^{\ev}_{\tot}) + \frac{1}{N^2} \sum_i  \Bigl( \theta_i + \sum_j   E_{ji}\theta_{\bsw,j}\Bigr)^2+O(\rhon),
		$$
		$$
		\hat{V}_{\tot}^{\ev} = \Op\Big(N^{-1}\Deltan +\rhon\Big).
		$$

		{\bf Step 2.} We prove (ii). By \eqref{eqn:vindev} and $\E (2\tilde{V}^{\ev}_{\ind}) -\Var(\tilde{\tau}^{\ev}_{\ind}) \geq 0$, we have
		\begin{align*}
			2 \hat{V}_{\ind}^{\ev} =&  2 \E \tilde{V}_{\ind}^{\ev} + \op\Big(N^{-1}\Deltan +\rhon\Big) \\
			=& \Var(\tilde{\tau}^{\ev}_{\ind})  + R_N^{(1)} + \op\Big(N^{-1}\Deltan +\rhon\Big), 
		\end{align*}
		where $R_N^{(1)} = \E (2\tilde{V}^{\ev}_{\ind}) -\Var(\tilde{\tau}^{\ev}_{\ind}) \geq 0$.
		
		Similarly, we have
		\begin{align*}
			2 \hat{V}_{\tot}^{\ev} =&  2 E \tilde{V}_{\tot}^{\ev} + \op\Big(N^{-1}\Deltan +\rhon\Big) \\
			=& \Var(\tilde{\tau}^{\ev}_{\tot})  + R_N^{(2)} + \op\Big(N^{-1}\Deltan +\rhon\Big), 
		\end{align*}
		where $R_N^{(2)} = \E (2\tilde{V}^{\ev}_{\tot}) -\Var(\tilde{\tau}^{\ev}_{\tot}) \geq 0$.
		
		{\bf Step 3.} We prove (iii). Combining \eqref{eqn:vindev} and \eqref{eqn:vindev2}, we have
		$$
		\hat{V}_{\ind}^{\ev} =  \Var(\tilde{\tau}^{\ev}_{\ind}) + \frac{1}{N^2} \sum_i  \Bigl(\sum_j   E_{ji}\theta_{\bsw,j}\Bigr)^2+O(\rhon) + \op\Big(N^{-1}\Deltan +\rhon\Big).
		$$
		If $\liminf N^{-1}\Deltan/\rhon \rightarrow \infty$, then, $O(\rhon) = o(N^{-1}\Deltan)$. Therefore, 
		$$
		\hat{V}_{\ind}^{\ev} =  \Var(\tilde{\tau}^{\ev}_{\ind}) + \frac{1}{N^2} \sum_i  \Bigl(\sum_j   E_{ji}\theta_{\bsw,j}\Bigr)^2 + \op\Big(N^{-1}\Deltan \Big).
		$$
		
		The proof for
		$$
		\tilde{V}^{\ev}_{\tot} =  \Var(\tilde{\tau}^{\ev}_{\tot}) + \frac{1}{N^2} \sum_i  \Bigl( \theta_i + \sum_j   E_{ji}\theta_{\bsw,j}\Bigr)^2+ \op\Big(N^{-1}\Deltan \Big) 
		$$
		is similar, so we omit it.
		
	\end{proof}

	\section{Proof for the results in \Cref{sec:application-in-some-networks}} \label{sec:F}

	\subsection{Proof of \Cref{cor:variance-estimator-adj-stratified}}
	
	\begin{proof}
		Applying \Cref{prop:order-and-formula-of-tau-ind} and \Cref{prop:order-of-tau-tot} with $\rhon = M\binom{N_M}{2}/\binom{N}{2}= O(M^{-1})$, the conclusion follows. 

	\end{proof}

	\subsection{Proof of \Cref{cor:variance-estimator-adj-stratified2}}
	Note that if we use $(\bs{V}_1,\ldots,\bs{V}_K)$ in the regression, where $\bs{V}_1,\ldots,\bs{V}_K$ are the top $K$ eigenvectors, the corresponding $\Deltan = O(\lambda_{K+1})$, where $\lambda_{K+1}$ is the $(K+1)$-th largest eigenvalues of $\bs{E}\bs{E}^\top$.
	
	Assume that we can further stratify the $M$ groups into $K$ strata: $\mathcal{M}_1\cup\cdots \cup \mathcal{M}_K = \{1,\ldots,M\}$, $\mathcal{M}_{k_1}\cap \mathcal{M}_{k_2} = \emptyset$ for $k_1 \ne k_2$. We assume that there is no between-group heterogeneity in each stratum; see \Cref{a:SM-no-between-group-heterogeneity} below.
	\begin{assumption}
		\label{a:SM-no-between-group-heterogeneity}
		If $m,m^\prime \in \mathcal{M}_k$ for $1\leq k \leq M$, we have $N_M^{-1} \sum_{i:C_i=m} Y_{Z_i=z} = N_M^{-1}$ $\sum_{i:C_i=m^\prime} Y_{Z_i=z}$, for $z=0,1$. Moreover, there exist constants $0 < \underline{c}\leq \bar{c}$, such that $ \underline{c} M \leq |\mathcal{M}_k| \leq \bar{c} M$.
	\end{assumption}
	
	Under \Cref{a:SM-no-between-group-heterogeneity}, we can merge the homogeneous groups into a stratum. Let ${\bs{W}}_i = (I(C_i\in \mathcal{M}_1),\ldots,I(C_i\in \mathcal{M}_K))$  be the vector of merged group indicators. We use $\hat{e}_i$ from the following regression in the eigenvector-adjusted IATE and GATE estimators:
	\begin{align}
		\label{eq:regression-in-merged-group-indicator}
		Y_{i}\sim  I(Z_{i}=0){\bs{W}}_i + I(Z_{i}=1){\bs{W}}_i.
	\end{align}
	
	Note that regression \eqref{eq:regression-reduced-group-indicator} in the main text under \Cref{a:no-heterogeneity-between-groups} is a special case of \eqref{eq:regression-in-merged-group-indicator} under \Cref{a:SM-no-between-group-heterogeneity}. With a slight abuse of notation, we still denote the resulting eigenvector-adjusted IATE and GATE estimators as $\hat{\tau}^{\ev}_{\ind}$ and $\hat{\tau}^{\ev}_{\tot}$, respectively. \Cref{cor:consistency-for-merged-group-regression} below is an extension of \Cref{cor:variance-estimator-adj-stratified2}. 
	
	\begin{corollary}
		\label{cor:consistency-for-merged-group-regression}
		Under Assumptions~\ref{a:bounded-parameter}--\ref{a:Lindberg-condition-unadj}, \ref{a:bounded-parameter-adjusted}--\ref{a:stratified-interference}, and \ref{a:SM-no-between-group-heterogeneity}, we have $\hat{\tau}_{\ind}^{\ev} - \tau_{\ind} = \Op(M^{-1/2})$ and $\hat{\tau}_{\tot}^{\ev} - \tau_{\tot} =  \Op(M^{-1/2})$.
	\end{corollary}
	
	\begin{proof}\
		Let $\tilde{\bs{W}}_i \equiv (I(C_i=1),\ldots,I(C_i=M))$. We denote $\tilde{e}_{Z_i=z}$ and ${e}_{Z_i=z}$ as the projection residual of $Y_{Z_i=z}$ on $\tilde{\bs{W}}_i$ and ${\bs{W}}_i$, respectively.
		Let $\mathcal{M}(i)$ denote the stratum unit $i$ belongs to, i.e., $C_i \in \mathcal{M}(i)$. Then,
		\begin{align*}
			&\tilde{e}_{Z_i=z} = Y_{Z_i=z} - N_M^{-1}\sum_{j:C_j=C_i} Y_{Z_j=z},\\
			&{e}_{Z_i=z} = Y_{Z_i=z} - (|\mathcal{M}(i)|N_M)^{-1}\sum_{j:C_j\in \mathcal{M}(i)} Y_{Z_j=z}.
		\end{align*}
		By \Cref{a:SM-no-between-group-heterogeneity}, we have $\tilde{e}_{Z_i=z} = {e}_{Z_i=z}$. We denote $\Deltan$ and $\tilde{\Delta}_{N}$ as the quantities corresponding to $\bs{W}_i$ and $\tilde{\bs{W}}_i$, respectively.
		Then, we have
		\begin{align*}
			\Deltan = \tldDeltan.
		\end{align*}
		Since $\tilde{\bs{W}}_i$ is the top $M$ eigenvectors of $\bs{E}=\diag(\{\bs{J}_{N_M}-\bs{I}_{N_M}\}_{m=1}^M)$ and $\lambda_{M+1}(\bs{E}) = -1$, it follows that $\tldDeltan = O(1)$. Using $\rhon = M\binom{N_M}{2}/\binom{N}{2}= O(M^{-1})$, $N \geq M$, we have 
	
		\[
		\Op(\rhon + \tldDeltan/N) = \Op(M^{-1}).
		\]
		The conclusion follows.
	\end{proof}
	
	\subsection{Proof of \Cref{cor:local-interference-for-marketplace}}
	\begin{proof}
		Applying \Cref{prop:order-and-formula-of-tau-ind} and \Cref{prop:order-of-tau-tot} with $\rhon = \big\{N_R\binom{N_C}{2}+ N_R\binom{N_C}{2}\big\}/N^2 = O(N_R^{-1}+N_C^{-1})$, the conclusion follows. 
	
	\end{proof}

	\subsection{Proof of \Cref{cor:local-interference-for-marketplace-adj}}
	
	Assume that we can group the rows and columns into $K_1$ and $K_2$ groups, such that \Cref{a:SM-no-between-group-heterogeneity-row-and-column} below holds: $\mathcal{R}_1\cup \cdots \cup \mathcal{R}_{K_1} = \{1,\ldots,N_R\}$ and $\mathcal{C}_1\cup \cdots \cup \mathcal{C}_{K_2} = \{1,\ldots,N_C\}$,  $\mathcal{R}_{k}\cap \mathcal{R}_{k^\prime} = \emptyset$ for $k \ne k^\prime$ and $\mathcal{C}_{k}\cap \mathcal{C}_{k^\prime} = \emptyset$ for $k \ne k^\prime$.
	\begin{assumption}
		\label{a:SM-no-between-group-heterogeneity-row-and-column}
		If $r,r^\prime \in \mathcal{R}_k$ for $1\leq k \leq K_1$, we have $N_C^{-1} \sum_{i:R_i=r} Y_{Z_i=z} = N_C^{-1}\sum_{i:R_i=r^\prime}$ $ Y_{Z_i=z}$, $z=0,1$. If $c,c^\prime \in \mathcal{C}_k$ for $1\leq k \leq K_2$, we have $N_R^{-1} \sum_{i:C_i=c} Y_{Z_i=z} = N_R^{-1}\sum_{i:C_i=c^\prime} Y_{Z_i=z}$, $z=0,1$. Moreover, there exist constants $0 < \underline{c}\leq \bar{c}$, such that $ \underline{c} N_{R} \leq |\mathcal{R}_k| \leq \bar{c} N_R$ for $1 \leq k \leq K_1$ and $  \underline{c} N_{C} \leq |\mathcal{C}_k| \leq \bar{c} N_C$ for $1 \leq k \leq K_2$.
	\end{assumption}
	\Cref{a:SM-no-between-group-heterogeneity-row-and-column} assumes seller homogeneity and buyer homogeneity within groups, which is a generalization of \Cref{a:local-interference-for-marketplace}.  With a slight abuse of notation, let $\tilde{\bs{W}}_i \equiv ( I(R_i=1),\ldots,I(R_i=N_R),I(C_i=1),\ldots,I(C_i=N_C))$ and  $\bs{W}_i\equiv (I(R_i\in\mathcal{R}_{1}),\ldots,I(R_i\in\mathcal{R}_{K_1}),I(C_i\in \mathcal{C}_1),\ldots,I(C_i\in \mathcal{C}_{K_2}))$. We denote $\tilde{e}_{Z_i=z}$ and ${e}_{Z_i=z}$ as the residuals of the linear regression of $Y_{Z_i=z}$ on $\tilde{\bs{W}}_i$ and ${\bs{W}}_i$, respectively.
	
	Under \Cref{a:SM-no-between-group-heterogeneity-row-and-column}, we extend regression \eqref{eq:regression-reduced-group-indicator} into the following regression:
	\begin{align}
		\label{eq:regression-in-two-way-merged-group-indicator}
		Y_{i}\sim  I(Z_{i}=0){\bs{W}}_i + I(Z_{i}=1){\bs{W}}_i. 
	\end{align}
	The above regression is an interacted two-way effect regression, which adjusts the effect of homogeneous buyer and seller groups under both treatment and control groups. Although ${\bs{W}}_i$ is rank deficient, the regression residual $\hat{e}_i$ is well-defined as the residual of $Y_i$ projected into the column space of $(I(Z_i=1){\bs{W}}_i,I(Z_i=0){\bs{W}}_i)$. When $K_1=K_2 =1$, the residual obtained by \Cref{eq:regression-in-two-way-merged-group-indicator} is equal to that of \Cref{eq:regression-reduced-group-indicator} in the main text. Hence, \Cref{cor:local-interference-for-marketplace-adj-two-way-effect-adjustment} below is a generalization of \Cref{cor:local-interference-for-marketplace-adj}. With a slight abuse of notation, we still denote the regression \eqref{eq:regression-in-two-way-merged-group-indicator} adjusted IATE and GATE estimators as $\hat{\tau}^{\ev}_{\ind}$ and $\hat{\tau}^{\ev}_{\tot}$, respectively.
	
	\begin{corollary}
		\label{cor:local-interference-for-marketplace-adj-two-way-effect-adjustment}
		Under Assumptions~\ref{a:bounded-parameter}--\ref{a:Lindberg-condition-unadj}, \ref{a:bounded-parameter-adjusted}--\ref{a:bounded-eigenvectors}, \ref{a:local-interference-for-marketplace}, and \ref{a:SM-no-between-group-heterogeneity-row-and-column}, we have $\hat{\tau}_{\ind}^{\ev} - \tau_{\ind} = \Op(N_R^{-1/2}+N_C^{-1/2})$ and $\hat{\tau}_{\tot}^{\ev} - \tau_{\tot} = \Op(N_R^{-1/2}+N_C^{-1/2})$. 
	\end{corollary}
	
	\begin{proof}
		Denote $Y_{z,\mathcal{R}_k} =( N_C|\mathcal{R}_k|)^{-1}\sum_{i:R_i \in \mathcal{R}_k} Y_{Z_i=z}$ for $1 \leq k \leq K_1$ and $Y_{z,\mathcal{C}_k} =( N_C|\mathcal{C}_k|)^{-1}$ $\sum_{i:R_i \in \mathcal{C}_k} Y_{Z_i=z}$ for $1 \leq k \leq K_2$. Denote $Y_{z,r} =N_C^{-1}\sum_{i:R_i = r} Y_{Z_i=z}$ for $1 \leq r \leq N_R$ and $Y_{z,c} = N_R^{-1}\sum_{i:C_i =c} Y_{Z_i=z}$ for $1 \leq c \leq N_C$. Let $Y_{z} = N^{-1}\sum_{i} Y_{Z_i=z}$. Let $\bs{U}_i = (I(R_i \in \mathcal{R}_1),\ldots,I(R_i \in \mathcal{R}_{K_1}))^\top$ and $\bs{V}_i = (I(C_i \in \mathcal{C}_1),\ldots,I(C_i \in \mathcal{C}_{K_2}))^\top$. Let $\mathcal{R}(i)$ and $\mathcal{C}(i)$ denote the seller group and buyer group unit $i$ belongs to, respectively, i.e., $C_i \in \mathcal{C}(i)$ and $R_i \in \mathcal{R}(i)$. 
		
		The projection residual of $Y_{Z_i=z}$ on $\bs{U}_i$ is equal to $Y_{Z_i=z}-Y_{z,\mathcal{R}(i)}$. The projection residual of $\bs{V}_i$ on $\bs{U}_i$ is equal to $\bs{V}_{i,\backslash \bs{U}} = (I(C_i \in \mathcal{C}_1)-|\mathcal{C}_1|/N_C,\ldots,I(C_i \in \mathcal{C}_{K_2})-|\mathcal{C}_{K_2}|/N_C)^\top$.
		Hence, ${e}_{Z_i=z}$, the residual of $Y_{Z_i=z}$ on $(\bs{U}_i^\top,\bs{V}_i^\top)$ is equal to the residual of  $Y_{Z_i=z}-Y_{z,\mathcal{R}(i)}$ on $\bs{V}_{i,\backslash \bs{U}}$, and equal to that of $Y_{Z_i=z}-Y_{z,\mathcal{R}(i)}$ on $\bs{V}_i$, which is equal to $Y_{Z_i=z}-Y_{z,\mathcal{R}(i)}-Y_{z,\mathcal{C}(i)}+Y_{z}$. Similarly, 
	$\tilde{e}_i = Y_{Z_i=z}-Y_{z,R_i}-Y_{z,C_i}+Y_{z}$. Under \Cref{a:SM-no-between-group-heterogeneity-row-and-column}, we have $\tilde{e}_{Z_i=z}={e}_{Z_i=z}$. 
		
		We denote by $\Deltan$ and $\tilde{\Delta}_{N}$ the quantities corresponding to $\bs{W}_i$ and $\tilde{\bs{W}}_i$, respectively. It is not difficult to show that
		\begin{align}
			\label{eq:eigen-decomposition-of-SMR}
			\bs{E} = \sum_{r=1}^{N_R}(\bs{e}_{r,N_R}\otimes\bs{1}_{N_C})^\top (\bs{e}_{r,N_R}\otimes\bs{1}_{N_C})+ \sum_{c=1}^{N_C} (\bs{1}_{N_R}\otimes \bs{e}_{c,N_C})^\top(\bs{1}_{N_R}\otimes \bs{e}_{c,N_C}) -2\bs{I}_N.
		\end{align}
		Therefore, the dominant eigenvectors of $\bs{E}$ are in the space spanned by $\{\bs{e}_{r,N_R}\otimes\bs{1}_{N_C}\}_{r=1}^{N_R}\cup\{\bs{1}_{N_R}\otimes \bs{e}_{c,N_C}\}_{c=1}^{N_c}$, which is $\tilde{\bs{W}}_i$. The eigenvalues of the eigenvectors orthogonal to $\tilde{\bs{W}}_i$ are $-2$. As a consequence, following the argument of Remark \ref{rmk:estimate-of-deltan-eigenvalue},
		$$
		\tilde{\Delta}_{N} = O(1),
		$$
		which yields
		\[
		\Deltan = \tldDeltan = O(1).
		\]
		
		On the other hand, we have $\rhon =\big\{N_R\binom{N_C}{2}+ N_R\binom{N_C}{2}\big\}/N^2 = O(N_R^{-1}+N_C^{-1})$. Therefore, $\Op(\rhon + N^{-1}\Deltan) = \Op(N_R^{-1}+N_C^{-1})$. The conclusion follows.
	\end{proof}
	
	\subsection{Proof of \Cref{prop:graphon-model}}
	
	\begin{proof}
		Let $\tilde{\lambda}_k$, $k=1,\ldots,N$, be the eigenvalues of $\bs{E}$ in a decreasing order. Let $\bs{\psi}_k = (\psi_k(U_i))_{i=1}^n$, $k=1,\ldots, r$, and define
		\begin{align}
			\label{eq:definition-of-G}
			\bs{G} := ( G(U_i,U_j))_{1\leq i,j \leq N} = \sum_{k=1}^r \lambda_k \bs{\psi}_k \bs{\psi}_k^\top.
		\end{align}
		Note that $\bs{G}_{ii} = \sum_{k=1}^r \lambda_k \psi_k(U_i)\psi_k(U_i)$ and is therefore not necessarily $0$. Let $\lambda^\star_1\geq \ldots\geq \lambda^\star_{K}\geq  0 \geq \lambda^\star_{K+1}\geq \ldots\geq \lambda^\star_{N} $ be the eigenvalues of $(\rho^\star_N)\bs{G}$, where $K$ is the number of positive eigenvalues of $\bs{G}$. Since, by \Cref{eq:definition-of-G}, $\operatorname{rank}(\bs{G}) = r$, there are at most $r$ eigenvalues in $\lambda^\star_{i}$ that are not equal to $0$.  By Weyl's inequality, we have, for $i=1,\ldots,N$
		\[
		|\tilde{\lambda}_{i}-{\lambda}_{i}^\star| \leq \|(\rho^\star_N)\bs{G}-\bs{E}\|_{\oprtnorm} = \Op ((\log N)^2 \sqrt{N\rhon^\star}),
		\]
		where the last equality follows from Lemma 25 of \cite{li2022random}. Let $|\tilde{\lambda}_{(r+1)}| $ be the $(r+1)$-th largest eigenvalue of $
		|\tilde{\lambda}_{k}|$, $k=1,\ldots,N$. Then, we have 
		\begin{align*}
			|\tilde{\lambda}_{(r+1)}| =& \max_{\mathcal{K}\in [N]:|\mathcal{K}|=r+1}\min_{k\in \mathcal{K}} |\tilde{\lambda}_{k}| \\
			\leq & \max_{\mathcal{K}\in [N]:|\mathcal{K}|=r+1}\min_{k\in \mathcal{K}} |{\lambda}_{k}^\star| + \|(\rho^\star_N)\bs{G}-\bs{E}\|_{\oprtnorm} \\
			=& 0 + \Op ((\log N)^2 \sqrt{N\rhon^\star}) =  \Op ((\log N)^2 \sqrt{N\rhon^\star}).
		\end{align*}

		Note that $|\tilde{\lambda}_{(r+1)}|^2=\Op ((\log N)^4 N\rhon^\star)$ is the $(r+1)$-th largest eigenvalues of $\bs{E}^2$, thus, $N^{-1}\Deltan = O(N^{-1} \tilde{\lambda}_{(r+1)}^2)=\Op ((\log N)^4 \rhon^\star) = \op ((\log N)^{4+\delta} \rhon^\star), $ for any $\delta > 0$. 
		
		By definition, $\rhon = \sum_{i,j} E_{ij}/N^2$. Then, $N^2 \E\rhon = N(N-1) \E (E_{ij})$.  On the other hand, we have
		\begin{align*}
			\Var(\rhon) = N^{-4} \sum_{i,j} \Cov(E_{ij}) + N^{-4} \sum_{i,j,k} \Cov(E_{ij},E_{ik}).
		\end{align*}
		By Lemma 16 of \cite{li2022random}, we have $\E (E_{ij})/\{\rhon^\star\E G(U_1,U_2)\} \rightarrow 1$ and $\E E_{ij}E_{ik} /$ $\{(\rhon^\star)^2\E G(U_1,U_2) G(U_2,U_3)\} \rightarrow 1 $. As a consequence, 
		\begin{align*}
			&\Cov(E_{ij}) = \E E_{ij} -(\E E_{ij})^2 = O(\E E_{ij}) = O(\rhon^\star),\\
			&\Cov(E_{ij},E_{ik}) = \E E_{ij}E_{ik} - (\E E_{ij})^2 = O((\rhon^\star)^2),
		\end{align*}
		which yield that $\rhon/\rhon^\star = \{\E\rhon + \Op(N^{-1}(\rhon^\star)^{1/2}+N^{-1/2}\rhon^\star)\}/\rhon^\star = \E G(U_1,U_2)+\op(1),$ where the last equality follows from \Cref{a:regularity-conditons-for-graphon-model} (iv).
	\end{proof}

\end{document}